\documentclass[envcountsame,runningheads,a4paper]{llncs}

\usepackage{microtype}
\usepackage{mathtools,amsthm}
\usepackage[hidelinks,bookmarks,bookmarksopen,bookmarksdepth=3]{hyperref}
\hypersetup{colorlinks=false}

\usepackage[utf8]{inputenc}

\usepackage{amssymb}
\usepackage{stmaryrd}
\usepackage{bussproofs}
\usepackage{mathtools}
\usepackage{bbold}

\usepackage{ifthen}
\usepackage{xspace}
\usepackage{fancybox}

\newcommand{\ignore}[1]{}

\newcommand{\myinput}[1]{\ifthenelse{\boolean{withimages}}{\input{#1}}{}}

\newcommand{\reflemma}[1]{Lemma~\ref{l:#1}}
\newcommand{\reflemmap}[2]{Lemma~\ref{l:#1}.\ref{p:#1-#2}}

\newcommand{\reflemmasps}[3]{Lemmas~\ref{l:#1}.\ref{p:#1-#2}-\ref{p:#1-#3}} %Giulio

\newcommand{\reflemmaeq}[1]{{L.\ref{l:#1}}}

\newcommand{\refpoint}[1]{Point~\ref{p:#1}}

\newcommand{\refthm}[1]{Thm.~\ref{thm:#1}}
\newcommand{\refthmp}[2]{Thm.~\ref{thm:#1}.\ref{p:#1-#2}}
\newcommand{\refthmsps}[3]{Thm.~\ref{thm:#1}.\ref{p:#1-#2}-\ref{p:#1-#3}}

\newcommand{\refprop}[1]{Prop.~\ref{prop:#1}}
\newcommand{\refprops}[2]{Prop.~\ref{prop:#1}-\ref{prop:#2}} %Giulio
\newcommand{\refpropp}[2]{Prop.~\ref{prop:#1}.\ref{p:#1-#2}} %Giulio
\newcommand{\refpropsps}[3]{Propositions~\ref{prop:#1}.\ref{p:#1-#2}-\ref{p:#1-#3}} %Giulio
\newcommand{\refsect}[1]{Sect.~\ref{sect:#1}}

\newcommand{\refapp}[1]{Appendix~\ref{app:#1} (p.~\pageref{app:#1})}

\renewcommand{\refeq}[1]{(\ref{eq:#1})} %renewcommand to avoid conflict with package mathtools
\newcommand{\reffig}[1]{Fig.~\ref{fig:#1}}
\newcommand{\refcoro}[1]{Cor.~\ref{coro:#1}}
\newcommand{\refcor}[1]{Cor.\,\ref{coro:#1}}
\newcommand{\refdef}[1]{Definition~\ref{def:#1}}

\newcommand{\refrmk}[1]{Remark~\ref{rmk:#1}} %Giulio
\newcommand{\refrmkp}[2]{Remark~\ref{rmk:#1}.\ref{p:#1-#2}} %Giulio
\newcommand{\refrmksps}[3]{Remarks~\ref{rmk:#1}.\ref{p:#1-#2}-\ref{p:#1-#3}} %Giulio

\newcommand{\ie}{\textit{i.e.}\xspace}
\newcommand{\eg}{\textit{e.g.}\xspace}
\newcommand{\ih}{\textit{i.h.}\xspace}

\newcommand{\defeq}{\coloneqq} %Giulio
\newcommand{\eqdef}{\eqqcolon} %Giulio
\newcommand{\grameq}{\Coloneqq} %Giulio
\newcommand{\set}[1]{\{#1\}}

\newcommand{\nat}{\mathbb{N}}
\newcommand{\size}[1]{|#1|}

\newcommand{\db}{{\tt dB}}

\newcommand{\vsym}{\mathsf{v}}

\renewcommand{\l}{\lambda}
\newcommand{\isub}[2]{\{#1/#2\}}
\renewcommand{\isub}[2]{\{#1{\shortleftarrow}#2\}}
\newcommand{\esub}[2]{[#1/#2]}
\renewcommand{\esub}[2]{[#1{\shortleftarrow}#2]}
\newcommand{\fv}[1]{{\tt fv}(#1)}

\newcommand{\mnf}[1]{{\msym}(#1)} %Giulio
\newcommand{\enf}[1]{{\esym}(#1)} %Giulio
\newcommand{\shufnf}{w} %Giulio

\newcommand{\rootRew}[1]{\mapsto_{#1}}
\newcommand{\Rew}[1]{\rightarrow_{#1}}
\newcommand{\Rewn}[1]{\rightarrow_{#1}^*}
\newcommand{\lRew}[1]{\; \mbox{}_{#1}{\leftarrow}\ }
\newcommand{\lRewn}[1]{\; \mbox{}^{*}_{#1}{\leftarrow}\ }
\newcommand{\lRewp}[1]{\; \mbox{}^{+}_{#1}{\leftarrow}\ }
\newcommand{\mult}{\mathsf{m}} %Giulio
\newcommand{\expo}{\mathsf{e}} %Giulio
\newcommand{\expoabs}{\expo_{\abssym}} %Giulio
\newcommand{\expovar}{\expo_{\varsym}} %Giulio
\newcommand{\rtom}{\rootRew{\mult}} %Giulio
\newcommand{\rtoe}{\rootRew{\expo}} %Giulio
\newcommand{\rtoeabs}{\rootRew{\expoabs}} 
\newcommand{\rtoevar}{\rootRew{\expovar}} 
\newcommand{\rtobv}{\rootRew{\betav}} %Giulio
\newcommand{\rtobabs}{\rootRew{\betaabs}} 
\newcommand{\rtobvar}{\rootRew{\betavar}} 
\newcommand{\rtosigma}{\rootRew{\sigma}} %Giulio
\newcommand{\rtoin}{\rootRew{\inert}}

\newcommand{\slsym}{\sigma_1}
\newcommand{\srsym}{\sigma_3}

\newcommand{\rtosl}{\rootRew{\slsym}}
\newcommand{\rtosr}{\rootRew{\srsym}}

\newcommand{\tob}{\Rew{\beta}}
\newcommand{\betav}{{\beta_v}} %Giulio
\newcommand{\abssym}{\lambda} 
\newcommand{\varsym}{y} 
\newcommand{\betaabs}{{\beta_{\!\abssym}}} 
\newcommand{\betavar}{{\beta_{\varsym}}}
\newcommand{\betavm}{{\bilancia\beta_v}} %Giulio
\newcommand{\betashuf}{\betavm} %Giulio
\newcommand{\tobv}{\Rew{\betav}} %Giulio
\newcommand{\tobabs}{\Rew{\betaabs}} 
\newcommand{\tobvar}{\Rew{\betavar}} 
\newcommand{\tobvm}{\Rew{\betavm}} %Giulio
\newcommand{\betain}{\beta_{\isym}} %Giulio
\newcommand{\inert}{\betain} %Giulio
\newcommand{\toin}{\Rew{\inert}}

\newcommand{\tosig}{\Rew{\sigma}} %Giulio$
\newcommand{\sigm}{\bilancia{\sigma}}

\newcommand{\tosigm}{\Rew{\sigm}} %Giulio
\newcommand{\sigl}{\bilancia{\sigma}_1}
\newcommand{\sigr}{\bilancia{\sigma}_3}
\newcommand{\tosl}{\Rew{\sigl}}
\newcommand{\tosr}{\Rew{\sigr}}

\newcommand{\toperm}{\Rew{\vmsym}} %Giulio
\newcommand{\toshuf}{\tovm} %Giulio
\newcommand{\tovm}{\Rew{\vmsym}} %Giulio
\newcommand{\esym}{{\mathtt e}}
\newcommand{\isym}{i}

\newcommand{\msym}{{\mathtt m}}
\newcommand{\fsym}{f}

\newcommand{\subsym}{{\mathsf{sub}}}

\newcommand{\ssym}{{\mathtt s}}
\newcommand{\wsym}{{\mathtt w}} %Giulio
\newcommand{\bilancia}[1]{{#1}^{\flat\!}}
\newcommand{\vmsym}{\mathsf{shuf}} 
\newcommand{\shuf}{\vmsym} %Giulio

\newcommand{\shufeqext}{\shufeqext} %Giulio
\newcommand{\tom}{\Rew{\mult}}
\newcommand{\toe}{\Rew{\expo}}
\newcommand{\toeabs}{\Rew{\expoabs}}
\newcommand{\toevar}{\Rew{\expovar}}

\newcommand{\eqstruct}{\equiv}
\newcommand{\tostruct}{\eqstruct}

\newcommand{\aplsym}{@\textup{l}}
\newcommand{\aprsym}{@\textup{r}}
\newcommand{\essym}{[\cdot]}
\newcommand{\comsym}{\textup{com}}
\newcommand{\tostructapl}{\tostruct_{\aplsym}}
\newcommand{\tostructapr}{\tostruct_{\aprsym}}
\newcommand{\tostructes}{\tostruct_{\essym}}
\newcommand{\tostructcom}{\tostruct_{\comsym}}

\newcommand{\seqbar}{\bumpeq}
\newcommand{\seqbarp}{\seqbar'} %Giulio
\newcommand{\seqbarmut}{\seqbar_{\mut\mut}} %Giulio 

\newcommand{\tm}{t}
\newcommand{\tmtwo}{u}
\newcommand{\tmthree}{s}
\newcommand{\tmfour}{r}
\newcommand{\tmfive}{q}
\newcommand{\tmsix}{p}

\newcommand{\cotm}{e}
\newcommand{\cotmtwo}{\cotm'}

\newcommand{\tmp}{\tm'}
\newcommand{\tmtwop}{\tmtwo'}
\newcommand{\tmthreep}{\tmthree'}
\newcommand{\tmfourp}{\tmfour'}
\newcommand{\tmfivep}{\tmfive'}

\newcommand{\var}{x}
\newcommand{\vartwo}{y}
\newcommand{\varthree}{z}

\newcommand{\covar}{\alpha}

\newcommand{\val}{v}
\newcommand{\valtwo}{\val'}

\newcommand{\ctxholep}[1]{\langle #1\rangle}
\newcommand{\ctxhole}{\ctxholep{\cdot}}

\newcommand{\sctx}{L}
\newcommand{\sctxtwo}{\sctx'}
\newcommand{\sctxthree}{\sctx''}
\newcommand{\sctxp}[1]{\sctx\ctxholep{#1}}
\newcommand{\sctxtwop}[1]{\sctxtwo\!\ctxholep{#1}}
\newcommand{\sctxthreep}[1]{\sctxthree\!\ctxholep{#1}}

\newcommand{\arbctxp}[1]{\arbctxp{#1}}
\newcommand{\arbctxtwop}[1]{\arbctxtwop{#1}}

\newcommand{\genevctx}{E}

\newcommand{\evctx}{\genevctx}

\newcommand{\evctxp}[1]{\evctx\ctxholep{#1}}

\newcommand{\cotctx}{D}
\newcommand{\cotctxtwo}{\cotctx'}
\newcommand{\cotctxthree}{\cotctx''}

\newcommand{\cotctxp}[1]{\cotctx\ctxholep{#1}}
\newcommand{\cotctxtwop}[1]{\cotctxtwo\ctxholep{#1}}
\newcommand{\cotctxthreep}[1]{\cotctxthree\ctxholep{#1}}

\newcommand{\env}{e}
\newcommand{\envtwo}{e'}

\newcommand{\stempty}{\epsilon}

\newcommand{\sizebabs}[1]{\size{#1}_{\betaabs}} 
\newcommand{\sizein}[1]{\size{#1}_{\betain}}
\newcommand{\sizebvm}[1]{\size{#1}_{\betavm}} %Giulio
\newcommand{\sizebshuf}[1]{\sizebvm{#1}} %Giulio
\newcommand{\sizevm}[1]{\size{#1}_{\vmsym}} %Giulio

\newcommand{\sizef}[1]{\size{#1}_{\betaf}}

\newcommand{\nes}[1]{|#1|_\textup{ES}} %number of explicit substitutions %Giulio

\newcommand{\deriv}{d}
\newcommand{\derivtwo}{e}
\newcommand{\derivthree}{f}
\newcommand{\derivfour}{g}

\newcommand{\sizehole}[2]{|#2|_{#1}}

\newcommand{\sizee}[1]{\sizehole{\expo}{#1}} %Giulio
\newcommand{\sizeeabs}[1]{\sizehole{\expoabs}{#1}} 
\newcommand{\sizeevar}[1]{\sizehole{\expovar}{#1}} 
\newcommand{\sizem}[1]{\sizehole{\mult}{#1}} %Giulio
\newcommand{\sizevsub}[1]{\sizehole{\vsub}{#1}} %Giulio
\newcommand{\sizevsubk}[1]{\sizehole{\vsubk}{#1}} %Giulio
\newcommand{\todb}{\Rew{\db}}

\newcommand{\const}{a}

 \newcommand{\gconst}{i} 
\newcommand{\gconsttwo}{\gconst'}
\newcommand{\gconstthree}{\gconst''}

\newcommand{\fire}{f}
\newcommand{\firetwo}{\fire'}
\newcommand{\firethree}{\fire''}

\newcommand{\tovs}{\Rew{\vsym\ssym}}
\newcommand{\betaf}{\beta_{\!\fsym}} %Giulio
\newcommand{\tof}{\Rew{\betaf}}

\newcommand{\vsub}{{\vsym\subsym}} %Giulio
\newcommand{\vsubeq}{{\vsub_\eqstruct}\!} %Giulio

\newcommand{\tovsub}{\Rew{\vsub}}
\newcommand{\tovsubk}{\Rew{\vsubk}} %Giulio

\newcommand{\tovsubeq}{\Rew{\vsubeq}}

\newcommand{\tow}{\Rew{\wsym}} %Giulio
\newcommand{\mctx}{B}

\newcommand{\mctxp}[1]{\mctx\ctxholep{#1}}

\newcommand{\unfsym}{\rotatebox[origin=c]{-90}{$\rightarrow$}}
\newcommand{\unf}[1]{#1\unfsym\,}

\newcommand{\la}[1]{\lambda #1.}
\newcommand{\tmseven}{m}

\newcommand{\myproof}[1]{
\ifthenelse{\boolean{omitproofs}}{\begin{IEEEproof} Proof available but omitted for readability. \end{IEEEproof}}{#1}}

\newcommand{\levy}{{L{\'e}vy}\xspace}
\newcommand{\gregoire}{Gr{\'{e}}goire\xspace}

\newcommand{\withproofs}[1]{\ifthenelse{\boolean{withproofs}}{#1}{}}

\newcommand{\withoutproofs}[1]{\ifthenelse{\boolean{withproofs}}{}{#1}}

\newcommand{\alphaeq}{=_\alpha}

\newcommand{\NoteProof}[1]{\withproofs{\marginpar{\scriptsize \ \ Proof p.\,{\pageref{#1}}}}} %Giulio
\newcommand{\NoteState}[1]{\withproofs{\marginpar{\scriptsize \ \ See p.~{\pageref{#1}}}}} %Giulio

\newcommand{\vsubterms}{\Lambda_\vsub}
\newcommand{\vsubkterms}{\Lambda_{\vsubk}} %Giulio

\newcommand{\vsubcalc}{\lambda_\vsub}
\newcommand{\vsubkcalc}{\lambda_{\vsubk}}

\newcommand{\firecalc}{\lambda_\mathsf{fire}}
\newcommand{\shufcalc}{\lambda_\shuf}
\newcommand{\plotcalc}{\lambda_\mathsf{Plot}}
\newcommand{\vseqsym}{\mathsf{vseq}}
\newcommand{\vseqcalc}{\lambda_{\vseqsym}}

\newcommand{\proper}{clean\xspace}

\newcommand{\harmless}{harmless\xspace}

\newcommand{\quiet}{inert\xspace}
\newcommand{\Quiet}{Inert\xspace}

\newcommand{\doubt}[1]{}

\newcommand{\letexp}{\mathtt{let}}

\newcommand{\lambdamucalc}{\lambdabar\tilde{\mu}}

\newcommand{\cm}{c}
\newcommand{\cmp}{\cm'}
\newcommand{\cmtwo}{\cm'}

\newcommand{\comm}[2]{\langle #1 \!\mid\! #2 \rangle}
\newcommand{\lbar}[3]{\lambda (#1,#2). #3}
\newcommand{\lambdabar}{\bar{\lambda}}
\newcommand{\mutildesym}{\tilde{\mu}}
\newcommand{\mut}{\tilde{\mu}}
\newcommand{\mutilde}[2]{\mutildesym #1.#2}
\newcommand{\stacker}[2]{#1 \! \cdot \! #2}
\newcommand{\append}[2]{#1  @  #2}

\newcommand{\cmctx}{C}
\newcommand{\cmctxtwo}{\cmctx'}
\newcommand{\cmctxthree}{\cmctx''}
\newcommand{\cmctxp}[1]{\cmctx\ctxholep{#1}}
\newcommand{\cmctxtwop}[1]{\cmctxtwo\ctxholep{#1}}
\newcommand{\cmctxthreep}[1]{\cmctxthree\ctxholep{#1}}
\newcommand{\rtobvmu}{\rootRew{\betav\mu}}
\renewcommand{\rtobvmu}{\rootRew{\lambdabar}} % beniamino
\newcommand{\rtomut}{\rootRew{\tilde{\mu}}}
\newcommand{\tobvmu}{\Rew{\betav\mu}}
\renewcommand{\tobvmu}{\Rew{\lambdabar}} % beniamino
\newcommand{\tomut}{\Rew{\tilde{\mu}}}
\newcommand{\tolbarmut}{\Rew{\vseqsym}} %Giulio
\newcommand{\tovseq}{\Rew{\vseqsym}} %Giulio
\newcommand{\tovseqeq}{{\tovseq}_\seqbar} %Giulio
\newcommand{\vseq}{{\vseqsym}} %Giulio
\newcommand{\tolbarmu}[1]{\underline{#1}}
\newcommand{\tolbarmuv}[1]{#1^\bullet}
\newcommand{\cotransl}[1]{#1^\diamond}

\newcommand{\Rule}{\mathsf{r}}
\newcommand{\torule}{\to_\Rule}
\newcommand{\sizemut}[1]{\size{#1}_{\mut}}

\newcommand{\sizevseq}[1]{\size{#1}_{\vseq}}
\newcommand{\sizelbar}[1]{\size{#1}_{\lambdabar}}

\newcounter{numberone}

\newenvironment{varenumerate}
{
\begin{list}{\arabic{numberone}.}
{
 \usecounter{numberone}
 \setlength{\itemsep}{.7pt}
 \setlength{\topsep}{.7pt}
 \setlength{\parsep}{.7pt}
 \setlength{\partopsep}{.7pt}
 \setlength{\leftmargin}{15pt}
 \setlength{\rightmargin}{0pt}
 \setlength{\itemindent}{0pt}
 \setlength{\labelsep}{5pt}
 \setlength{\labelwidth}{15pt}
}}
{
\end{list} 
}

\newcommand{\vsubk}{\vsub_{k}}

\newcommand{\vsubtoker}[1]{#1^+}

\newboolean{withproofs}
\setboolean{withproofs}{true}

\newtheorem{lemmaAppendix}{Lemma} %Lemma in appendix
\newtheorem{theoremAppendix}{Theorem} %Theorem in appendix
\newtheorem{propositionAppendix}{Proposition} %Proposition in appendix
\newtheorem{corollaryAppendix}{Corollary} %Corollary in appendix

\theoremstyle{remark}
\begin{document}

% \special{papersize=8.5in,11in}
% \setlength{\pdfpageheight}{\paperheight}
% \setlength{\pdfpagewidth}{\paperwidth}
% 
% \conferenceinfo{PPDP'16}{Month d--d, 20yy, City, ST, Country}
% \copyrightyear{2016}
% \copyrightdata{978-1-nnnn-nnnn-n/yy/mm}
% \copyrightdoi{nnnnnnn.nnnnnnn}
% 
% % Uncomment the publication rights you want to use.
% %\publicationrights{transferred}
% %\publicationrights{licensed}     % this is the default
% %\publicationrights{author-pays}
% 
% \titlebanner{(Submitted to PPDP 2016)}        % These are ignored unless
% \preprintfooter{Accattoli, Guerrieri - Open Call-by-Value (Submitted to PPDP 2016)}   % 'preprint' option specified.
% 
% \title{Open Call-by-Value}
% % \subtitle{Subtitle Text, if any}
% 
% \authorinfo{Beniamino Accattoli}
%            {INRIA, UMR 7161, LIX, \'Ecole Polytechnique}
%            {\href{mailto:beniamino.accattoli@inria.fr}{beniamino.accattoli@inria.fr}}
% \authorinfo{Giulio Guerrieri}
%            {Aix-Marseille Université, %Centrale 
%            Marseille, I2M UMR 7373}
%            {\href{mailto:giulio.guerrieri@univ-amu.fr}{giulio.guerrieri@univ-amu.fr}, \href{mailto:gguerrieri@uniroma3.it}{gguerrieri@uniroma3.it}}

\mainmatter  % start of an individual contribution

\title{Open Call-by-Value (Extended Version)}

\titlerunning{Open Call-by-Value}

\author{Beniamino Accattoli\inst{1}\and Giulio Guerrieri\inst{2}}
\authorrunning{B.~Accattoli\and G.~Guerrieri}


\institute{INRIA, UMR 7161, LIX, \'Ecole Polytechnique, \email{\href{mailto:beniamino.accattoli@inria.fr}{beniamino.accattoli@inria.fr}},\and
Aix Marseille Université, CNRS, Centrale Marseille, 
I2M UMR 7373, \\ F-13453 Marseille, France, 
\email{\href{mailto:giulio.guerrieri@univ-amu.fr}{giulio.guerrieri@univ-amu.fr}}
}


\toctitle{Lecture Notes in Computer Science}
\tocauthor{Authors' Instructions}
\maketitle

% \begin{abstract}
% This is the text of the abstract.
% \end{abstract}

%\category{CR-number}{subcategory}{third-level}

% general terms are not compulsory anymore,
% you may leave them out
%\terms
%term1, term2

% \keywords
% $\l$-calculus, call-by-value, operational semantics, abstract machines, cost models

 % !TEX root = main.tex
\begin{abstract}
The elegant theory of the call-by-value lambda-calculus relies on weak evaluation and closed terms, that are natural hypotheses in the study of programming languages. To model proof assistants, however, strong evaluation and open terms are required, and it is well known that the operational semantics of call-by-value becomes problematic in this case. Here we study the intermediate setting---that we call Open Call-by-Value---of weak evaluation with open terms, on top of which \gregoire and Leroy designed the abstract machine of Coq. Various calculi for Open Call-by-Value already exist, each one with its pros and cons. This paper presents a detailed comparative study of the operational semantics of four of them, coming from different areas such as the study of abstract machines, denotational semantics, linear logic proof nets, and sequent calculus. We show that these calculi are all equivalent from a termination point of view, justifying the slogan Open Call-by-Value. 
\end{abstract}

 % !TEX root = main.tex
\section{Introduction}
\label{sect:intro}
Plotkin's call-by-value $\l$-calculus \cite{DBLP:journals/tcs/Plotkin75} is at the heart of programming languages such as OCaml and proof assistants such as Coq. In the study of programming languages, call-by-value (CBV)  evaluation is usually \emph{weak}, \ie it does not reduce under abstractions, and terms are assumed to be \emph{closed}.
These constraints give rise to a beautiful theory---let us call it \emph{Closed CBV}---having the following \emph{harmony property}, that relates rewriting and normal forms:
\begin{center}
 \emph{Closed normal forms are values} (and values are normal forms)
\end{center}
where \emph{values} are variables and abstractions.
Harmony expresses a form of internal completeness with respect to unconstrained $\beta$-reduction: the restriction to CBV $\beta$-reduction (referred to as \emph{$\betav$-reduction}, according to which a $\beta$-redex can be fired only when the argument is a value) has an impact on the order in which redexes are evaluated, but evaluation never gets stuck, as every $\beta$-redex will eventually become a $\betav$-redex and be fired, unless evaluation diverges.

It often happens, however, that one needs to go beyond the perfect setting of Closed CBV by considering \emph{Strong CBV}, where reduction under abstractions is allowed and terms may be open, or the intermediate setting of \emph{Open CBV}, where evaluation is weak but terms are not necessarily closed. 
The need arises, most notably, when trying to describe the implementation model of Coq \cite{DBLP:conf/icfp/GregoireL02}, but also from other motivations, as denotational semantics \cite{DBLP:journals/ita/PaoliniR99,parametricBook,DBLP:conf/flops/AccattoliP12,DBLP:conf/fossacs/CarraroG14},  monad and CPS translations and the associated equational theories \cite{DBLP:conf/lics/Moggi89,DBLP:journals/lisp/SabryF93,DBLP:journals/toplas/SabryW97,%DBLP:conf/icfp/CurienH00,
DBLP:journals/logcom/DyckhoffL07,DBLP:conf/tlca/HerbelinZ09}, bisimulations \cite{DBLP:conf/lics/Lassen05}, partial evaluation \cite{Jones:1993:PEA:153676}, linear logic proof nets \cite{DBLP:journals/tcs/Accattoli15}, or cost models \cite{fireballs}.

\paragraph{Na\"ive Open CBV.} In call-by-name (CBN) turning to open terms or strong evaluation is harmless because CBN does not impose any special form to the arguments of $\beta$-redexes. 
On the contrary, turning to Open or Strong CBV is delicate. 
If one simply considers Plotkin's weak $\betav$-reduction on open terms---let us call it \emph{Na\"ive Open CBV}---then harmony does no longer hold, as there are open $\beta$-normal forms that are not values% (\ie neither a variable nor an abstraction)
, \eg $\var \var$, $\var (\la\vartwo\vartwo)$, $\var(\vartwo\varthree)$ or $\var\vartwo\varthree$. 
As a consequence, there are \emph{stuck $\beta$-redexes} such as $(\la\vartwo\tm) (\var \var)$, \ie $\beta$-redexes that will never be fired because their argument is normal, but it is not a value, nor will it ever become one. 
Such stuck $\beta$-redexes are a disease typical of (Na\"ive) Open CBV, but they spread to Strong CBV as well (also in the closed case), because evaluating under abstraction forces to deal with locally open terms: \eg the variable $\var$ is locally open with respect to $(\la\vartwo\tm) (\var \var)$ in $\tmthree = \la\var ((\la\vartwo\tm) (\var \var))$.%, even if %the whole term
% $\tmthree$ is closed.

The real issue with stuck $\beta$-redexes is that they prevent the creation of other redexes, and 
provide \emph{premature} $\betav$-normal forms.
The issue is serious, as it can affect termination, and thus impact on notions of observational equivalence. 
Let $\delta \defeq \la\var(\var\var)$. The problem is exemplified by the terms $\tm$ and $\tmtwo$ in Eq.~\refeq{unsolvable} below. 
\begin{align}\label{eq:unsolvable}
  \tm &\defeq ((\la\vartwo\delta) (\varthree\varthree)) \delta & \tmtwo &\defeq \delta ((\la\vartwo\delta) (\varthree\varthree) )
\end{align}
In Na\"ive Open CBV, $\tm$ and $\tmtwo$ are premature $\betav$-normal forms because they both have a stuck $\beta$-redex forbidding evaluation to keep going, while one would expect them to behave like the %famous 
divergent term $\Omega \defeq \delta\delta$ \mbox{(see \cite{DBLP:journals/ita/PaoliniR99,parametricBook,DBLP:conf/flops/AccattoliP12,DBLP:journals/tcs/Accattoli15,DBLP:conf/fossacs/CarraroG14,GuerrieriPR15} and pp.~\pageref{par:creation}-\pageref{par:potentially-valuable}% in \refsect{bird}
).}

\paragraph{Open CBV.} In his seminal work, Plotkin already pointed out an asymmetry between CBN and CBV: his CPS translation is sound and complete for CBN, but only sound for CBV. This fact led to a number of studies about monad, CPS, and logical translations \cite{DBLP:conf/lics/Moggi89,DBLP:journals/lisp/SabryF93,DBLP:journals/toplas/SabryW97,DBLP:journals/tcs/MaraistOTW99,DBLP:journals/logcom/DyckhoffL07,DBLP:conf/tlca/HerbelinZ09} that introduced many proposals of improved calculi for CBV.
 Starting with the seminal work of Paolini and Ronchi Della Rocca \cite{DBLP:journals/ita/PaoliniR99,DBLP:conf/ictcs/Paolini01,parametricBook}, the dissonance between open terms and CBV has been repeatedly pointed out and studied \emph{per se} via various calculi \cite{DBLP:conf/icfp/GregoireL02,DBLP:conf/flops/AccattoliP12,DBLP:journals/tcs/Accattoli15,DBLP:conf/fossacs/CarraroG14,GuerrieriPR15,Guerrieri15,fireballs}. 
A further point of view on CBV comes from the computational interpretation of sequent calculus due to Curien and Herbelin \cite{DBLP:conf/icfp/CurienH00}.
An important point is that the focus of most of these works is on Strong CBV. 

These solutions inevitably extend %CBV $\beta$-reduction 
$\betav$-reduction with some other rewriting rule(s) or constructor (as $\letexp$-expressions) to deal with stuck $\beta$-redexes, or even go as far as changing the applicative structure of terms, as in the sequent calculus approach. 
They arise from different perspectives and each one has its pros and cons. 
By design, these calculi (when looked at in the context of Open CBV) are never observationally equivalent to Na\"ive Open CBV, as they all manage to (re)move stuck $\beta$-redexes and may diverge when %ordinary 
Na\"ive Open CBV is instead stuck.  Each one of these calculi, however, has its own notion of evaluation and normal form, and their mutual relationships are not evident.

The aim of this paper is to draw the attention of the community on Open CBV. We believe that it is somewhat deceiving that the mainstream operational theory of CBV, however elegant, has to rely on closed terms, because it restricts the modularity of the framework, and raises the suspicion that the true essence of CBV has yet to be found. There is a real gap, indeed, between Closed and Strong CBV, as Strong CBV cannot be seen as an iteration of Closed CBV under abstractions because such an iteration has to deal with open terms. To improve the implementation of Coq \cite{DBLP:conf/icfp/GregoireL02}, \gregoire and Leroy see Strong CBV as the iteration of the intermediate case of Open CBV, but they do not explore its theory. Here we exalt their point of view, providing a thorough operational study of Open CBV. We insist on Open CBV rather than Strong CBV because:
	
	\begin{varenumerate}
% 		\item The issue with stuck $\beta$-redexes and premature $\betav$-normal forms is already visible in Open CBV;
		\item Stuck $\beta$-redexes and premature $\betav$-normal forms already affect Open CBV;
		\item Open CBV has a simpler rewriting theory than Strong CBV; 
		\item Our previous studies of Strong CBV in \cite{DBLP:conf/flops/AccattoliP12} and \cite{DBLP:conf/fossacs/CarraroG14} naturally organized themselves as properties of Open CBV that were lifted to Strong CBV by a simple iteration under abstractions.		
	\end{varenumerate}

Our contributions are along two axes:
\begin{varenumerate}
\item \emph{Termination Equivalence of the Proposals}: we show that the proposed generalizations of Na\"ive Open CBV 
 are all equivalent, in the sense that they have exactly the same sets of normalizing and diverging terms. 
%  Therefore, 
So, \emph{there is just one notion of Open CBV}, independently of its specific syntactic incarnation.
 
\item \emph{Quantitative Analyses and Cost Models}: the termination results are complemented with quantitative analyses establishing precise relationships between the number of steps needed to evaluate a given term in the various calculi. In particular, we relate the cost models of the various proposals.
 \end{varenumerate}

\paragraph{The Fab Four.} 
We focus on four proposals for Open CBV, as other solutions, \eg Moggi's \cite{DBLP:conf/lics/Moggi89} or Herbelin and Zimmerman's  \cite{DBLP:conf/tlca/HerbelinZ09}, are already known to be equivalent to these ones (see the end of \refsect{bird}): 
\begin{varenumerate}
 \item \emph{The Fireball Calculus} $\firecalc$, that extends values to \emph{fireballs} by adding so-called \emph{inert terms} in order to restore harmony---it was introduced without a name by Paolini and Ronchi Della Rocca \cite{DBLP:journals/ita/PaoliniR99,parametricBook}, then rediscovered independently first by Leroy and \gregoire \cite{DBLP:conf/icfp/GregoireL02} to improve the implementation of Coq, and then by Accattoli and Sacerdoti Coen \cite{fireballs} to study cost models;
 \item \emph{The Value Substitution Calculus} $\vsubcalc$, coming from the linear logic interpretation of CBV and using explicit substitutions and contextual rewriting rules to circumvent stuck $\beta$-redexes---it was introduced by Accattoli and Paolini \cite{DBLP:conf/flops/AccattoliP12} and it is a graph-free presentation of proof nets for the CBV $\l$-calculus \cite{DBLP:journals/tcs/Accattoli15};
 \item \emph{The Shuffling Calculus} $\shufcalc$, that has rules to shuffle constructors, similar to Regnier's $\sigma$-rules for CBN \cite{regnier94}, as an alternative to explicit substitutions---it was introduced by Carraro and Guerrieri \cite{DBLP:conf/fossacs/CarraroG14} (and further analyzed in \cite{GuerrieriPR15,Guerrieri15}) to study the adequacy of Open/Strong CBV with respect to denotational semantics related to linear logic.
 
 \item \emph{The Value Sequent Calculus} $\vseqcalc$, \ie the intuitionistic fragment of  Curien and Herbelin's $\lambdamucalc$-calculus \cite{DBLP:conf/icfp/CurienH00}, that is a CBV calculus for classical logic providing a computational interpretation of sequent calculus rather than natural deduction (in turn a fragment of the $\overline\lambda\mu\mutildesym$-calculus \cite{DBLP:conf/icfp/CurienH00}, \mbox{further studied in \eg  \cite{DBLP:journals/toplas/AriolaBS09,DBLP:conf/ifipTCS/CurienM10}).} 
 \end{varenumerate} 

\paragraph{A Robust Cost Model for Open CBV.} The number of $\betav$-steps is the canonical time cost model of Closed CBV, as first proved by Blelloch and Greiner \cite{DBLP:conf/fpca/BlellochG95,DBLP:conf/birthday/SandsGM02,DBLP:journals/tcs/LagoM08}. 
 In \cite{fireballs}, Accattoli and Sacerdoti Coen generalized this result: the number of steps in $\firecalc$ is a reasonable cost model for Open CBV. 
%This result is generalized in \cite{fireballs}: the number of steps in $\firecalc$ is a reasonable cost model for Open CBV. 
Here we show that the number of steps in $\vsubcalc$ and $\vseqcalc$ are \emph{linearly} related to the steps in $\firecalc$, thus providing reasonable cost models for these incarnations of Open CBV. 
As a consequence, complexity analyses can now be smoothly transferred between $\firecalc$, $\vsubcalc$, and $\vseqcalc$. Said differently, our results guarantee that the number of steps is a \emph{robust} cost model for Open CBV, in the sense that it does not depend on the chosen incarnation. 
For $\shufcalc$ we obtain a similar but strictly weaker result, due to some structural difficulties suggesting that \mbox{$\shufcalc$ is less apt to complexity analyses.} 
	
\paragraph{On the Value of The Paper.} While the equivalences showed here are new, they might not be terribly surprising. Nonetheless, we think they are interesting, for the following reasons:

	\begin{varenumerate}
	\item \emph{Quantitative Relationships}: $\l$-calculi are usually related only \emph{qualitatively}, while our relationships are \emph{quantitative} and thus  stronger: not only we show simulations, but we also relate the number of steps.

	\item \emph{Uniform View}: we provide a new uniform view on a known problem, that will hopefully avoid further proliferations of CBV calculi for open/strong settings. %In particular, the recent rediscover of \emph{fireballs} \cite{fireballs} provides the key tool to understand Open CBV. 
	%Fireballs were already at work in \cite{DBLP:conf/icfp/GregoireL02,DBLP:journals/ita/PaoliniR99,parametricBook} but they did not receive---we believe---the attention they deserve. %One of our contributions is the recognition of their theoretical importance.

	\item \emph{Expected but Non-Trivial}: while the equivalences are more or less expected, establishing them is informative, because it forces to reformulate and connect concepts among the different settings, and often tricky. 
		
		\item \emph{Simple Rewriting Theory}: the relationships between the systems are developed using basic rewriting concepts. 
	The technical development is simple, according to the best tradition of the CBV $\l$-calculus, and yet it provides a sharp and detailed decomposition of Open CBV evaluation. 

	\item \emph{Connecting Different Worlds}: while $\firecalc$ is related to Coq and implementations, $\vsubcalc$ and $\shufcalc$ have a linear logic background, and $\vseqcalc$ is rooted in sequent calculus. With respect to linear logic, $\vsubcalc$ has been used for syntactical studies while $\shufcalc$ for semantical ones. Our results therefore establish bridges between these different (sub)communities.
	\end{varenumerate}	
	
Finally, an essential contribution of this work is the recognition of Open CBV
as a simple and yet rich framework in between Closed and Strong CBV. 

\paragraph{Road Map.} \refsect{bird} provides an overview of the different presentations of Open CBV. 
\refsect{fireball-vsub} proves the termination equivalences for $\vsubcalc$, $\firecalc$ and $\shufcalc$, enriched with quantitative information.  
\refsect{kernel} proves the quantitative termination equivalence of $\vsubcalc$ and $\vseqcalc$, via an intermediate calculus $\vsubkcalc$.
\withproofs{
\refapp{rewriting} collects definitions and notations for the rewriting notions at work in the paper. 
Omitted proofs are in \refapp{omitted}.
}
\withoutproofs{

A longer version of this paper is available on Arxiv \cite{ourTechReport}. It contains two Appendices, %one collecting a glossary of rewriting theory and one with omitted proofs
one with a glossary of rewriting theory and one with omitted proofs.}

 % !TEX root = main.tex
\section{Incarnations of Open Call-by-Value}\label{sect:bird}

\begin{figure}[t]
  \centering
  \scalebox{0.85}{
%     \ovalbox{
    $\begin{array}{c@{\hspace{1cm}}rlllllllll}
	    \mbox{Terms} & \tm,\tmtwo,\tmthree,\tmfour &\grameq& \val \mid \tm\tmtwo\\
	    \mbox{Values} & \val,\valtwo & \grameq & \var \mid \la\var\tm\\
	    \mbox{Evaluation Contexts} & \evctx &\grameq & \ctxhole\mid \tm\evctx\mid \evctx\tm\\\\
	      
	      \textsc{Rule at Top Level} & \multicolumn{3}{c}{\textsc{Contextual closure} }\\
	      (\la\var\tm)\la\vartwo\tmtwo \rtobabs \tm\isub\var{\la\vartwo\tmtwo} &
	    \multicolumn{3}{c}{\evctxp \tm \tobabs \evctxp \tmtwo \textrm{~~~if } \tm \rtobabs \tmtwo \ }\\
	    
	    (\la\var\tm)\vartwo \rtobvar \tm\isub\var\vartwo &
	    \multicolumn{3}{c}{\evctxp \tm \tobvar \evctxp \tmtwo \textrm{~~~if } \tm \rtobvar \tmtwo \ }\\\\
	    
	    \mbox{Reduction} & \multicolumn{3}{c}{\tobv \, \defeq \, \tobabs \!\cup \tobvar}	    
    \end{array}$%}
  }
  \caption{\label{fig:plotkin-calculus} Na\"ive Open CBV $\plotcalc$}
\end{figure}
Here we recall Na\"ive Open CBV, noted $\plotcalc$, and introduce the four forms of Open CBV that will be compared ($\firecalc$, $\vsubcalc$, $\shufcalc$, and $\vseqcalc$) together with a semantic notion (\emph{potential valuability}) reducing Open CBV to Closed CBV.
%, and %known to be 
%equivalent to normalization in $\firecalc$ and $\shufcalc$ \cite{parametricBook,DBLP:conf/fossacs/CarraroG14,GuerrieriPR15}.
%\emph{Keep in mind}: 
In this paper terms are always possibly open. Moreover, we focus on Open CBV and avoid on purpose to study Strong CBV  (we hint at how to define it, though). %Strong CBV appears, \eg, in our previous works \cite{DBLP:conf/flops/AccattoliP12} and \cite{DBLP:conf/fossacs/CarraroG14}, where the results are obtained by iterating results for Open CBV.

\paragraph{%Plotkin's call-by-value $\l$-calculus 
Na\"ive Open CBV: Plotkin's calculus $\plotcalc$ \textnormal{\cite{DBLP:journals/tcs/Plotkin75}}.} 
Na\"ive Open CBV is Plotkin's weak CBV $\l$-calculus $\plotcalc$ on possibly open terms, defined in \reffig{plotkin-calculus}. 
Our presentation of the rewriting is unorthodox because we split $\betav$-reduction into two rules, according to the kind of value (abstraction or variable). 
The set of terms is denoted by $\Lambda$.
Terms (in $\Lambda$) are always identified up to $\alpha$-equivalence and the set of the free variables of a term $\tm$ is denoted by $\fv{\tm}$.
We use $\tm\isub\var\tmtwo$ for the term obtained by the capture-avoiding substitution of $\tmtwo$ for each free occurrence of $\var$ in $\tm$.
Evaluation $\tobv$ is weak and non-deterministic, since in the case of an application there is no fixed order in the evaluation of the left and right subterms. 
As it is well-known, non-determinism is only apparent: the system is strongly confluent (see \withproofs{Appendix~\ref{app:rewriting}}\withoutproofs{the appendix in \cite{ourTechReport}} for a glossary of rewriting theory). 

\newcounter{prop:basic-plotkin-strong-confluence} %new counter in order to use it in appendix
\addtocounter{prop:basic-plotkin-strong-confluence}{\value{proposition}}
\begin{proposition}
\label{prop:basic-plotkin}%\hfill
  \label{p:basic-plotkin-strong-confluence} $\tobvar$,
\NoteProof{propappendix:basic-plotkin-strong-confluence}
  $\tobabs$ and $\tobv$ are strongly confluent.
\end{proposition}
% }

Strong confluence is a remarkable property, much stronger than plain confluence. It implies that, given a term, \emph{all} derivations to its normal form (if any) have the \emph{same length}, and that \emph{normalization and strong normalization coincide}, \ie if there is a normalizing derivation then there are no diverging derivations. Strong confluence will also hold for %two out of the three calculi that we shall consider, 
$\firecalc$, $\vsubcalc$ and $\vseqcalc$, not for $\shufcalc$.

Let us come back to the splitting of $\tobv$. In Closed CBV it is well-known that $\tobvar$ is superfluous, at least as long as small-step evaluation is considered, see \cite{DBLP:conf/wollic/AccattoliC14}. For Open CBV, $\tobvar$ is instead necessary, but---as we explained in the introduction---it is not enough, which is why we shall consider extensions of $\plotcalc$. The main problem of Na\"ive Open CBV is that there are stuck $\beta$-redexes breaking the harmony of the system. 
There are three kinds of solution: those \emph{restoring a form of harmony} ($\firecalc$), to be thought as more semantical approaches; those \emph{removing stuck $\beta$-redexes} ($\vsubcalc$ and $\shufcalc$), that are more syntactical in nature; those \emph{changing the applicative structure of terms} \mbox{($\vseqcalc$)%, corresponding to a comuptational interpretation of sequent calculus
, inspired by sequent calculus.}
\begin{figure}[t]
  \centering
  \scalebox{0.85}{
%     \ovalbox{
    $
    \begin{array}{c@{\hspace{.5cm}}rll}
	    \mbox{Terms and Values} & \multicolumn{3}{c}{\textup{As in Plotkin's Open CBV (Fig.~\ref{fig:plotkin-calculus})\ }} \\
    \mbox{Fireballs} & \!\!\!\!\!\!\!\!\!\!\!\!\!\!\!\!\!\!\!\!\!\!\fire, \firetwo, \firethree & \grameq & \la\var\tm \mid \gconst\\
	    \mbox{Inert Terms} & \,\,\,\gconst,\gconsttwo, \gconstthree & \grameq &  \var \fire_1\ldots \fire_n\ \ \ \ n\geq 0\\
    %	\mbox{Inert Terms} & \gconst,\gconsttwo, \gconstthree & \grameq &  \var \fire \mid \gconst\fire
	    \mbox{Evaluation Contexts} & \evctx  & \grameq & \ctxhole\mid \tm\evctx \mid \evctx\tm \\\\	
	    
      \textsc{Rule at Top Level} & \multicolumn{3}{c}{\textsc{Contextual closure}} \\
	    \!(\la\var\tm)(\la\vartwo\tmtwo) \rtobabs \tm\isub\var{\la\vartwo\tmtwo} &
	    \multicolumn{3}{c}{\evctxp \tm \tobabs \evctxp \tmtwo \textrm{~~~if } \tm \rtobabs \tmtwo} \\
	    \,\,(\l\var.\tm)\gconst\rtoin \tm\isub\var\gconst &
	    \multicolumn{3}{c}{\evctxp \tm \toin \evctxp \tmtwo \quad\textup{if } \tm \rtoin \tmtwo} \\\\ 

	    \mbox{Reduction} & \multicolumn{3}{c}{\tof \, \defeq \, \tobabs \!\cup \toin}
    \end{array}
    $%}
  }
  \caption{\label{fig:fireball-calculus} The Fireball Calculus $\firecalc$}
\end{figure}

% \paragraph{Open Call-by-Value 1: The Fireball Calculus $\firecalc$.}
\subsection{\texorpdfstring{Open Call-by-Value 1: The Fireball Calculus $\firecalc$}{Open Call-by-Value 1: The Fireball Calculus}}
\label{subsect:fireball}
% !TEX root = main.tex
The Fireball Calculus $\firecalc$, defined in \reffig{fireball-calculus}, was introduced without a name by Paolini and Ronchi Della Rocca in \cite{DBLP:journals/ita/PaoliniR99} and \cite[Def.~3.1.4, p.~36]{parametricBook} where its basic properties are also proved. We give here a presentation inspired by Accattoli and Sacerdoti Coen's \cite{fireballs}, departing from it only for inessential, cosmetic details. Terms, values and evaluation contexts are the same as in $\plotcalc$.

The idea is to restore harmony by generalizing $\tobvar$ to fire when the argument is a more general \emph{inert term}---the new rule is noted $\toin$. The generalization of values as to include inert terms is called \emph{fireballs}. Actually fireballs and inert terms are defined by mutual induction (in \reffig{fireball-calculus}). 
For instance, $\la\var\vartwo$ is a fireball %because they are
as an abstraction, while $\var$, $\vartwo(\la\var\var)$, $\var\vartwo$, and $(\varthree(\la\var\var))(\varthree\varthree) (\la\vartwo(\varthree\vartwo))$ are fireballs as inert terms. 
Note that $\gconst\gconsttwo$ is an inert term for all inert terms $\gconst$ and $\gconsttwo$. 
Inert terms can be equivalently defined as $\gconst \grameq  \var \mid \gconst\fire$\withproofs{---such a definition is used in the proofs in the Appendix
%\withoutproofs{ of \cite{ourTechReport}}\withproofs{ 
(where, moreover, inert terms that are not variables are referred to as \emph{compound inert terms})}.
The main feature of an inert term is that it is open, normal and that when plugged in a context it cannot create a redex, hence the name (it is not a so-called \emph{neutral term} because it might have redexes under abstractions). In \gregoire and Leroy's presentation \cite{DBLP:conf/icfp/GregoireL02}, inert terms are called \emph{accumulators} and fireballs are simply called values.

%it is not a value and if it contains $\beta$-redexes (\ie subterms of the form $(\la\var\tm)\tmtwo$) then they are under the scope of some $\l$.

% Dynamically, the new inert rule $\toin$ fires $\beta$-redexes when the argument is a inert term. 
% The fireball rule $\tof$ then is %simply defined as 
% the union of $\tobv$ and $\toin$, \ie $\tm \tof \tmtwo$ if $\tm \tobv \tmtwo$ or $\tm \toin \tmtwo$
% The reduction $\toin$ removes stuck $\beta$-redexes. 
Evaluation is given by the fireball rule $\tof$, that is the union of $\tobabs$ and $\toin$. For instance, consider %the terms 
$\tm \defeq ((\la\vartwo \delta)(\varthree\varthree))\delta$ and $\tmtwo \defeq \delta((\la\vartwo \delta)(\varthree\varthree))$ as in Eq.~\refeq{unsolvable}, p.~\pageref{eq:unsolvable}: $\tm$ and $\tmtwo$ are $\betav$-normal but they diverge when evaluated in $\firecalc$, as desired: $\tm \toin \delta\delta \tobabs  \delta\delta \tobabs \!\dots$ and $\tmtwo \toin \delta\delta \tobabs \delta\delta \tobabs \!\dots$\,.

The distinguished, key property of $\firecalc$ is (for any $\tm \in \Lambda$):

\newcounter{prop:open-harmony} %new counter in order to use it in appendix
\addtocounter{prop:open-harmony}{\value{proposition}}
\begin{proposition}[Open Harmony]\label{prop:open-harmony}\label{prop:characterize-fnormal}
%   A (possibly open) term $\tm$ is a fireball iff it is a $\betaf$-normal form.
  %Let $\tm \in \Lambda$:
  $\tm$ 
  \NoteProof{propappendix:open-harmony}
  is $\betaf$-normal iff $\tm$ is a fireball.
\end{proposition}

%\noindent Be careful: it looks like $\toin$ allows to reduce all stuck $\beta$-redexes, but this is not quite true. The argument of a stuck $\beta$-redex is not always an inert term, in fact. Nonetheless, in a premature $\betav$-normal form there always is at least one stuck $\beta$-redex whose argument is a inert term, so that $\toin$ can be applied and evaluation can continue, \ie the argument of every $\beta$-redex will \emph{eventually} become a fireball (unless evaluation diverges).
The advantage of $\firecalc$ is its simple notion of normal form, \ie fireballs, that have a clean syntactic description akin to that for call-by-name. The other calculi will lack a nice, natural notion of normal form. 
%The concepts of $\firecalc$, however, will allow us to somewhat identify a good notion of normal form also for $\vsubcalc$ and $\vseqcalc$.
The drawback of the fireball calculus---and probably the reason why its importance did not emerge before---is the fact that as a strong calculus it is not confluent:
this is %essentially 
due to the fact that fireballs are not closed by substitution (see \cite[p.~37]{parametricBook}). 
% Let $\delta$ be the usual duplicator combinator, s.t. $\delta\delta$ diverges. 
Indeed, if evaluation is strong, the following critical pair cannot be joined, %for instance
where $\tm \defeq (\la\vartwo I) (\delta \delta)$ and $I \defeq \la\varthree\varthree$ is the identity combinator:
% \begin{equation}\label{eq:counterexample}
%   \varthree \lRew{\betav\!\!} (\la\var \varthree) \delta \lRew{\betain\!} (\la \var (\la\vartwo \varthree) (\var\var)) \delta \tobv\! (\la\vartwo \varthree) (\delta \delta) \tobv\! (\la\vartwo \varthree) (\delta \delta) \tobv \!\ldots
% \end{equation}
\begin{equation}\label{eq:counterexample}
  I \!\lRew{\betaabs\!\!} (\la\var I) \delta \!\lRew{\betain\!\!} (\la \var (\la\vartwo I) (\var\var)) \delta \tobabs\! \tm \tobabs\! \tm \tobabs \!\!\ldots
\end{equation}

On the other hand, as long as evaluation is weak (that is the case we consider) everything works fine---the strong case can then be caught by repeatedly iterating the weak one under abstraction, once a weak normal form has been obtained (thus forbidding the left part of \refeq{counterexample}). 
In fact, %fireball evaluation 
the weak evaluation of $\firecalc$ has a simple rewriting theory, as next proposition shows. In particular \mbox{it is strongly confluent.}

%The presentation we give here differs from the one in \cite{fireballs} in three points. First, we use \emph{theoretical values} in the sense of \cite{DBLP:conf/wollic/AccattoliC14}, \ie values here include variables (as in \cite{DBLP:journals/ita/PaoliniR99,parametricBook}), while in \cite{fireballs} (and \cite{DBLP:conf/icfp/GregoireL02}) they do not. Such a change is inessential, especially because here we consider small-step rules, while the difference is observable only with micro-level rules (see the submitted long version of \cite{DBLP:conf/wollic/AccattoliC14}, on Accattoli's webpage). Second, in \cite{fireballs} free variables where turned in a different syntactic category and noted $\const,\consttwo, \ldots$, while here both bound and free variables are both noted $\var,\vartwo,\varthree$---the change is clearly only cosmetic.  Third, the fireball rule $\tof$ of \cite{fireballs} is here split in two rules $\tobv$ and $\toin$. This point is instead relevant and has a role in our study of cost models. In particular, the inert rule alone cannot diverge. %Such a feature will be repeatedly used in this paper.

\newcounter{prop:basic-fireball} %new counter in order to use it in appendix
\addtocounter{prop:basic-fireball}{\value{proposition}}
\begin{proposition}[Basic Properties of $\firecalc$]
\label{prop:basic-fireball}\hfill
\NoteProof{propappendix:basic-fireball}
\begin{varenumerate}
	  \item\label{p:basic-fireball-toin-strong-normalization}\label{p:basic-fireball-toin-strong-confluence} $\toin$ is strongly normalizing and strongly confluent.
	  \item\label{p:basic-fireball-tobv-toin-strong-commutation} $\tobabs$ and $\toin$ strongly commute.	  
	  \item\label{p:basic-fireball-strong-confluence} \label{p:basic-fireball-number-steps} $\tof$ is strongly confluent, and all $\betaf$-normalizing derivations $\deriv$ from $\tm \in \Lambda$ (if any) have the same length $\sizef{\deriv}$, the same number $\sizebabs{\deriv}$ of $\betaabs$-steps, and the same number $\sizein{\deriv}$ of $\betain$-steps.	  
  \end{varenumerate}

\end{proposition}

% \newcounter{prop:basic-fireball} %new counter in order to use it in appendix
% \addtocounter{prop:basic-fireball}{\value{proposition}}
% \begin{proposition}[Basic properties of the Fireball Calculus]
% \label{prop:basic-fireball}\hfill
% \NoteProof{propappendix:basic-fireball}
%   \begin{enumerate}
% 	  \item\label{p:basic-fireball-strong-confluence} $\tof$ is strongly confluent.
% 	  \item\label{p:basic-fireball-weak-strong-normalize} Let $\tm \in \Lambda$: $\tm$ is strongly $\betaf$-normalizable iff $\tm$ is $\betaf$-normalizable.
% %  	  \item\label{p:basic-fireball-number-steps} If $\deriv \colon \tm \tof^* \tmtwo$ and $\derivtwo \colon \tm \tof^* \tmtwo$ with $\tmtwo$ $\betaf$-normal, then $\size\deriv = \size\derivtwo$, $\sizebv\deriv = \sizebv\derivtwo$ and $\sizein\deriv = \sizein\derivtwo$.
%  	  \item\label{p:basic-fireball-number-steps} Let $\tm \in \Lambda$: all $\betaf$-normalizing reduction sequences from $\tm$ (if any) have the same length, the same number of $\betav$-steps and the same number of $\betain$-steps.
%   \end{enumerate}
%   \end{proposition}

	\subsection{\texorpdfstring{Open Call-by-Value 2: The Value Substitution Calculus $\vsubcalc$}{Open Call-by-Value 2: The Value Substitution Calculus}}
	\label{subsect:vsub}

	\paragraph{Rewriting Preamble: Creations of Type 1 and 4.}\label{par:creation}
	The problem with stuck $\beta$-redexes can be easily understood at the rewriting level as an issue about creations. According to  \levy \cite{thesislevy}, in the ordinary CBN $\l$-calculus redexes can be created in 3 ways. 
	Creations of type 1 take the following form
	\begin{equation*}
	  ((\la\var\la\vartwo\tm)\tmfour)\tmthree \tob (\la\vartwo\tm \isub\var\tmfour) \tmthree
	\end{equation*}
	where the redex involving $\l\vartwo$ and $\tmthree$ has been created by the $\beta$-step. 
	%Now, 
	In Na\"ive Open CBV if $\tmfour$ is a normal form %that is 
	and not a value then the creation cannot take place, blocking evaluation. 
	%Note that this 
	\mbox{This is %exactly 
	the problem concerning the term $\tm$ in %the introduction.
	Eq.~\refeq{unsolvable}, p.~\pageref{eq:unsolvable}.}	
% 	Actually, in CBV there also is a form of creation not considered by \levy, let's call it of \emph{type 4}: 
	In CBV there is another form of creation---of \emph{type 4}---not considered by \levy: 
	\begin{equation*}(\la\var\tm) ((\la\vartwo\val) \valtwo) \tobv (\la\var\tm) (\val\isub\vartwo\valtwo)\end{equation*}
	\ie a reduction in the argument turns the argument itself into a value, creating a $\betav$-redex. 
	As before, in an open setting $\valtwo$ may be replaced by a normal form that is not a value, blocking the creation of type 4.
	This is exactly the problem concerning the term $\tmtwo$ in %the introduction.
	Eq.~\refeq{unsolvable}, p.~\pageref{eq:unsolvable}. 
	
	The proposals of this and the next sections introduce some way to enable creations of type 1 and 4, without %actually
	\mbox{substituting stuck $\beta$-redexes nor inert terms.}\bigskip
	
% 	\paragraph{Open Call-by-Value 2: The Value Substitution Calculus $\vsubcalc$.} 
	% !TEX root = main.tex
	The \emph{value substitution calculus} $\vsubcalc$ of Accattoli and Paolini \cite{DBLP:conf/flops/AccattoliP12,DBLP:journals/tcs/Accattoli15} was introduced as a calculus for Strong CBV inspired by linear logic proof nets. In \reffig{valuesub-calculus} we present its adaptation to Open CBV, obtained by simply removing abstractions from evaluation contexts. 
	It extends the syntax of terms with the constructor $\esub\var\tmtwo$, called \emph{explicit substitution} (shortened ES, to not be confused with the meta-level substitution $\isub\var\tmtwo$). A $\vsub$-term $\tm\esub\var\tmtwo$ represents %and representing
      the delayed substitution of $\tmtwo$ for $\var$ in $\tm$, \ie %standing 
      stands for $\mathtt{let} \ \var = \tmtwo \ \mathtt{in} \ \tm$. So, $\tm\esub\var\tmtwo$ binds the free occurrences of $\var$ in $\tm$. The set of $\vsub$-terms---identified up to $\alpha$-equivalence---is \mbox{denoted by $\vsubterms$ (clearly $\Lambda \subsetneq \vsubterms$).}
	\doubt{Note that $\la\var(\vartwo\esub\var{\varthree\varthree}) \in \vsubterms \smallsetminus \Lambda$ is a $\vsub$-value and not a value, and that evaluation contexts are extended to ES.
	Substitution is extended to ES as expected: $(\tm\esub\var\tmtwo)\isub\vartwo\tmthree \!\defeq\! \tm\isub\vartwo\tmthree \esub\var{\tmtwo\isub\vartwo\tmthree}$ with $\var \!\notin\! \fv{\tmthree} \cup \{\vartwo\}$.}% \eg 

\begin{figure}[t]
  \centering
  \scalebox{0.85}{
%     \ovalbox{
    $
    \begin{array}{c@{\hspace{-.3cm}}rcc}
	    \vsub\textup{-Terms} &\!\!\!\!\!\! \tm, \tmtwo, \tmthree & \grameq & \multicolumn{1}{l}{\val \mid \tm\tmtwo \mid \tm\esub\var\tmtwo}\\
	    \vsub\textup{-Values} & \val  & \grameq  & \multicolumn{1}{l}{\var \mid \la\var\tm}\\	
	    \textup{Evaluation Contexts} & \evctx & \grameq &  \multicolumn{1}{l}{\ctxhole\mid \tm\evctx\mid \evctx\tm \mid \evctx\esub\var\tmtwo \mid \tm\esub\var\evctx} \\
	    \textup{Substitution Contexts} & \sctx & \grameq & \multicolumn{1}{l}{\ctxhole\mid \sctx\esub\var\tmtwo}\\\\
	    
      \textsc{Rule at Top Level} & \multicolumn{3}{c}{\textsc{Contextual closure}} \\
	    \!\sctxp{\l\var.\tm}\tmtwo\rtom \sctxp{\tm\esub\var\tmtwo} &
	    \multicolumn{3}{c}{\!\evctxp \tm \tom \evctxp \tmtwo \textrm{~~~\,if } \tm \rtom \tmtwo} \\

% 	    \tm\esub\var{\sctxp\val} \rtoe \sctxp{\tm\isub\var\val} &
% 	    \multicolumn{3}{c}{\evctxp \tm \toe \evctxp \tmtwo \textrm{~~~if } \tm \rtoe \tmtwo} \\\\

	    \tm\esub\var{\sctxp{\la\vartwo\tmtwo}} \rtoeabs \sctxp{\tm\isub\var{\la\vartwo\tmtwo}} &
	    \multicolumn{3}{c}{\evctxp \tm \toeabs \evctxp \tmtwo \textrm{~~~if } \tm \rtoeabs \tmtwo} \\

	    \tm\esub\var{\sctxp\vartwo} \rtoevar \sctxp{\tm\isub\var\vartwo} &
	    \multicolumn{3}{c}{\evctxp \tm \toevar \evctxp \tmtwo \textrm{~~~if } \tm \rtoevar \tmtwo} \\\\

	    \mbox{Reductions} & \multicolumn{3}{c}{\!\!\!\!\toe \, \defeq \, \toeabs \!\cup \toevar, \ \ \tovsub \, \defeq \, \tom \!\cup \toe} %\\\\

    % 	\textup{Explicit Fireballs} & \efire & \grameq & \multicolumn{1}{l}{\isctxp\val \mid \isctxp \econst}\\
    % 	
    % 		\textup{Explicit Inert Terms} & \econst & \grameq & \multicolumn{1}{l}{\isctxp\var \efire_1\ldots \efire_n\ \ \ \ n\geq 1}\\
    % 
    % 	\textup{Inert Substitution Contexts} & \isctx & \grameq & \multicolumn{1}{l}{\ctxhole \mid \isctx \esub\var{\isctxtwop\econst}}

    \end{array}
    $%}
  }
  \caption{\label{fig:valuesub-calculus} The Value Substitution Calculus $\vsubcalc$}
\end{figure}

	ES are used to remove stuck $\beta$-redexes: the idea is that $\beta$-redexes can be fired whenever---even if the argument is not a ($\vsub$-)value---by means of the \emph{multiplicative rule} $\tom$; however the argument is not substituted but placed in a ES. 
	The actual substitution is done only when the content of the %substitution 
	ES is a $\vsub$-value, by means of the \emph{exponential rule} $\toe$. 
	These two rules are sometimes noted $\todb$ ($\beta$ at a distance) and $\tovs$ (substitution by value)---the names we use here are due to the interpretation of the calculus into linear logic proof-nets, see \cite{DBLP:journals/tcs/Accattoli15}. 
	\withproofs{Note that in \reffig{valuesub-calculus} the definition of the rewriting rules at top level $\rtom$ (resp.~$\rtoevar$ and $\rtoeabs$) assumes that the variables binded by the substitution context $\sctx$ are not free in $\tmtwo$ (resp.~$\tm$).} 
	A characteristic feature coming from such an interpretation is that the rewriting rules are contextual, or \emph{at a distance}: they are generalized as to act up to a list of substitutions (noted $\sctx$, from \emph{L}ist). 
	Essentially, stuck $\beta$-redexes are turned into ES and then ignored by the rewriting rules---this is how creations of type 1 and 4 are enabled. 
	For instance, the terms $\tm \defeq ((\la\vartwo \delta)(\varthree\varthree))\delta$ and $\tmtwo \defeq \delta((\la\vartwo \delta)(\varthree\varthree))$ (as in Eq.~\refeq{unsolvable}, p.~\pageref{eq:unsolvable}) are $\expo$-normal but $\tm \tom  \delta\esub\vartwo{\varthree\varthree} \delta \tom (xx) \esub\var\delta \esub\vartwo{\varthree\varthree} \toe (\delta\delta) \esub\vartwo{\varthree\varthree} \tom (xx) \esub\var\delta \esub\vartwo{\varthree\varthree} \toe (\delta\delta) \esub\vartwo{\varthree\varthree} \tom \dots$ and similarly for $\tmtwo$.
	%$\tmtwo \tom \delta\delta\esub\vartwo{\varthree\varthree} \tom (\var\var)\esub\var{\delta\esub\vartwo{\varthree\varthree}} \toe (\delta\delta)\esub\vartwo{\varthree\varthree} \tom (xx) \esub\var\delta \esub\vartwo{\varthree\varthree} \toe (\delta\delta) \esub\vartwo{\varthree\varthree} \tom  \dots$\,.
	
	The drawback of $\vsubcalc$ is that it requires explicit substitutions.	The advantage of $\vsubcalc$ is its simple and well-behaved rewriting theory, even simpler than the rewriting for $\firecalc$, since every rule terminates separately (while $\betaabs$ does not)---in particular strong confluence holds. Moreover, the theory has a sort of flexible second level given by a notion of structural equivalence, coming up next.

\newcounter{prop:basic-value-substitution} %new counter in order to use it in appendix
\addtocounter{prop:basic-value-substitution}{\value{proposition}}
\begin{proposition}[Basic Properties of $\vsubcalc$, \cite{DBLP:conf/flops/AccattoliP12}]
  {\ }\NoteProof{propappendix:basic-value-substitution}
\hfill
\label{prop:basic-value-substitution}
  \begin{varenumerate}
%     \item\label{p:basic-value-substitution-tom-toe-terminates} $\tom$ and $\toe$ are strongly normalizing (separately).
%     \item\label{p:basic-value-substitution-tom-toe-strong-confluence} $\tom$ and $\toe$ are strongly confluent (separately).
    \item\label{p:basic-value-substitution-tom-toe-terminates}\label{p:basic-value-substitution-tom-toe-strong-confluence} $\tom$ and $\toe$ are strongly normalizing and strongly confluent (separately).
    
    \item\label{p:basic-value-substitution-tom-toe-commute}  $\tom$ and $\toe$ strongly commute.
    \item\label{p:basic-value-substitution-strong-confluence}\label{p:basic-value-substitution-number-steps} $\tovsub$ is strongly confluent, and all $\vsub$-normalizing derivations $\deriv$ from $\tm \!\in\! \vsubterms$ (if any) have the same length $\sizevsub{\deriv}$, the same number $\sizee{\deriv}$ of $\expo$-steps, and the same number $\sizem{\deriv}$ of $\mult$-steps %(not necessarily the same number of ).
    \item\label{p:basic-value-substitution-expo-less-than-mult} Let $\tm \!\in\! \Lambda$. For any $\vsub$-derivation $\deriv$ from $\tm$, $\sizee{\deriv} \leq \sizem{\deriv}$.
  \end{varenumerate}
\end{proposition}
% }

% \ben{[To do: integrare il lemma successivo come un punto ulteriore della proposizione precedente]}
% \newcounter{l:expo-less-than-mult}
% \addtocounter{l:expo-less-than-mult}{\value{lemma}}
% \begin{lemma}[Number of $\mult$- and $\expo$-Steps in a $\vsub$-Derivation from a Term]
%   \label{l:expo-less-than-mult}
%   Let 
% \NoteProof{lappendix:expo-less-than-mult}
%   $\tm \!\in\! \Lambda$. For any $\vsub$-derivation $\deriv$ from $\tm$, $\sizee{d} \leq \sizem{\deriv}$.
% \end{lemma}

\noindent\emph{Structural Equivalence.}
The theory of $\vsubcalc$ comes with a notion of structural equivalence $\eqstruct$, that equates $\vsub$-terms that differ only for the position of ES. The basic idea is that the action of an ES via the exponential rule depends on the position of the ES itself only for inessential details (as long as the scope of binders is respected), namely the position of other ES, and thus can be abstracted away. 
A strong justification for the equivalence comes from the linear logic interpretation of $\vsubcalc$, in which structurally equivalent $\vsub$-terms translate to the same (recursively typed) proof net, see \cite{DBLP:journals/tcs/Accattoli15}.

Structural equivalence $\eqstruct$ is defined as the least equivalence relation on $\vsubterms$ closed by evaluation contexts (see \reffig{valuesub-calculus}) and generated by the following axioms:
% \begin{center}
% \label{eq:eqstruct}
%   $\begin{array}{rll@{\hspace*{0.75cm}}l}
% 		\tm\esub{\var}{\tmtwo}\esub{\vartwo}{\tmthree} &\tostructcom& \tm\esub{\vartwo}{\tmthree}\esub{\var}{\tmtwo}&\mbox{if $\vartwo\notin\fv{\tmtwo}$ and $\var\notin\fv{\tmthree}$}\\
% % 		(\tm\tmthree)\esub{\var}{\tmtwo} &\tostructapr&  \tm\tmthree\esub{\var}{\tmtwo} & \textrm{if }\var\not\in\fv\tm \\
% 		(\tm\tmthree)\esub\var\tmtwo &\tostructapr&  \tm\esub\var\tmtwo\tmthree & \textrm{if }\var\not\in\fv{\tm} \\
% 		(\tm\tmthree)\esub\var\tmtwo &\tostructapl& \tm\esub\var\tmtwo\tmthree & \textrm{if }\var\not\in\fv\tmthree \\
% 		\tm\esub{\var}{\tmtwo}\esub{\vartwo}{\tmthree} &\tostructes& \tm\esub{\var}{\tmtwo\esub{\vartwo}{\tmthree}} & \mbox{if $\vartwo\not\in\fv{\tm}$}  
%   \end{array}$
% \end{center}
\begin{align*}
	\tm\esub{\vartwo}{\tmthree}\esub{\var}{\tmtwo} &\tostructcom \tm\esub{\var}{\tmtwo}\esub{\vartwo}{\tmthree} &&\mbox{if $\vartwo\notin\fv{\tmtwo}$ and $\var\notin\fv{\tmthree}$}\\
% 		(\tm\tmthree)\esub{\var}{\tmtwo} &\tostructapr&  \tm\tmthree\esub{\var}{\tmtwo} & \textrm{if }\var\not\in\fv\tm \\
	\tm\,\tmthree\esub\var\tmtwo &\tostructapr  (\tm\tmthree)\esub\var\tmtwo && \textrm{if }\var\not\in\fv{\tm} \\
	\tm\esub\var\tmtwo\tmthree &\tostructapl (\tm\tmthree)\esub\var\tmtwo && \textrm{if }\var\not\in\fv\tmthree \\
	\tm\esub{\var}{\tmtwo\esub{\vartwo}{\tmthree}} &\tostructes \tm\esub{\var}{\tmtwo}\esub{\vartwo}{\tmthree} && \mbox{if $\vartwo\not\in\fv{\tm}$}  
\end{align*}
\label{eq:eqstruct}
We %use $\tovsubeq$ for 
set $\tovsubeq \, \defeq \, \eqstruct\tovsub\eqstruct$ (\ie for all $\tm, \tmfour \in \vsubterms$: $\tm \tovsubeq \tmfour$ iff $\tm\eqstruct\tmtwo\tovsub\tmthree\eqstruct\tmfour$ for some $\tmtwo, \tmthree \in \vsubterms$). The notation $\tovsubeq^+$  keeps its usual meaning, while $\tovsubeq^*$ stands for $\eqstruct \cup\tovsubeq^+$, \ie a $\vsubeq\,$-derivation of length zero can apply $\eqstruct$ and is not just the identity.
As $\eqstruct$ is reflexive, $\tovsub \, \subsetneq \, \tovsubeq$. 

The rewriting theory of $\vsubcalc$ enriched with structural equivalence $\eqstruct$ is remarkably simple, as next lemma shows. In fact, $\eqstruct$ commutes with evaluation, and can thus be postponed. Additionally, the commutation is \emph{strong}, as it preserves the number and kind of steps---one says that it is a \emph{strong bisimulation} (with respect to $\tovsub$). 
In particular, the equivalence is not needed to compute and it does not break, or make more complex, any property of $\vsubcalc$. On the contrary, it enhances %the equational theory and 
the flexibility of the system: it will be essential to establish simple and clean relationships with the other calculi for Open CBV.

\begin{lemma}[Basic Properties of Structural Equivalence $\eqstruct$, \cite{DBLP:conf/flops/AccattoliP12}]
\label{l:eqstruct-post-and-term} % \reflemmap{eqstruct-post-and-term}{strong-confluence}
  Let $\tm, \tmtwo \in \vsubterms$ and $\mathsf{x}\in\set{\mult,\expoabs, \expovar,\expo,\vsub}$.
    \begin{varenumerate}
      \item\label{p:eqstruct-post-and-term-locpost} \emph{Strong Bisimulation of $\eqstruct$ wrt $\tovsub$}: 
      if $\tm\eqstruct\tmtwo$ and $\tm\Rew{\mathsf{x}}\tmp$ then there exists $\tmtwop \in \vsubterms$ such that $\tmtwo\Rew{\mathsf{x}}\tmtwop$ and $\tmp\eqstruct\tmtwop$.      
      \item\label{p:eqstruct-post-and-term-globpost} \emph{Postponement of $\eqstruct$ wrt $\tovsub$}: 
% 	$\tovsubeq^k \, \subseteq \, \tovsub^k\eqstruct$ preserving the kind and number of steps.
      if $\deriv \colon \tm \tovsubeq^* \tmtwo$ then there are $\tmthree \eqstruct \tmtwo$ and $\derivtwo \colon \tm \tovsub^* \tmthree$ such that $\size{\deriv} = \size{\derivtwo}$, $\sizeeabs{\deriv} = \sizeeabs{\derivtwo}$, $\sizeevar{\deriv} = \sizeevar{\derivtwo}$ and $\sizem{\deriv} = \sizem{\derivtwo}$.
      \item\label{p:eqstruct-post-and-term-normal}\emph{Normal Forms}: if $\tm \eqstruct \tmtwo$ then $\tm$ is $\mathsf{x}$-normal iff $\tmtwo$ is $\mathsf{x}$-normal.
% % % 	\item \emph{Simulation and Termination Equivalence of $\tovsub$ and $\tovsubeq$}: \label{p:eqstruct-post-and-term-term}
% % % % 	let $\deriv \colon \tm\tovsub^*\tmtwo$ be a normalising derivation.
% % % % 	Then there exists a normalising derivation $\derivtwo \colon \tm\tovsubeq^*\tmthree$ with $\sizee\derivtwo = \sizee\deriv$.
% % % 	If $\deriv \colon \tm\tovsubeq^*\tmtwo$ then there exists $\derivtwo \colon \tm\tovsub^*\tmthree$ with $\sizee\derivtwo = \sizee\deriv$.
% % % 	If moreover $\deriv$ is $\vsub$-normalizing then $\derivtwo$ is $\vsub$-normalizing.
      \item\label{p:eqstruct-post-and-term-strong-confluence} \emph{Strong confluence}:       $\tovsubeq$ is strongly confluent.
    \end{varenumerate}	
    
\end{lemma}

%The first point is a variant of \cite[Lemma 12]{DBLP:conf/flops/AccattoliP12}, stating that $\eqstruct$ is a strong bisimulation, the other points are immediate consequences. \doubt{It is easy to check that $\tm \tovsub \tmtwo$ implies $\tm \not\eqstruct \tmtwo$; so, given a $\vsubeq$-derivation $d$, the sizes $\sizee{\deriv}$, $\sizem{\deriv}$ and $\sizevsub{\deriv}$ are uniquely determined.}

% 	\paragraph{Open Call-by-Value 3: The Shuffling Calculus $\shufcalc$.} 
	\subsection{\texorpdfstring{Open Call-by-Value 3: The Shuffling Calculus $\shufcalc$}{Open Call-by-Value 3: The Shuffling Calculus}}
	\label{subsect:shuffling}
	% !TEX root = main.tex
The calculus introduced by Carraro and Guerrieri in \cite{DBLP:conf/fossacs/CarraroG14}, and here deemed \emph{Shuffling Calculus}, has the same syntax of terms as Plotkin's calculus. 
Two additional commutation rules help $\tobv$ to deal with stuck $\beta$-redexes, by shuffling constructors so as to enable creations of type 1 and 4. As for $\vsubcalc$, $\shufcalc$ was actually introduced, and then used in \cite{DBLP:conf/fossacs/CarraroG14,Guerrieri15,GuerrieriPR15}, to study Strong CBV.
In \reffig{shuffling-calculus} we present its adaptation to Open CBV, based on \emph{balanced contexts}, a special notion  of evaluation contexts. The reductions $\tosigm$ and $\tobvm$ are non-deterministic and---because of balanced contexts---can reduce under abstractions, but they are \emph{morally} weak: they reduce under a $\lambda$ only when the $\lambda$ is applied to an argument.
Note that the condition $\var \notin \fv{\tmthree}$ (resp.~$\var \notin \fv{\val}$) in the definition of the shuffling rule $\rtosl$ (resp.~$\rtosr$) can always be fulfilled by $\alpha$-conversion.

%The reduction of this calculus is denoted by $\tovm$.

% 
% 	\begin{center}$\begin{array}{l@{\hspace{1.5cm}}llllllllllll}
% 	\mbox{Balanced Contexts} & \mctx &\grameq & \ctxhole\mid \tm\mctx \mid \mctx\tm \mid (\la\var\mctx) \tm
% \end{array}$\end{center}
% \begin{align*}
%   \textup{Balanced Contexts}& & \mctx &\grameq  \ctxhole\mid \tm\mctx \mid \mctx\tm \mid (\la\var\mctx) \tm
% \end{align*}

% The shuffling rules $\tosl$ and $\tosr$ and the extension of $\tobv$ to balanced contexts are defined by:
\begin{figure}[t]
  \centering
  \scalebox{0.85}{
%     \ovalbox{
    $
    \begin{array}{c@{\hspace{-.5cm}}c}
	    \mbox{Terms and Values} & \textup{As in Plotkin's Open CBV (Fig.~\ref{fig:plotkin-calculus})\ }\\
	    \mbox{Balanced Contexts} & \mctx \grameq  \ctxhole\mid \tm\mctx \mid \mctx\tm \mid (\la\var\mctx) \tm\\\\	
	    
      \textsc{Rule at Top Level} & \ \textsc{Contextual closure} \\
	    ((\l\var.\tm)\tmtwo)\tmthree \rtosl (\l\var.\tm \tmthree)\tmtwo,  \ \var \!\notin\! \fv{\tmthree}&
	    \quad\ \mctxp \tm \tosl \mctxp \tmtwo \textrm{~~~if } \tm \rtosl \tmtwo \\

	    \val ((\l\var.\tmthree)\tmtwo) \rtosr (\l\var.\val \tmthree)\tmtwo,  \ \var \!\notin\! \fv{\val} &
	    \quad\ \mctxp \tm \tosr \mctxp \tmtwo \textrm{~~~if } \tm \rtosr \tmtwo \\

	    \,\,(\l\var.\tm)\val\rtobv \tm\isub\var\val \qquad\qquad &
	    \quad\ \mctxp \tm \tobvm \mctxp \tmtwo \textrm{~~~if } \tm \rtobv \tmtwo \\\\

	    \mbox{Reductions} &\!\!\!\!\!\!\!\!\!\!\!\tosigm \, \defeq \, \tosl \!\cup \tosr , \ \ \toperm \, \defeq \, \tobvm \!\cup \tosigm \ 

    \end{array}
    $%}
  }
  \caption{\label{fig:shuffling-calculus} The Shuffling Calculus $\shufcalc$}
\end{figure}

The rewriting (shuffling) rules $\tosl$ and $\tosr$ unblock stuck $\beta$-redexes. 
For instance, consider the terms $\tm \defeq ((\la\vartwo \delta)(\varthree\varthree))\delta$ and $\tmtwo \defeq \delta((\la\vartwo \delta)(\varthree\varthree))$ where $\delta \defeq \la\var\var\var$ (as in Eq.~\refeq{unsolvable}, p.~\pageref{eq:unsolvable}): $\tm$ and $\tmtwo$ are $\betavm$-normal but $\tm \tosl (\la\vartwo \delta\delta)(\varthree\varthree) \tobvm  (\la\vartwo\delta\delta)(\varthree\varthree) \tobvm \dots$ and $\tmtwo \tosr (\la\vartwo \delta\delta)(\varthree\varthree) \tobvm (\la\var\delta\delta)(\varthree\varthree) \tobvm \dots$\,.

%The shuffling rules $\tosl$ and $\tosr$ can be added to $\tobv$ also in the strong calculus %(the reduction is still confluent)(confluence is preserved), studied in . 

%An interesting point is that 
The similar shuffling rules in CBN, better known as Regnier's $\sigma$-rules \cite{regnier94}, are \emph{contained} in CBN $\beta$-equivalence, while in Open (and Strong) CBV they are more interesting because they are not contained into (\ie they enrich) $\betav$-equivalence.

% In \cite{DBLP:conf/fossacs/CarraroG14}, various semantical results for Open CBV are proved using $\shufcalc$, in particular an adequacy result with respect to a relational model inspired by linear logic. 
% In \cite{DBLP:conf/fossacs/CarraroG14} it is shown also that $\shufcalc$ is suited to resource-sensitive analysis, in particular through the notion of Taylor expansion. 

The advantage of %presenting Open CBV via 
$\shufcalc$ is with respect to denotational investigations. In \cite{DBLP:conf/fossacs/CarraroG14}, $\shufcalc$ is indeed used to prove various semantical results in connection to linear logic, resource calculi, and the notion of Taylor expansion due to Ehrhard. 
In particular, in \cite{DBLP:conf/fossacs/CarraroG14} it has been proved the adequacy of $\shufcalc$ with respect to the relational model induced by linear logic: a by-product of our paper is the extension of this adequacy result to all incarnations of Open CBV.
% In particular, $\shufcalc$ is used to prove the adequacy of Open CBV with respect to the relational model induced by %(the relation model of) 
% linear logic. 
%Various semantical results for Open CBV have been proved using $\shufcalc$. In \cite{DBLP:conf/fossacs/CarraroG14} it is shown that $\shufcalc$ is suited for denotational investigations (at least in all models coming from linear logic) and resource-sensitive analysis, in particular through the notion of Taylor expansion. In particular, $\shufcalc$ is used in \cite{DBLP:conf/fossacs/CarraroG14} to prove the adequacy of Open CBV with respect to the relational model induced by the relation model of linear logic. 
%Standardization and normalizing strategies for $\shufcalc$ are studied in \cite{Guerrieri15,GuerrieriPR15}.
%	\emph{Pros}: statically conservative and it does not involve substitutions; it is deeply related to linear logic proof-nets. 
%	
%	\emph{Cons}: no harmony (\ie complex notion of normal form), and the rewriting of shuffling rules can be quite technical, also because of the extended notion of evaluation context.
The drawback of $\shufcalc$ is its technical rewriting theory. 
We summarize some properties of $\shufcalc$:%, most of them proved in \cite{DBLP:conf/fossacs/CarraroG14}: %(where $\toperm$ is actually denoted $\tow$):

\newcounter{prop:basic-shuffling} %new counter in order to use it in appendix
\addtocounter{prop:basic-shuffling}{\value{proposition}}
\begin{proposition}[Basic Properties of $\shufcalc$, \cite{DBLP:conf/fossacs/CarraroG14}]
\label{prop:basic-shuffling}
\hfill
\NoteProof{propappendix:basic-shuffling}

  \begin{varenumerate}
    \item\label{p:basic-shuffling-different} Let $\tm, \tmtwo, \tmthree \in \Lambda$. If $\tm \tobvm \tmtwo$ and $\tm \tosigm \tmthree$ then $\tmtwo \neq\tmthree$.
    \item\label{p:basic-shuffling-terminates} $\tosigm$ is strongly normalizing and (not strongly) confluent.
    \item\label{p:basic-shuffling-confluence} $\toshuf$ is (not strongly) confluent.
    \item\label{p:basic-shuffling-weak-strong-normalize} Let $\tm \in \Lambda$: $\tm$ is strongly $\shuf$-normalizable iff $\tm$ is $\shuf$-normalizable.
  \end{varenumerate}
\end{proposition}

In contrast to $\firecalc$ and $\vsubcalc$, $\shufcalc$ is not strongly confluent and not all $\shuf$-norma\-lizing derivations (if any) from a given term have the same length (consider, for instance, all $\shuf$-normalizing derivations from $ (\la\vartwo\varthree)(\delta(\varthree\varthree))\delta$). 
Nonetheless, normalization and strong normalization still coincide in $\shufcalc$ (\refpropp{basic-shuffling}{weak-strong-normalize}), and \refcoro{shuffling-number-steps} in \refsect{fireball-vsub} will show that the discrepancy is encapsulated inside the additional shuffling rules% $\tosl$ and $\tosr$
, since all $\shuf$-normalizing derivations (if any) from a given term have the same number of $\betashuf$-steps. 
% Standardization and normalizing strategies for $\shufcalc$ are studied in \cite{Guerrieri15,GuerrieriPR15}.
\doubt{\refpropp{basic-shuffling}{different} ensures that, given a $\vmsym$-derivation $\deriv$, the number of $\betavm$- and $\sigm$-steps are uniquely determined.}

% 	\medskip
% 	\noindent\emph{Open Call-by-Value 4: The Value Sequent Calculus $\vseqcalc$.} 
	\subsection{\texorpdfstring{Open Call-by-Value 4: The Value Sequent Calculus $\vseqcalc$}{Open Call-by-Value 4: The Value Sequent Calculus}}
	\label{subsect:vseq}
	% !TEX root = main.tex
\begin{figure}[t]
  \centering
  \scalebox{0.85}{
%     \ovalbox{
    $\begin{array}{c@{\hspace{1cm}}rlllllllll}
	    \mbox{Commands} & \cm,\cmtwo &\grameq & \comm{\val}{\cotm}\\
	    \mbox{Values} & \val,\valtwo & \grameq & \var \mid \la\var\cm\\
	    \mbox{Environments} & \cotm,\cotmtwo & \grameq & \stempty \mid \mutilde{\var}{\cm} \mid \stacker{\val}{\cotm}\\
% 	    \mbox{Command Evaluation Contexts} & \cmctx &\!\!\grameq\!\! & \ctxhole \mid \comm{\val}{\cotctxp{\mutilde\var\cmctx}}\\
% 	    \mbox{Environment Evaluation Contexts} & \cotctx &\!\!\grameq\!\! & \ctxhole \mid \stacker\val\cotctx \mid \mutilde\var{\comm\val\cotctx}\\
	    \mbox{Command Evaluation Contexts} & \cmctx &\grameq & \ctxhole \mid \cotctxp{\mutilde\var\cmctx}\\
	    \mbox{Environment Evaluation Contexts} & \cotctx &\grameq & \comm\val\ctxhole \mid \cotctxp{\stacker\val\ctxhole} \\\\
	    
	    %\cotmctx &\!\!\grameq\!\! & \mutilde{\var}{\cmctx} \mid \stacker{\val}{\cotmctx} \\\\
	      
	    \textsc{Rule at Top Level} & \multicolumn{3}{c}{\textsc{Contextual closure} }\\
% 	      (\comm{\lbar{\var}{\covar}{\cm}}{\stacker{\val}{\cotm}} \rtobvmu \cm\isubtwo{\var}{\val}{\covar}{\cotm} &
% 	    \multicolumn{3}{c}{\cmctxp \cm \tobvmu \cmctxp \cmtwo \textrm{~~~if } \cm \rtobvmu \cmtwo \ }\\
	    
	    \comm{\la\var\cm}{\stacker{\val}{\cotm}} \rtobvmu \comm\val{\append{(\mutilde\var\cm)}\cotm} &
	    \multicolumn{3}{c}{\cmctxp \cm \tobvmu \cmctxp \cmtwo \textrm{~~~if } \cm \rtobvmu \cmtwo \ }\\
	    
	    \comm{\val}{\mutilde\var\cm} \rtomut \cm\isub\var\val &
	    \multicolumn{3}{c}{\cmctxp \cm \tomut \cmctxp \cmtwo \textrm{~~~if } \cm \rtomut \cmtwo \ }\\\\
	    
	    \mbox{Reduction} & \multicolumn{3}{c}{\tolbarmut \, \defeq \, \tobvmu \!\cup \tomut}	    
    \end{array}$%}
  }
  \caption{\label{fig:lambdamu-calculus} The Value Sequent Calculus $\vseqcalc$}
\end{figure}
A more radical approach to the removal of stuck $\beta$-redexes is provided by what is here called the \emph{Value Sequent Calculus} $\vseqcalc$, defined in \reffig{lambdamu-calculus}. In $\vseqcalc$, it is the applicative structure of terms that is altered, by replacing the application constructor with more constructs, namely commands $\cm$ and environments $\cotm$. 
Morally, $\vseqcalc$ looks at a sequence of applications from the head, that is the value on the left of a command $\comm\val\cotm$ rather than from the tail as in natural deduction. 
In fact, $\vseqcalc$ is a handy presentation of the intuitionistic fragment of $\lambdamucalc$, that in turn is the CBV fragment of $\overline\lambda\mu\mutildesym$, a calculus obtained as the computational interpretation of a sequent calculus for classical logic. 
Both $\lambdamucalc$ and $\overline\lambda\mu\mutildesym$ are due to Curien and Herbelin \cite{DBLP:conf/icfp/CurienH00}, see %\eg  
\cite{DBLP:journals/toplas/AriolaBS09,DBLP:conf/ifipTCS/CurienM10} for further \mbox{investigations about these systems.}

A peculiar trait of the sequent calculus approach is the environment constructor $\mutilde\var\cm$, that is a binder for the free occurrences of $\var$ in $\cm$. It is often said that it is a sort of explicit substitution---we will see exactly in which sense, in \refsect{kernel}.

The change of the intuitionistic variant $\vseqcalc$ with respect to $\lambdamucalc$ is that $\vseqcalc$ does not need the syntactic category of co-variables $\covar$, as there can be only one of them, denoted here by $\stempty$. 
From a logical viewpoint, this is due to the fact that in intuitionistic sequent calculus the right-hand-side of $\vdash$ has exactly one formula, that is neither contraction nor weakening are allowed on the right.
Consequently, the binary abstraction $\lbar\var\covar \cm$ of $\lambdamucalc$ is replaced by a more traditional unary one $\la\var\cm$, and substitution on co-variables is replaced by a notion of \emph{appending of environ\-ments}, defined by mutual induction \mbox{on commands and environments as follows:}
\begin{align*}
   \append{\comm\val\cotmtwo}\cotm & \defeq \comm\val{\append\cotmtwo\cotm}
   & \append\stempty\cotm & \defeq \cotm
   \\
   \append{(\stacker\val\cotmtwo)}\cotm & \defeq \stacker\val {(\append\cotmtwo\cotm)}
   &
   \append{(\mutilde\var\cm)}\cotm 
   & \defeq \mutilde\vartwo (\append {\cm \isub\var\vartwo}\cotm) \text{ with }\vartwo \notin \fv{\cm} \cup \fv{\cotm}%\text{ with  $\vartwo$ fresh wrt $\cm$ and $\cotm$}
  \end{align*}
  Essentially, $\append\cm\cotm$ is a capture-avoiding substitution of $\cotm$ for the only occurrence of $\stempty$ in $\cm$ that is out of all abstractions, %such an occurrence , morally, is 
  standing for the output of the term.
  The append operation is used in the rewrite rule $\tobvmu$ of $\vseqcalc$ (\reffig{lambdamu-calculus}).
% 
%Values are closed under substitution: for all values $\val, \valtwo$, $\val\isub{\var}{\valtwo}$ is a value.
%Moreover, values are $\lambdabar$-, $\mut$- and $\lambdamucalc$-normal. 
% There are two rewriting rules, akin to the rules of $\vsubcalc$. Note that $\tobvmu$ uses the append operation. 
Strong CBV can be obtained by simply extending the grammar of evaluation contexts to commands under abstractions.

We will provide a translation from $\vsubcalc$ to $\vseqcalc$ that, beyond termination equivalence, will show that switching to a sequent calculus representation is equivalent to a transformation in administrative normal form \cite{DBLP:journals/lisp/SabryF93}.

The advantage of $\vseqcalc$ is that it avoids both rules at a distance and shuffling rules.
The drawback of $\vseqcalc$ is that, syntactically, it requires to step out of the $\l$-calculus. We will show in \refsect{kernel} how to reformulate it as a fragment of $\vsubcalc$, \ie in natural deduction. However, it will still be necessary to restrict the application constructor, thus preventing the natural way of writing terms.

The rewriting of $\vseqcalc$ is very well-behaved, in particular it is strongly confluent and every rewriting rule terminates separately.

\newcounter{prop:basic-lambdamu}
\addtocounter{prop:basic-lambdamu}{\value{proposition}}
\begin{proposition}[Basic properties of $\vseqcalc$]
\label{prop:basic-lambdamu}\hfill
\NoteProof{propAppendix:basic-lambdamu}
  \begin{varenumerate}
%     \item\label{p:basic-lambdamu-tobvmu-strong-confluence} $\tobvmu$ is strongly normalizing and strongly confluent.
%     \item\label{p:basic-lambdamu-tomut-strong-confluence} $\tomut$ is strongly normalizing and strongly confluent.
    \item\label{p:basic-lambdamu-tobvmu-strong-confluence}\label{p:basic-lambdamu-tomut-strong-confluence} $\tobvmu$ and $\tomut$ are strongly normalizing and strongly confluent (separately).
    \item\label{p:basic-lambdamu-strong-commutation} $\tobvmu$ and $\tomut$ strongly commute.
    \item\label{p:basic-lambdamu-strong-confluence} $\tolbarmut$ is strongly confluent, and all $\vseq$-normalizing derivations $\deriv$ from a command $\cm$ %or an environment $\env$ 
    (if any) have the same length $\size{\deriv}$, the same number $\sizemut{\deriv}$ of $\tilde{\mu}$-steps, and the same number $\sizelbar{\deriv}$ of $\lambdabar$-steps.
  \end{varenumerate}
\end{proposition}

 \subsection{Variations on a Theme}
\paragraph{Reducing Open to Closed Call-by-Value: Potential Valuability.}\label{par:potentially-valuable}
Potential valuability relates Na\"ive Open CBV to Closed CBV via a meta-level substitution closing open terms: a (possibly open) term $\tm$ is \emph{potentially valuable} if there is a substitution of (closed)\emph{\ignore{closed} values} for its free variables, for which it $\betav$-evaluates \ignore{(in Closed CBV)} to a (closed) \emph{value}.%
% \footnote{Actually, in the definition of potential valuability, it is equivalent to require that the substitution is made by closed or possibly open values, } 
\footnote{Potential valuability for Plotkin's CBV $\lambda$-calculus can be equivalently defined using weak or strong $\betav$-reduction: it is the same notion for Na\"ive Open and Strong CBV.} 
In Na\"ive Open CBV, potentially valuable terms do not coincide with normalizable terms because of premature $\betav$-normal forms---such as $\tm$ and $\tmtwo$ in Eq.~\refeq{unsolvable} at p.~\pageref{eq:unsolvable}---
which are not potentially valuable.

Paolini, Ronchi Della Rocca and, later, Pimentel \cite{DBLP:journals/ita/PaoliniR99,DBLP:conf/ictcs/Paolini01,parametricBook,paolini04itrs,paolini11tcs} gave several operational, logical, and semantical characterizations of potentially valuable terms in Na\"ive Open CBV.
In particular, in \cite{DBLP:journals/ita/PaoliniR99,parametricBook} %Paolini and Ronchi Della Rocca prove
it is proved that a term is potentially valuable in Plotkin's Na\"ive Open CBV iff its normalizable in $\firecalc$.

Potentially valuable terms can be defined for every incarnation of Open CBV: it is enough to update the notions of evaluation and values in the above definition to the considered calculus.
This has been done for $\shufcalc$ in \cite{DBLP:conf/fossacs/CarraroG14}, and for $\vsubcalc$ in \cite{DBLP:conf/flops/AccattoliP12}. 
For both calculi it has been proved that, in the weak setting, potentially valuable terms coincides with normalizable terms.
In \cite{GuerrieriPR15}, it has been proved that Plotkin's potentially valuable terms coincide with $\shuf$-potentially valuable terms (which coincide in turn with $\vmsym$-normalizable terms). 
Our paper makes a further step: proving that termination coincides for $\firecalc$, $\vsubcalc$, $\shufcalc$, and $\vseqcalc$ it implies that all their notions of potential valuability coincide with Plotkin's, \ie there is just one notion of potential valuability for Open (and Strong) CBV. 
%TODO: For Open CBV but also for Strong CBV, since potential valuability in Open and Strong CBV coincide.
	
%There are other reasons why potential valuability is of interest. 
%For instance, it is a key notion for characterizing solvability in Strong CBV \cite{DBLP:journals/ita/PaoliniR99,parametricBook,DBLP:conf/fossacs/CarraroG14,DBLP:conf/flops/AccattoliP12}.
%Moreover, in \cite{paolini04itrs,paolini11tcs} it has been proved that the class of strongly normalizable terms in Open CBN (aka the lazy $\l$-calculus, see \cite{DBLP:journals/iandc/AbramskyO93}) coincides with that of potentially valuable terms (and that of normalizable terms in $\firecalc$), providing an interesting connection between CBV and CBN.

\paragraph{Open CBV 5, 6, 7, \ldots} The literature contains many other calculi for CBV, usually presented for Strong CBV and easily adaptable to Open CBV. Some of them have $\letexp$-expressions (avatars of ES) and all of them have rules permuting constructors, therefore they lie somewhere in between $\vsubcalc$ and $\shufcalc$. Often, they have been developed for other purposes, usually to investigate the relationship with monad or CPS translations. Moggi's equational theory \cite{DBLP:conf/lics/Moggi89} is a classic standard of reference, known to coincide with that of Sabry and Felleisen \cite{DBLP:journals/lisp/SabryF93}, Sabry and Wadler \cite{DBLP:journals/toplas/SabryW97}, Dychoff and Lengrand \cite{DBLP:journals/logcom/DyckhoffL07}, Herbelin and Zimmerman  \cite{DBLP:conf/tlca/HerbelinZ09} and Maraist et al's $\lambda_{let}$ in \cite{DBLP:journals/tcs/MaraistOTW99}. In \cite{DBLP:conf/flops/AccattoliP12}, $\vsubcalc$ modulo $\eqstruct$ is shown to be termination equivalent to Herbelin and Zimmerman's calculus, and to strictly contain its equational theory, and thus Moggi's. At the level of rewriting these presentations of Open CBV are all more involved than those that we consider here. Their equivalence to our calculi can be shown along the lines of that of $\shufcalc$ with $\vsubcalc$.% (or the one in \cite{DBLP:conf/flops/AccattoliP12}).
 % !TEX root = main.tex
\section{\texorpdfstring{Quantitative Equivalence of $\firecalc$, $\vsubcalc$, and $\shufcalc$}{Quantitative Equivalence of lambda-fire, lambda-vsub, and lambda-shuf}}
\label{sect:fireball-vsub}

Here we show the equivalence with respect to termination of $\firecalc$, $\vsubcalc$, and $\shufcalc$, enriched with quantitative information on the number of steps.

\paragraph{On the Proof Technique.} We show that termination in $\vsubcalc$ implies termination in $\firecalc$ and $\shufcalc$ by studying simulations of $\firecalc$ and $\shufcalc$ into $\vsubcalc$. 
% For the converse implication, however, we do not rely on simulations of $\vsubcalc$ into $\firecalc$ and $\shufcalc$. 
To prove the converse implications we do not use inverse simulations.
Alternatively, we show that $\betaf$- and $\shuf$-normal forms are essentially projected into $\vsub$-normal forms, so that if evaluation terminates in $\firecalc$ or $\shufcalc$ then it also terminates on $\vsubcalc$. 

Such a simple technique works because in the systems under study \emph{normalization and strong normalization coincide}: if there is a normalizing derivation from a given term $\tm$ then there are no diverging derivations from $\tm$ (for $\vsubcalc$ and $\firecalc$ it follows from strong confluence, for $\shufcalc$ is given by \refpropp{basic-shuffling}{weak-strong-normalize}). 
  This fact is also the reason why the statements of our equivalences (forthcoming \refcor{equivalence-vsub-fire-termination} and \refcor{equivalence-vsub-shuf-termination}) address a single derivation from $\tm$ rather than considering \emph{all} derivations from $\tm$%, since all derivations from $\tm$ have the same number of the considered steps
. Moreover, for any calculus, all normalizing derivations from $\tm$ have the same number of steps (in $\shufcalc$ it holds for $\betashuf$-steps, see \refcor{shuffling-number-steps}), hence also the quantitative claims of \refcoro{equivalence-vsub-fire-termination} and \refcor{equivalence-vsub-shuf-termination} hold actually \mbox{for \emph{all} normalizing derivations from $\tm$.}

In both %cases
simulations, the structural equivalence $\eqstruct$ of $\vsubcalc$ plays a role.

%\paragraph{Simulating $\firecalc$ in $\vsubcalc$.}
\subsection{\texorpdfstring{Equivalence of  $\firecalc$ and $\vsubcalc$}{Equivalence of fireball-calculus and vsub-calculus}}
\label{subsect:fire-vsub}
A single $\betav$-step $(\la\var\tm)\val \tobv\! \tm\isub\var\val$ is simulated in $\vsubcalc$ by two steps: $(\la\var\tm)\val \allowbreak\tom \tm\esub\var\val \toe \tm\isub\var\val$, \ie a $\mult$-step that creates a ES, and a $\expo$-step that turns the ES into the meta-level substitution performed by the $\betav$-step.
The simulation of an inert step of $\firecalc$ is instead trickier, because in $\vsubcalc$ there is no rule to substitute an inert term, if it is not a variable. 
The idea is that an inert step $(\la\var\tm)\gconst \toin \tm\isub\var\gconst$ is simulated only by $(\la\var\tm)\gconst \tom \tm\esub\var\gconst$, \ie only by the $\mult$-step that creates the ES, and such a ES will never be fired---so the simulation is up to the unfolding of substitutions containing inert terms (defined right next). Everything works because of the key property of inert terms: they are normal and their substitution cannot create redexes, so it is useless to substitute them. 

  The \emph{unfolding} of a $\vsub$-term $\tm$ is the term $\unf{\tm}$ obtained from $\tm$ by turning ES into %ordinary 
  meta-level substitutions; it is defined by\doubt{ induction on $\tm \in \Lambda_\vsub$ as follows}:
  \begin{align*}
    \unf{\var} &\defeq \var			& \unf{(\tm\tmtwo)} &\defeq \unf{\tm} \, \unf{\tmtwo} &
    \unf{(\la\var\tm)} &\defeq \la\var\unf{\tm} & \unf{(\tm\esub\var\tmtwo)} &\defeq \unf{\tm}\isub\var{\unf{\tmtwo}}
  \end{align*}
For all $\tm, \tmtwo \in \vsubterms$, $\tm \eqstruct \tmtwo$ implies $\unf{\tm} = \unf{\tmtwo}$.
Also, $\unf{\tm} = \tm$ iff $\tm \in \Lambda$.

In the simulation we are going to show, structural equivalence $\eqstruct$ plays a role. It is used to \emph{clean} the $\vsub$-terms (with ES) obtained by simulation, putting them in a canonical form where ES do not appear among other constructors.

A $\vsub$-term is \emph{\proper} if it has the form $\tmtwo \esub{\var_1}{\gconst_1} \dots \esub{\var_n}{\gconst_n}$ (with  $n \in \nat$), $\tmtwo\in \Lambda$ %(\ie without ES) 
is called the \emph{body}, and $\gconst_1,\ldots,\gconst_n \in \Lambda$ are inert terms.
Clearly, any term (as it is without ES) is \proper.
We first show how to simulate a single fireball step.

\newcounter{l:proj-tof-on-vsub} %new counter in order to use it in appendix
\addtocounter{l:proj-tof-on-vsub}{\value{lemma}}
\begin{lemma}[Simulation of a $\betaf$-Step in $\vsubcalc$]
\label{l:proj-tof-on-vsub}
  Let $\tm, \tmtwo \in \Lambda$.
\NoteProof{lappendix:proj-tof-on-vsub}
  \begin{varenumerate}
    \item\label{p:proj-tof-on-vsub-tobv} If $\tm \tobabs \tmtwo$ then $\tm \tom\toeabs \tmtwo$.

    \item\label{p:proj-tof-on-vsub-toin} If $\tm \toin\tmtwo$ then $\tm \tom\eqstruct \tmthree$, 
  with $\tmthree \!\in\! \vsubterms$ \proper and $\unf\tmthree = \tmtwo$.
  \end{varenumerate}
\end{lemma}

%Unfortunately, 
%It is not possible to 
We cannot simulate derivations by iterating \reflemma{proj-tof-on-vsub}, because the starting term $\tm$ has no ES but the simulation of inert steps introduces ES. 
%Therefore, 
Hence, we have to generalize %the statement 
\reflemma{proj-tof-on-vsub} up to the unfolding of ES. 
In general, unfolding ES is a dangerous operation with respect to (non-)termination, as it may erase a diverging subterm (\eg $\tm \defeq \var\esub\vartwo{\delta\delta}$ is $\vsub$-divergent and $\unf\tm = \var$ is normal). In our case, however, the simulation produces \proper $\vsub$-terms, %and 
so the unfolding is safe since it can erase only inert terms %, that 
and cannot create, erase, \mbox{nor carry redexes.}

By means of a technical lemma \withproofs{in the appendix}\withoutproofs{(see the appendix in \cite{ourTechReport})} we obtain: 

\newcounter{l:proj-via-unfold} %new counter in order to use it in appendix
\addtocounter{l:proj-via-unfold}{\value{lemma}}
\begin{lemma}[Projection of a $\betaf$-Step on $\tovsub$ via Unfolding]
\label{l:proj-via-unfold}
  Let% 
\NoteProof{lappendix:proj-via-unfold}
  $\tm$ be a \proper $\vsub$-term and $\tmtwo$ be a term. 
  \begin{varenumerate}
    \item \label{p:proj-via-unfold-tobv} If $\unf{\tm} \tobabs\! \tmtwo$ then $\tm \tom\toeabs \tmthree$, 
    %where $\tmthree$ is a \proper $\vsub$-term such that $\unf\tmthree = \tmtwo$.
    with $\tmthree \!\in\! \vsubterms$ \proper and $\unf\tmthree \!= \tmtwo$.
    
    \item\label{p:proj-via-unfold-toin} If $\unf{\tm} \toin\! \tmtwo$ then $\tm \tom\eqstruct \tmthree$, 
%     where $\tmthree$ is a \proper $\vsub$-term such that $\unf\tmthree = \tmtwo$.
    with $\tmthree \!\in\! \vsubterms$ \proper and $\unf\tmthree \!= \tmtwo$.

 \end{varenumerate}
\end{lemma}
\noindent

Via \reflemma{proj-via-unfold} we can now simulate whole derivations (in forthcoming \refthm{sim-f-into-vsubeq}). 

\paragraph{Simulation and Normal Forms.}
The next step towards the equivalence is to relate normal forms in $\firecalc$ (aka fireballs) to those in $\vsubcalc$. The relationship is not perfect, since the simulation does not directly map the former to the latter---we have to work a little bit more. 
First of all, let us characterize the terms in $\vsubcalc$ obtained by projecting normalizing derivations (that always produce a fireball).

\newcounter{l:normal-anti-unfold} %new counter in order to use it in appendix
\addtocounter{l:normal-anti-unfold}{\value{lemma}}
\begin{lemma}
\label{l:normal-anti-unfold}  
  Let%
\NoteProof{lappendix:normal-anti-unfold}
  $\tm$ be a \proper $\vsub$-term. 
  If $\unf{\tm}$ is a fireball, then $\tm$ is $\set{\msym,\expoabs}$-normal and its body is a fireball.
\end{lemma}

Now, a $\set{\msym,\expoabs}$-normal form $\tm$ morally is $\vsub$-normal, as $\toevar$ terminates (\refpropp{basic-value-substitution}{tom-toe-terminates}) and it cannot create $\set{\msym,\expoabs}$-redexes%, so that a $\expovar$-normalization produces a $\vsub$-normal form
. %Termination is already known to hold (\refpropp{basic-value-substitution}{tom-toe-terminates}). 
The part about creations %, instead, 
is better expressed as a postponement property.

\newcounter{l:toevar-post}
\addtocounter{l:toevar-post}{\value{lemma}}
\begin{lemma}[Linear Postponement of $\toevar$]
\label{l:toevar-post}
  Let%
\NoteProof{lappendix:toevar-post} 
  $\tm, \tmtwo \in \vsubterms$. 
  If $\deriv \colon \tm\tovsub^*\tmtwo$ then $\derivtwo \colon \tm\Rew{\msym,\expoabs}^*\!\toevar^*\!\tmtwo$ with $\sizevsub\derivtwo = \sizevsub\deriv$, $\sizem\derivtwo = \sizem\deriv$, \mbox{$\sizee{\derivtwo} \!= \sizee{\deriv}$ and $\sizeeabs\derivtwo \!\geq \sizeeabs\deriv$.} 
\end{lemma}

%From \reflemma{normal-anti-unfold} it follows that a clean $\vsub$-normal form is a fireball followed by ES with\withproofs{ compound} inert terms\withoutproofs{ that are not variables}%
%. Therefore---as previously announced---we obtained a nice new description of normal forms for $\vsubcalc$, inherited from $\firecalc$. 

The next theorem puts all the pieces together.

\newcounter{thm:sim-f-into-vsubeq} %new counter in order to use it in appendix
\addtocounter{thm:sim-f-into-vsubeq}{\value{theorem}}
\begin{theorem}[Quantitative Simulation of $\firecalc$ in $\vsubcalc$]
\label{thm:sim-f-into-vsubeq}
  Let%
\NoteProof{thmappendix:sim-f-into-vsubeq}
  $\tm, \tmtwo \in \Lambda$.
  If $\deriv \colon \tm \tof^*\! \tmtwo$ then %there exists $\derivtwo \colon \tm \tovsub^* \eqstruct \tmthree \in \vsubterms$ such that
  there are $\tmthree, \tmfour \!\in\! \vsubterms$ and $\derivtwo \colon \tm \tovsub^*\! \tmfour$ such that
  
  \begin{varenumerate}
  \item\label{p:sim-f-into-vsubeq-qual} \emph{Qualitative Relationship}: %$\tmtwo = \unf{\tmthree}$ and $\tmthree$ is \proper;
  $\tmfour \eqstruct \tmthree$, $\tmtwo = \unf{\tmthree} = \unf{\tmfour}$ and $\tmthree$ is \proper;

  \item \emph{Quantitative Relationship}: \label{p:sim-f-into-vsubeq-quant} 
    \begin{varenumerate}
    \item \label{p:sim-f-into-vsubeq-quant-mult} \emph{Multiplicative Steps:} $\sizef{\deriv} = \sizem{\derivtwo}$;
    \item \label{p:sim-f-into-vsubeq-quant-exp} \emph{Exponential (Abstraction) Steps:} $\sizebabs{\deriv} = \sizeeabs{\derivtwo} = \sizee{\derivtwo}$.
    \end{varenumerate}
    
    \item \emph{Normal Forms}: if $\tmtwo$ is $\betaf$-normal then there exists $\derivfour\colon\tmfour \toevar^* \tmfive$ such that $\tmfive$ is a $\vsub$-normal form and $\sizeevar{\derivfour} \leq \sizem{\derivtwo} - \sizeeabs{\derivtwo}$.
    \end{varenumerate}
\end{theorem}

\newcounter{coro:equivalence-vsub-fire-termination}
\addtocounter{coro:equivalence-vsub-fire-termination}{\value{corollary}}
\begin{corollary}[Linear Termination Equivalence of $\vsubcalc$ and $\firecalc$]
\label{coro:equivalence-vsub-fire-termination}
  Let%
\NoteProof{coroAppendix:equivalence-vsub-fire-termination}
  $\tm \in \Lambda$. 
  There is a $\betaf$-normalizing derivation $\deriv$ from $\tm$ iff there is a $\vsub$-normalizing derivation $\derivtwo$ from $\tm$. 
  Moreover, $\sizef{\deriv} \leq \sizevsub{\derivtwo} \leq 2\sizef{\deriv}$, \ie they are linearly related.
\end{corollary}

% Note that 
% The statement of \refcoro{equivalence-vsub-fire-termination} is stronger than it may look at first sight, because by strong confluence in both $\firecalc$ and $\vsubcalc$, given a term $\tm$, if there is a normalizing derivation from $\tm$ then there are no diverging derivations from $\tm$, and \emph{all} normalizing derivations from $\tm$ have the same length (\refpropp{basic-fireball}{number-steps} and \refpropp{basic-value-substitution}{number-steps}).

The number of $\betaf$-steps in $\firecalc$ is a reasonable cost model for Open CBV \cite{fireballs}. 
Our result implies that %the same holds for the number of $\mult$-steps in $\vsubcalc$ since
also \emph{the number of $\mult$-steps in $\vsubcalc$ is a reasonable cost model}, since the number of $\mult$-steps is \emph{exactly} the number of $\betaf$-steps %in $\firecalc$ (\refthmp{sim-f-into-vsubeq}{quant})
. This fact is quite surprising: in $\firecalc$ arguments of $\betaf$-redexes are required to be fireballs, while for $\mult$-redexes there are no restrictions on arguments, and yet in any normalizing derivation %to normal form 
 their number coincide. Note, moreover, that $\expo$-steps are linear in $\mult$-steps, but only because the initial term has no ES: in general, this is not true.

%\paragraph{Simulating $\shufcalc$ in $\vsubcalc$.}
\subsection{\texorpdfstring{Equivalence of $\shufcalc$ and $\vsubcalc$}{Equivalence of shuffling-calculus and vsub-calculus}}
\label{subsect:shuf-vsub}
A derivation $\deriv\colon \tm \tovm^*\tmtwo$ in $\shufcalc$ is simulated via a projection on multiplicative normal forms in $\vsubcalc$, \ie as a derivation $\mnf\tm \tovsubeq^* \mnf\tmtwo$ (for any $\vsub$-term $\tm$, its multiplicative and exponential normal forms, denoted by $\mnf{\tm}$ and $\enf{\tm}$ respectively, exist and are unique by \refprop{basic-value-substitution}).
Indeed, a $\betashuf$-step of $\shufcalc$ is simulated in $\vsubcalc$ by a $\expo$-step followed by some $\mult$-steps to reach the $\mult$-normal form.
%The basic idea is that 
Shuffling rules $\tosigm$ of $\shufcalc$ are simulated by structural equivalence $\eqstruct$ in $\vsubcalc$: applying $\mnf\cdot$ to $((\la\var\tm)\tmtwo)\tmthree \tosl (\la\var(\tm \tmthree))\tmtwo$ we obtain exactly an instance of the axiom $\tostructapl$ defining $\eqstruct$: $\mnf\tm\esub\var{\mnf\tmtwo}\mnf\tmthree \tostructapl (\mnf\tm\mnf\tmthree)\esub\var{\mnf\tmtwo}$ (with the side conditions matching exactly). Similarly, $\tosr$ projects to $\tostructapr$ or $\tostructes$ \mbox{(depending on whether $\val$ in $\tosr$ is a variable or an abstraction). Therefore,}

\newcounter{l:tovm-mproj} %new counter in order to use it in appendix
\addtocounter{l:tovm-mproj}{\value{lemma}}
\begin{lemma}[Projecting a $\vmsym$-Step on $\tovsubeq$ via $\mult$-NF]%[Permutation steps project on structural equivalence]
% \label{l:tovm-mproj-on-eqstruct} % \reflemmap{tosig-mproj-on-eqstruct}{mdev}
\label{l:tovm-mproj}
  Let 
\NoteProof{lappendix:tovm-mproj}
  $\tm, \tmtwo \!\in\! \Lambda$.
  
  \begin{varenumerate}
    \item %\emph{Projecting a $\betashuf$-step on $\eqstruct$:} 
    \label{p:tovm-mproj-sigm-on-eqstruct} If $\tm \tosigm \tmtwo$ then $\mnf\tm \eqstruct \mnf\tmtwo$.
    
    \item %\emph{Projecting a $\sigm$-step on $\eqstruct$:} 
    \label{p:tovm-mproj-betavm-on-tovsub} If $\tm \tobvm \tmtwo$ then $\mnf\tm \toe\tom^* \mnf\tmtwo$.
  \end{varenumerate}
\end{lemma}

In contrast to the simulation of $\firecalc$ in $\vsubcalc$, here the projection of a single step can be extended to derivations without problems, obtaining that the number of $\betashuf$-steps in $\shufcalc$ matches exactly the number of $\esym$-steps in $\vsubcalc$. Additionally, we  apply the postponement of $\eqstruct$ (\reflemmap{eqstruct-post-and-term}{globpost}), factoring out the use of $\eqstruct$ (\ie of shuffling rules) without affecting the number of $\expo$-steps.
%So, via \reflemma{tovm-mproj} %we can now simulate whole derivations. 

To obtain the termination equivalence we also need to study normal forms. Luckily, the case of $\shufcalc$ is simpler than that of $\firecalc$, as next lemma shows.

\newcounter{l:normal-vsub-shuffling} %new counter in order to use it in appendix
\addtocounter{l:normal-vsub-shuffling}{\value{lemma}}
\begin{lemma}[Projection %of $\shuf$-Normal Forms on $\vsub$-Normal 
Preserves Normal Forms]
\label{l:normal-vsub-shuffling}
  Let% 
\NoteProof{lappendix:normal-vsub-shuffling}
  $\tm \in \Lambda$. 
  If $\tm$ is $\vmsym$-normal then $\mnf{\tm}$ is $\vsub$-normal.
\end{lemma}

The next theorem puts all the pieces together (for any $\shuf$-deri\-vation $\deriv$, $\sizebshuf{\deriv}$ is the number of $\betashuf$-steps in $\deriv$: this notion is well defined by \refpropp{basic-shuffling}{different}).

\newcounter{thm:sim-shuf-into-vsubeq} %new counter in order to use it in appendix
\addtocounter{thm:sim-shuf-into-vsubeq}{\value{theorem}}
\begin{theorem}[Quantitative Simulation of $\shufcalc$ in $\vsubcalc$]
\label{thm:sim-shuf-into-vsubeq}
  Let%
\NoteProof{thmappendix:sim-shuf-into-vsubeq}
  $\tm, \tmtwo \in \Lambda$.
  If $\deriv \colon \tm \toshuf^* \tmtwo$ then there are $\tmthree \in \vsubterms$ and $\derivtwo \colon \tm \tovsub^* \tmthree$ such that
  
  \begin{varenumerate}
    \item\label{p:sim-shuf-into-vsubeq-qual} \emph{Qualitative Relationship}: $\tmthree \eqstruct \mnf{\tmtwo}$;

    \item\label{p:sim-shuf-into-vsubeq-quant} \emph{Quantitative Relationship (Exponential Steps)}: $\sizebshuf{\deriv} = \sizee{\derivtwo}$;
    \item\label{p:sim-shuf-into-vsubeq-normal} \emph{Normal Form}: if $\tmtwo$ is $\shuf$-normal then $\tmthree$ and $\mnf{\tmtwo}$ are $\vsub$-normal.
  \end{varenumerate}
\end{theorem}

\newcounter{coro:equivalence-vsub-shuf-termination}
\addtocounter{coro:equivalence-vsub-shuf-termination}{\value{corollary}}
\begin{corollary}[Termination Equivalence of $\vsubcalc$ and $\shufcalc$]
\label{coro:equivalence-vsub-shuf-termination}
  Let%
\NoteProof{coroappendix:equivalence-vsub-shuf-termination}
  $\tm \in \Lambda$. 
  There is a $\vmsym$-normalizing derivation $\deriv$ from $\tm$ iff there is a $\vsub$-normalizing derivation $\derivtwo$ from $\tm$. 
  Moreover, $\sizebshuf{\deriv} = \sizee{\derivtwo}$.
\end{corollary}
% As for \refcoro{equivalence-vsub-fire-termination}, the claim of \refcoro{equivalence-vsub-shuf-termination} is stronger than it seems, since for both $\vsubcalc$ and $\shufcalc$, given a term $\tm$, if there is a normalizing derivation from $\tm$ then there are no diverging derivations from $\tm$ (for $\vsubcalc$ it follows from strong \mbox{confluence, for $\shufcalc$ is given by \refpropp{basic-shuffling}{weak-strong-normalize}).}

The obtained quantitative equivalence has an interesting corollary that shows some light on why $\shufcalc$ is not strongly confluent. 
Our simulation maps $\betashuf$-steps in $\shufcalc$ to exponential steps in $\vsubcalc$, that are strongly confluent, and thus in equal number in all normalizing derivations (if any) from a given term. Therefore,

\newcounter{coro:shuffling-number-steps} %new counter in order to use it in appendix
\addtocounter{coro:shuffling-number-steps}{\value{corollary}}
\begin{corollary}[Number of $\betavm$-Steps is Invariant]
\label{coro:shuffling-number-steps}
  All%
\NoteProof{coroappendix:shuffling-number-steps}
  $\shuf$-nor\-ma\-lizing derivations from $\tm \in \Lambda$ (if any) have the same number of $\betashuf$-steps.
\end{corollary}

Said differently, in $\shufcalc$ normalizing derivations may have different lengths but the difference is encapsulated inside the shuffling rules $\tosl$ and $\tosr$.

%Analogous results hold for $\vsubcalc$ (\refpropp{basic-value-substitution}{number-steps}) and $\firecalc$ (\refpropp{basic-fireball}{number-steps}), thanks to their good rewriting theory.
%  In a way, the quantitative simulation of $\shufcalc$ in $\vsubcalc$ (\refthm{sim-shuf-into-vsubeq}) ``imposes the good behavior'' of $\vsubcalc$  (\refpropp{basic-value-substitution}{number-steps}) 
%  on $\shufcalc$. 
% % The existence of a quantitative invariant in 
% $\shuf$-normalizing derivations is not obvious, indeed, as $\shufcalc$ is not strongly confluent.
%already remarked at p.~\pageref{prop:basic-shuffling}.

% For what concerns 
Concerning the cost model, things are subtler for $\shufcalc$. Note that the relationship between $\shufcalc$ and $\vsubcalc$ uses the number of $\expo$-steps, while the cost model (inherited from $\firecalc$) is the number of $\mult$-steps. Do $\esym$-steps provide a reasonable cost model? 
% They probably do not, 
Probably not, because there is a family of terms that evaluate in exponentially more $\mult$-steps than $\expo$-steps. Details are left to a longer version.
 %\input{04_-_Equational_Theories}
 % !TEX root = main.tex
\section{\texorpdfstring{Quantitative Equivalence of $\vsubcalc$ and $\vseqcalc$, via $\vsubkcalc$}{Quantitative Equivalence of lambda-vsub and lambda-vseq, via lambda-vsubk}}
\label{sect:kernel}

The quantitative termination equivalence of $\vsubcalc$ and $\vseqcalc$ is shown in two steps: first, we identify a sub-calculus $\vsubkcalc$ of $\vsubcalc$ equivalent to the whole of $\vsubcalc$, and then show that $\vsubkcalc$ and $\vseqcalc$ are equivalent (actually isomorphic). Both steps reuse the  technique of \refsect{fireball-vsub}, \ie simulation plus study of normal forms.

\subsection{\texorpdfstring{Equivalence of $\vsubkcalc$ and $\vsubcalc$}{Equivalence of vsubk and vsub}} 
\label{subsect:vsubk-vsub}
The kernel $\vsubkcalc$ of $\vsubcalc$ is the sublanguage of $\vsubcalc$ obtained by  replacing the application constructor $\tm \tmtwo$ with the restricted form $\tm\val$ where the right subterm can only be a value $\val$---\ie, $\vsubkcalc$ is the language of so-called \emph{administrative normal forms} \cite{DBLP:journals/lisp/SabryF93} of $\vsubcalc$. The rewriting rules are the same of $\vsubcalc$%, just restricted to the set of $\vsubk$-terms
. 
It is easy to see that $\vsubkcalc$ is stable by $\vsub$-reduction. For lack of space, more details about $\vsubkcalc$ \withoutproofs{are in the appendix of \cite{ourTechReport}}\withproofs{have been moved to Appendix \ref{s:kernel-proofs} (page \pageref{s:kernel-proofs})}. 

The translation $\vsubtoker{(\cdot)}$ of $\vsubcalc$ into $\vsubkcalc$, which simply places the argument of an application into an ES, is defined by \mbox{(note that $\fv{\tm} = \fv{\vsubtoker\tm\!}$ for all $\tm \!\in\! \vsubterms$):}
\begin{align*}
  \vsubtoker{\var} &\defeq \var  &
    \vsubtoker{(\tm\tmtwo)} &\defeq (\vsubtoker\tm \var)\esub\var{\vsubtoker\tmtwo} \ \textup{ where } \var \textup{ is fresh } \\% &\textup{where $\var$ is fresh wrt $\vsubtoker\tm$ and $\vsubtoker\tmtwo$}\\
  \vsubtoker{(\la\var\tm)} &\defeq \la\var\vsubtoker\tm &
  \vsubtoker{\tm \esub\var\tmtwo } &\defeq \vsubtoker\tm\esub\var{\vsubtoker\tmtwo}
\end{align*}

%Note that $\fv{\tm} = \fv{\vsubtoker{\tm}}$.

% \newcounter{l:vsubtoker-sub}
% \addtocounter{l:vsubtoker-sub}{\value{lemma}}
% \begin{lemma}[Substitution]
%  \label{l:vsubtoker-sub}
%   For 
% \NoteProof{lappendix:vsubtoker-sub}
%   any $\vsub$-term $\tm$ and any $\vsub$-value $\val$, one has
%  $\vsubtoker{\tm\isub\var\val} = \vsubtoker\tm\isub\var{\vsubtoker\val}$.
% \end{lemma}

\newcounter{l:vsub-to-ker-sim}
\addtocounter{l:vsub-to-ker-sim}{\value{lemma}}
\begin{lemma}[Simulation] 
\label{l:vsub-to-ker-sim} % \reflemmap{vsub-to-ker-sim}{mult}
  Let 
\NoteProof{lappendix:vsub-to-ker-sim}
  $\tm, \tmtwo \in \vsubterms$.
\begin{varenumerate}
	\item \emph{Multiplicative}: \label{p:vsub-to-ker-sim-mult}
	if $\tm \tom \tmtwo$ then $\vsubtoker\tm \tom\toevar\eqstruct\vsubtoker\tmtwo$;
	\item \emph{Exponential}: \label{p:vsub-to-ker-sim-exp}
	if $\tm \toeabs \tmtwo$ then $\vsubtoker\tm \toeabs \vsubtoker\tmtwo$, and if $\tm \toevar\tmtwo$ then $\vsubtoker\tm \toevar \vsubtoker\tmtwo$.
	\item \emph{Structural Equivalence}: $\tm \eqstruct\tmtwo$ implies $\vsubtoker\tm \eqstruct \vsubtoker\tmtwo$.
\end{varenumerate}
\end{lemma}

The translation of a $\vsub$-normal form is not $\vsubk$-normal (\eg $\vsubtoker{(\var\vartwo)} = (\var\varthree)\esub\varthree\vartwo$) but a further exponential normalization provides a $\vsubk$-normal form.

\newcounter{thm:sim-vsub-into-vsubk}
\addtocounter{thm:sim-vsub-into-vsubk}{\value{theorem}}
\begin{theorem}[Quantitative Simulation of $\vsubcalc$ in $\vsubkcalc$]
\label{thm:sim-vsub-into-vsubk} % \refthmp{sim-vsub-into-vsubk}{normal}
  Let%
\NoteProof{thmappendix:sim-vsub-into-vsubk}
  $\tm, \tmtwo \in \vsubterms$.
  If $\deriv \colon \tm \tovsub^* \tmtwo$ then there are $\tmthree \in \vsubkterms$ and $\derivtwo \colon \vsubtoker\tm \tovsubk^* \tmthree$ such that
  
  \begin{varenumerate}
    \item\label{p:sim-vsub-into-vsubk-qual} \emph{Qualitative Relationship}: $\tmthree \eqstruct \vsubtoker\tmtwo$;

    \item\label{p:sim-vsub-into-vsubk-quant} \emph{Quantitative Relationship}: 
    \begin{enumerate}
    	\item \emph{Multiplicative Steps}: $\sizem{\derivtwo} = \sizem{\deriv}$;
	\item \emph{Exponential Steps}: $\sizeeabs\derivtwo = \sizeeabs\deriv$ and $\sizeevar\derivtwo = \sizeevar\deriv + \sizem\deriv$;
    \end{enumerate}
    \item\label{p:sim-vsub-into-vsubk-normal} \emph{Normal Form}: if $\tmtwo$ is $\vsub$-normal then $\tmthree$ is $\mult$-normal \mbox{and $\enf{\!\tmthree\!}$ is $\vsubk$\!-normal.}
  \end{varenumerate}
\end{theorem}

Unfortunately, the length of the exponential normalization in \refthmp{sim-vsub-into-vsubk}{normal} cannot be easily bounded, forbidding a precise quantitative equivalence. Note however that turning from $\vsubcalc$ to its kernel $\vsubkcalc$ does not change the number of multiplicative steps: the transformation preserves the cost model.

\newcounter{coro:equivalence-vsub-vsubk-termination}
\addtocounter{coro:equivalence-vsub-vsubk-termination}{\value{corollary}}

\begin{corollary}[Termination and Cost Equivalence of $\vsubcalc$ and $\vsubkcalc$]
\label{coro:equivalence-vsub-vsubk-termination}
  Let
\NoteProof{coroappendix:equivalence-vsub-vsubk-termination}
  $\tm \in \vsubterms$. 
  There exists a $\vsub$-normalizing derivation $\deriv$ from $\tm$ iff there exists a $\vsubk$-normalizing derivation $\derivtwo$ from $\vsubtoker\tm$. 
  Moreover, $\sizem{\deriv} = \sizem{\derivtwo}$.
\end{corollary}

\subsection{\texorpdfstring{Equivalence of $\vsubkcalc$ and $\vseqcalc$}{Equivalence of vsubk and vseq}}
\label{subsect:vsubk-vseq}
The translation $\tolbarmu{\cdot}$ of $\vsubkcalc$ into $\vseqcalc$ relies on an auxiliary translation $\tolbarmuv{(\cdot)}$ of values and it is defined as follows:
\[ \begin{array}{rcl@{\hspace{1cm}}rcl@{\hspace{1cm}}rcl}
      \tolbarmuv\var &\defeq  &\var
      &
      \tolbarmuv{(\la{\var}\tm)} &\defeq  &\la\var\tolbarmu\tm\\
      \tolbarmu\val &\defeq &\comm{\val}\stempty       
      &   
   \tolbarmu{\tm\val} &\defeq &\append{\tolbarmu\tm}{(\stacker{\tolbarmuv\val}\stempty)} 
   	  
   &\tolbarmu{\tm \esub\var\tmtwo } &\defeq &\append{\tolbarmu\tmtwo}{\mutilde\var \tolbarmu\tm }
 \end{array}\]
Note the subtle mapping of ES to $\mutildesym$: ES correspond to appendings of $\mutildesym$ to the output of the term $\tmtwo$ to be substituted, and not of the term $\tm$ where to substitute.

It is not hard to see that $\vsubkcalc$ and $\vseqcalc$ are actually isomorphic, where the converse translation $\cotransl{(\cdot)}$, that maps values and commands to terms, and environments to evaluation contexts, is given by:
 \[\begin{array}{rcl@{\hspace{1cm}}rcl@{\hspace{1cm}}rcl}
  \cotransl\var & \defeq & \var 
  &	
  \cotransl \stempty & \defeq & \ctxhole
  &
  \cotransl{\comm\val\cotm} & \defeq & \cotransl\cotm\ctxholep{\cotransl\val}
  \\
  \cotransl{(\la\var\cm)} & \defeq & \la\var\cotransl\cm 
  &
  \cotransl {(\stacker\val\cotm)} & \defeq & \cotransl\cotm\ctxholep{\ctxhole \cotransl\val}     
   &
   \cotransl{(\mutilde\var\cm)} & \defeq & \cotransl\cm \esub\var\ctxhole  
 \end{array}\]

 For the sake of uniformity, we follow the same structure of the other weaker equivalences (\ie simulation plus mapping of normal forms, here working smoothly) rather than proving the isomorphism formally. The simulation maps multiplicative steps to $\lambdabar$ steps, whose number, then, is a reasonable cost model for $\vseqcalc$.

\newcounter{l:lbarmut-simulates-vsubk}
\addtocounter{l:lbarmut-simulates-vsubk}{\value{lemma}}
\begin{lemma}[Simulation of $\tovsubk$ by $\tovseq$]
\label{l:lbarmut-simulates-vsubk}
  Let 
  \NoteProof{lappendix:lbarmut-simulates-vsubk}
%   $\tm, \tmtwo \in \vsubkterms$.
  $\tm$ and $\tmtwo$ be $\vsubk$-terms.
  
  \begin{varenumerate}
    \item \emph{Multiplicative}: if $\tm \tom \tmtwo$ then 
    $\tolbarmu\tm \tobvmu \tolbarmu\tmtwo$.
    \item \emph{Exponential}: if $\tm \toe \tmtwo$ then 
    $\tolbarmu\tm \tomut \tolbarmu\tmtwo$.
  \end{varenumerate}
\end{lemma}

\newcounter{thm:sim-vsubk-into-vseq}
\addtocounter{thm:sim-vsubk-into-vseq}{\value{theorem}}
\begin{theorem}[Quantitative Simulation of $\vsubkcalc$ in $\vseqcalc$]
\label{thm:sim-vsubk-into-vseq}
  Let %
\NoteProof{thmappendix:sim-vsubk-into-vseq}
%   $\tm, \tmtwo \in \vsubkterms$.
  $\tm$ and $ \tmtwo $ be $ \vsubk$-terms.
  If $\deriv \colon \tm \tovsubk^* \tmtwo$ then there is $\derivtwo \colon \tolbarmu\tm \tolbarmut^* \tolbarmu\tmtwo$ such that
  
  \begin{varenumerate}
    	\item\label{p:sim-vsubk-into-vseq-mult} \emph{Multiplicative Steps}: $\sizem{\deriv} = \sizehole{\lambdabar}{\derivtwo}$ (the number $\lambdabar$-steps in $\derivtwo$);
	\item\label{p:sim-vsubk-into-vseq-exp} \emph{Exponential Steps}: $\sizee\deriv = \sizehole{\mutildesym}\derivtwo$ (the number $\mut$-steps in $\derivtwo$), so $\sizevsubk\deriv = \sizevseq\derivtwo$;
    \item\label{p:sim-vsubk-into-vseq-normal} \emph{Normal Form}: if $\tmtwo$ is $\vsubk$-normal then $\tolbarmu\tmtwo$ is $\vseq$-normal.
  \end{varenumerate}
\end{theorem}

\newcounter{coro:equivalence-vsubk-vseq-termination}
\addtocounter{coro:equivalence-vsubk-vseq-termination}{\value{corollary}}
\begin{corollary}[Linear Termination Equivalence of $\vsubkcalc$ and $\vseqcalc$]
\label{coro:equivalence-vsubk-vseq-termination}
  Let
\NoteProof{coroappendix:equivalence-vsubk-vseq-termination}
%   $\tm \!\in \vsubkterms$. 
  $\tm $ be a $ \vsubk$-term. 
  There is a $\vsubk$-norma\-lizing derivation $\deriv$ from $\tm$ iff there is a $\vseqsym$-norma\-lizing derivation $\derivtwo$ from $\tolbarmu\tm$. 
  Moreover, $\sizevsubk\deriv = \sizevseq\derivtwo$, $\sizee\deriv = \sizemut\derivtwo$ and $\sizem\deriv = \sizelbar\derivtwo$.
\end{corollary}

%The number of $\lambdabar$ steps can then be taken as a reasonable cost model for $\vseqcalc$.

\paragraph{Structural Equivalence for $\vseqcalc$.} The equivalence of $\vsubcalc$ and $\vsubkcalc$ relies on the structural equivalence $\eqstruct$ of $\vsubcalc$, so it is natural to wonder how does $\eqstruct$ look on $\vseqcalc$.
The structural equivalence $\seqbar$ of $\vseqcalc$ is defined as the closure by evaluation contexts of the following axiom
\begin{align*}
\cotctxp{\mutilde\var{\cotctxtwop{\mutilde\vartwo{\cm}}}}
&\seqbar_{\mut\mut} \cotctxtwop{\mutilde\vartwo{\cotctxp{\mutilde\var{\cm}}}}
&\textup{where } \var \notin \fv\cotctxtwo \textup{ and } \vartwo \notin \fv\cotctx.
\end{align*}

% \newcounter{l:seq-to-seqbar}
% \addtocounter{l:seq-to-seqbar}{\value{lemma}}
% \begin{lemma}
% \label{l:seq-to-seqbar} % \reflemmap{seq-to-seqbar}{seq}
%  Let 
% \NoteProof{lappendix:seq-to-seqbar} 
%  $\tm$ and $\tmtwo$ be $\vsubk$-terms. If $\tm \eqstruct_{\essym,\aplsym} \tmp$ then $\tolbarmu\tm = \tolbarmu\tmp$, and if $\tm \eqstruct_{\comsym} \tmp$ then $\tolbarmu\tm \seqbar \tolbarmu\tmp$.
% \end{lemma}

As expected, $\seqbar$ has, with respect to $\vseqcalc$, all the properties of $\eqstruct$ (see \reflemma{eqstruct-post-and-term}). 
They are formally stated in \withproofs{Appendix \ref{s:kernel-proofs}, \refprops{seq-to-seqbar}{seqbar-post-and-term}.}\withoutproofs{the appendix of \cite{ourTechReport}, for lack of space.}
%
%We %use $\tovsubeq$ for 
%set $\tolbarmuteq \, \defeq \, \seqbar\tovsub\seqbar$ (\ie for all commands $\cm, \cmfour$: $\cm \tolbarmuteq \cmfour$ iff $\cm\seqbar\cmtwo\tolbarmut \cmthree\seqbar \cmfour$ for some commands $\cmtwo, \cmthree$). The notation $\tolbarmuteq^+$  keeps its usual meaning, while $\tolbarmuteq^*$ stands for $\seqbar \cup \, \tolbarmuteq^+$, \ie a $\vseqcalc$-derivation of length zero can apply $\seqbar$ and is not just the identity.
%As $\seqbar$ is reflexive, $\tolbarmut \, \subsetneq \, \tolbarmuteq$. 
%
\newcounter{l:seqbar-post-and-term}
\addtocounter{l:seqbar-post-and-term}{\value{lemma}}
\section{Conclusions and Future Work}
\label{conclusions}

This paper proposes Open CBV as a setting halfway between Closed CBV, the simple framework used to model programming languages such as OCaml, and Strong CBV, the less simple setting underling proof assistants such as Coq. 
Open CBV is a good balance: its rewriting theory is simple---in particular it is strongly confluent, as the one of Closed CBV---and it can be iterated under abstractions to recover Strong CBV, which is not possible with Closed CBV.

We compared four representative calculi for Open CBV, developed with different motivations, and showed that they share the same qualitative (termination/divergence) and quantitative (number of steps) properties with respect to termination. 
% Therefore, 
Therefore, they can be considered as different incarnations of the same immaterial setting, justifying the slogan \emph{Open  CBV}.
%The literature contains many other proposals for CBV that lie somewhere in between these four incarnations, and our results can easily be adapted to them. 
%This justifies our slogan that \emph{there is just one Open CBV}, hoping that it could avoid further proliferations of CBV calculi, at least for the open setting.

The qualitative equivalences carry semantical consequences: \emph{the adequacy of relational semantics} for the shuffling calculus proved in \cite{DBLP:conf/fossacs/CarraroG14} actually gives a semantic (and type-theoretical, since the relational model can be seen as a non-idempotent intersection type system) characterization of normalizable terms for Open CBV, \ie it extends to the other three calculi.
Similarly, the notion of \emph{potential valuability} for Plotkin's CBV $\lambda$-calculus, well-studied in \cite{DBLP:journals/ita/PaoliniR99,DBLP:conf/ictcs/Paolini01,parametricBook,paolini04itrs,paolini11tcs} and recalled at the end of  \refsect{bird}, becomes a robust notion characterizing the same terms in Open (and Strong) CBV.

Quantitatively, we showed that in three out of four calculi for Open CBV, namely $\firecalc$, $\vsubcalc$ and $\vseqcalc$, evaluation takes exactly the same number of $\betaf$-steps, $\mult$-steps and $\lambdabar$-steps, respectively. Since such a number is known to be a reasonable time cost model for $\firecalc$  \cite{fireballs}, the cost model lifts to $\vsubcalc$ and $\vseqcalc$, showing that the cost model is robust, \ie incarnation-independent. 
For the shuffling calculus $\shufcalc$ we obtain a weaker quantitative relationship that does not allow to transfer the cost model. The $\betashuf$-steps in $\shufcalc$, indeed, match $\expo$-steps in $\vsubcalc$, but not $\mult$-steps. Unfortunately, the two quantities are not necessarily polynomially related, since there is a family of terms that evaluate in exponentially more $\mult$-steps than $\expo$-steps (details are left to a longer version). Consequently, $\shufcalc$ is an incarnation more apt to semantical investigations rather than complexity analyses.

%Technically, the relationships between the four calculi are all established by means of $\vsubcalc$ and its structural equivalence $\eqstruct$.

	 \paragraph{Future Work.}
This paper is just the first step towards a new, finer understanding of CBV. We plan to purse at the least the following research directions:

\begin{varenumerate}
  \item \emph{Equational Theories.} The four incarnations are termination equivalent but their rewriting rules do not induce the same equational theory. In particular, $\firecalc$ equates more than the others, and probably too much because its theory is not a congruence, \ie it is not stable by context closure. The goal is to establish the relationships between the theories and understand how to smooth the calculi as to make them both \emph{equational} and \emph{termination} equivalent.

  \item \emph{Abstract Machines.} %Another direction for further work is the study of abstract machines for Open CBV. 
  Accattoli and Sacerdoti Coen introduce in \cite{fireballs} \emph{reasonable} abstract machines for Open CBV, that is, implementation schemas whose overhead is proven to be polynomial, and even linear. Such machines are quite complex, especially the linear one. Starting from a fine analysis of the overhead, we are currently working on a simpler approach providing cost equivalent but much simpler abstract machines.

  \item \emph{From Open CBV to Strong CBV.} We repeatedly said that Strong CBV can be seen as an iteration of Open CBV under abstractions. 
  This is strictly true for $\vsubcalc$, $\shufcalc$, and $\vseqcalc$, for which the simulations studied here lift to the strong setting.   On the contrary, the definition of a good strong $\firecalc$ is a subtle open issue. The natural candidate, indeed, is not confluent (but enjoys uniqueness of normal forms) and normalizes more terms than the other calculi for Strong CBV.
  Another delicate point is the design and the analysis of abstract machines for Strong CBV, of which there are no examples in the literature (both \gregoire and Leroy's \cite{DBLP:conf/icfp/GregoireL02} and Accattoli and Sacerdoti Coen's \cite{fireballs} study machines for Open CBV only).
  
  \item \emph{Open Bisimulations.} In \cite{DBLP:conf/lics/Lassen05} Lassen studies open (or normal form) bisimulations for CBV. He points out that his bisimilarity is not fully abstract with respect to contextual equivalence, and his counterexamples are all based on stuck $\beta$-redexes in Na\"ive Open CBV. An interesting research direction is to recast his study in Open CBV and see whether full abstraction holds or not.
\end{varenumerate}

% (they have exactly the same sets of normalizable and diverging terms)
% (the lengths of the evaluations of a given term in the various calculi are precisely related)

% \input{05_-_How_to_Stop_Worrying_and_Love_the_Bomb}
% \input{06_-_Easy_GLAMOUr}
% \input{07_-_Minimality_of_the_Cost_Model}

\paragraph{Acknowledgment}
This work has been partially supported by the A*MIDEX project %(n$^{\rm o}$
ANR-11-IDEX-0001-02 %) 
funded by the ``Investissements d'Avenir''  French
Government program, managed by the French National Research Agency (ANR).

\phantomsection
\addcontentsline{toc}{section}{References}
\bibliographystyle{splncs03}
% \bibliography{\macrospath/biblio}
\bibliography{biblio} %For Arxiv

\withproofs{
\newpage
\appendix
\chapter*{Technical Appendix}
\newcounter{appendix}
\setcounter{appendix}{\value{theorem}}%inizializzo il contatore appendix al valore dell'ultimo teorema/lemma/proposizione/definizione nella sezione precedente

% !TEX root = main.tex
\section{Rewriting Theory: Definitions, Notations, and Basic Results}
\label{app:rewriting}

Given %a set $I$ and 
a binary relation $\to_\mathsf{r}$ on a set $I$, the reflexive-transitive (resp.~reflexive; transitive; reflexive-transitive and symmetric) closure of $\to_\mathsf{r}$ is denoted by $\to^*$ (resp.~$\to_\mathsf{r}^=$; $\to_\mathsf{r}^+$; $\simeq_\mathsf{r}$).
The transpose of $\to_\mathsf{r}$ is denoted by \!\!$\lRew{\mathsf{r}}$\!\!.
A (\emph{$\mathsf{r}$-})\emph{derivation $\deriv$ from $\tm$ to~$\tmtwo$}, denoted by $\deriv \colon \tm \to_\mathsf{r}^* \tmtwo$, is a finite sequence $(\tm_i)_{0 \leq i \leq n}$ of elements of $I$ (with $n \in \nat$) s.t.~$\tm = \tm_0$, $\tmtwo = \tm_n$ and $\tm_i \to_\mathsf{r} \tm_{i+1}$ for all $1 \leq i < n$;
% we set $\size\deriv = n$ (so, $\size\deriv$ is the length, \ie the number of $\to$-steps, of $\deriv$); also, if $\to_1 \, \subseteq \, \to$, $\size{\deriv}_1$ is the number of $\to_1$-steps in $\deriv$.

The \emph{number of $\mathsf{r}$-steps of a derivation $\deriv$}, \ie its \emph{length}, is denoted by $\size{\deriv}_\mathsf{r} \defeq n$, or simply $\size\deriv$. If $\to_\mathsf{r} \, = \, \to_1 \!\cup \to_2$ with $\to_1 \!\cap \to_2 \, = \emptyset$, $\size{\deriv}_i$ is the number of $\to_i$-steps in $\deriv$, for $i = 1,2$.
We say that:
\begin{itemize}
  \item $\tm \!\in\! I$ is \emph{$\mathsf{r}$-normal} or a \emph{$\mathsf{r}$-normal form} if 
%   there is no $\tmtwo \!\in\! I$ s.t.~$\tm \to_\mathsf{r} \tmtwo$; 
  $\tm \not\to_\mathsf{r} \tmtwo$ for all $\tmtwo \!\in\! I$;
  %\item 
  $\tmtwo \in I$ is a \emph{$\mathsf{r}$-normal form of} $\tm %\in I
  $ if $\tmtwo$ is $\mathsf{r}$-normal and $\tm \to_\mathsf{r}^* \tmtwo$; 
  \item $\tm \in I$ is \emph{$\mathsf{r}$-normalizable} if there is a $\mathsf{r}$-normal $\tmtwo \in I$ s.t. $\tm \to_\mathsf{r}^* \tmtwo$;
  $\tm$ is \emph{strongly $\mathsf{r}$-normalizable} if there is no infinite sequence $(\tm_i)_{i \in \nat}$ s.t. $\tm_0 = \tm$ and $\tm_i \to_\mathsf{r} \tm_{i+1}$;
% \end{itemize}
% \begin{itemize}
  \item a $\mathsf{r}$-derivation $\deriv \colon \tm \to_\mathsf{r}^* \tmtwo$ is (\emph{$\mathsf{r}$-})\emph{normalizing} if $\tmtwo$ is $\mathsf{r}$-normal;
  \item $\to_\mathsf{r}$ is \emph{strongly normalizing} if all $\tm \!\in\! I$ is strongly $\mathsf{r}$-normalizable;
%   \item $\to$ is \emph{confluent} if, given $a \to^* b$ and $a \to^* c$, there is $a' \in I$ such that $b \to^* a'$ and $c \to^* a'$;
  \item $\to_\mathsf{r}$ is \emph{strongly confluent} if, for all $\tm, \tmtwo, \tmthree \in\! I$ s.t.~$\tmthree \lRew{\mathsf{r}} \tm \to_\mathsf{r} \tmtwo$ and $\tmtwo \neq \tmthree$, there is $\tmfour \in I$ s.t. $\tmthree \to_\mathsf{r} \tmfour \lRew{\mathsf{r}} \tmtwo$;
  $\to_\mathsf{r}$ is \emph{confluent} \mbox{if $\to_\mathsf{r}^*\!$ is strongly confluent.}
\end{itemize}

Let $\to_1, \to_2 \, \subseteq I \times I$. 
Composition of relations is denoted by juxtaposition: for instance, $\tm \to_1\to_2 \tmtwo$ means that there is $\tmthree \in I$ s.t. $\tm \to_1 \tmthree \to_2 \tmtwo$; for any $n \in \nat$, $\tm \to_1^n \tmtwo$ means that there is a $\to_1$-derivation with length $n$ ($\tm = \tmtwo$ for $n = 0$).
We say that $\to_1$ and $\to_2$ \emph{strongly commute} if, for any $\tm, \tmtwo, \tmthree \in I$ s.t. $\tmtwo \lRew{1} \tm \to_2 \tmthree$, one has $\tmtwo \neq \tmthree$ and there is $\tmfour \in I$ s.t. $\tmtwo \to_2 \tmfour \lRew{1} \tmthree$.
Note that if $\to_1$ and $\to_2$ strongly commute and $\to \, = \, \to_1 \!\cup \to_2$, then for any derivation $\deriv \colon \tm \to^* \tmtwo$ the sizes $\size{\deriv}_1$ and $\size{\deriv}_2$ are uniquely determined.

The following proposition collects some basic and well-known results of rewriting theory.

\begin{proposition}\label{prop:basic-confluence}
  Let $\to_\mathsf{r}$ be a binary relation on a set $I$.
  \begin{enumerate}
    \item\label{p:basic-confluence-unique-normal} If $\to_\mathsf{r}$ is confluent then: 
    \hfill
    \begin{enumerate}
      \item every $\mathsf{r}$-normalizable term has a unique $\mathsf{r}$-normal form;
      \item\label{p:basic-confluence-equivalent} for all $\tm, \tmtwo \in I$, $\tm \simeq_\mathsf{r} \tmtwo$ iff there is $\tmthree \in I$ s.t.~$\tm \to_\mathsf{r}^* \tmthree \lRewn{\mathsf{r}} \tmtwo$.
    \end{enumerate}
    \item\label{p:basic-confluence-strong} If $\to_\mathsf{r}$ is strongly confluent then $\to_\mathsf{r}$ is confluent and, for any $\tm \in I$, one has:
    \begin{enumerate}
      \item\label{p:basic-confluence-strong-length} all normalizing $\mathsf{r}$-derivations from $\tm$ have the same length;
      \item\label{p:basic-confluence-strong-normalizable} $\tm$ is strongly $\mathsf{r}$-normalizable if and only if $\tm$ is $\mathsf{r}$-normalizable.
    \end{enumerate}
  \end{enumerate}
\end{proposition}

As all incarnations of Open CBV we consider are confluent, the use of \refpropp{basic-confluence}{unique-normal} is left implicit.

For $\firecalc$ and $\vsubcalc$, we use \refpropp{basic-confluence}{strong} and the following more informative version of Hindley--Rosen Lemma, whose proof is just a more accurate reading of the proof in \cite[Prop.\,3.3.5.(i)]{Barendregt84}:

\begin{lemma}[Strong Hindley--Rosen]\label{l:hindley-rosen}
  Let $\to \, = \, \to_1 \cup \to_2$ be a binary relation on a set $I$ s.t. $\to_1$ and $\to_2$ are strongly confluent.
  If $\to_1$ and $\to_2$ strongly com\-mute, then $\to$ is strongly confluent and, 
% for all $\tm \in I$ and normalizing reduction sequences $\deriv \colon \tm \to^* \tmtwo$ and $\derivtwo \colon \tm \to^* \tmtwo$, \mbox{one has $\size\deriv = \size\derivtwo$, $\size{\deriv}_1 = \size{\derivtwo}_1$ and $\size{\deriv}_2 = \size{\derivtwo}_2$.}
 for any $\tm \in I$ and any normalizing derivations $\deriv$ and $\derivtwo$ from $\tm$, one has $\size\deriv = \size\derivtwo$, $\size{\deriv}_1 = \size{\derivtwo}_1$ and $\size{\deriv}_2 = \size{\derivtwo}_2$.
\end{lemma}
% !TEX root = main.tex
\section{Omitted Proofs}
\label{app:omitted}

\subsection{Proofs of Section~\ref{sect:bird} (Incarnations of Open Call-by-Value)}

% !TEX root = main.tex
\paragraph{Na\"ive Open CBV: Plotkin's Calculus $\plotcalc$}

\begin{remark}\label{rmk:value-betav-normal}
  Since $\tobv$ does not reduce under $\l$'s, any value is $\betav$-normal, and so $\betavar$-normal and $\betaabs$-normal, as $\tobvar, \tobabs \, \subseteq \, \tobv$.
\end{remark}

\setcounter{propositionAppendix}{\value{prop:basic-plotkin-strong-confluence}}
\begin{propositionAppendix}
\label{propappendix:basic-plotkin-strong-confluence}
  $\tobvar$, %
  \NoteState{prop:basic-plotkin}
  $\tobabs$ and $\tobv$ are strongly confluent.
\end{propositionAppendix}

\begin{proof}\hfill
%   \begin{enumerate}
%     \item 
    We prove that $\tobv$ is strongly confluent.
    The proofs that $\tobvar$ and $\tobabs$ are strongly confluent are perfectly analogous.
    
    So, we prove, by induction on $\tm$, that if $\tm \tobv \tmtwo$ and $\tm \tobv \tmthree$ with $\tmtwo \neq \tmthree$, then there exists $\tm'$ such that $\tmtwo \tobv \tm'$ and $\tmthree \tobv \tm'$.
    
%     By \refrmk{values-betav-normal}, $\tm$ is not a value \ie $\tm$ is an application.
    Observe that neither $\tm \tobv\tmtwo$ nor $\tm \tobv \tmthree$ can be a step at the root: 
    indeed, if $\tm \defeq (\la\var\tmfour)\val \tobv \tmfour\isub\var\val \eqdef \tmtwo$ and $\tm \tobv \tmthree$ (or if $\tm \defeq (\la\var\tmfour)\val \tobv \tmfour\isub\var\val \eqdef \tmthree$ and $\tm \tobv \tmtwo$), then $\tmtwo = \tmthree$ since $\la\var\tmfour$ and $\val$ are $\betav$-normal by \refrmk{value-betav-normal};
    but this contradicts the hypothesis $\tmtwo \neq \tmthree$.
    So, according to the definition of $\tm \tobv \tmtwo$ and $\tm \tobv \tmthree$, there are only four cases.
    \begin{itemize}      
      \item \emph{Application Left for $\tm \tobv \tmtwo$ and $\tm \tobv \tmthree$}, \ie $\tm = \tmfour\tmfive \tobv \tmsix\tmfive = \tmtwo$ and $\tm = \tmfour\tmfive \tobv \tmseven\tmfive = \tmthree$ with $\tmfour \tobv \tmsix$ and $\tmfour \tobv \tmseven$. 
      By the hypothesis $\tmtwo \neq \tmthree$ it follows that $\tmsix \neq \tmseven$.
      By \ih, there exists $\tmfour'$ such that $\tmsix \tobv \tmfour'$ and $\tmseven \tobv \tmfour'$.
      So, setting $\tm' = \tmfour'\tmfive$, one has $\tmtwo = \tmsix\tmfive \tobv \tm'$ and $\tmthree = \tmseven\tmfive \tobv \tm'$.

      \item \emph{Application Right for $\tm \tobv \tmtwo$ and $\tm \tobv \tmthree$}%
%       . This case is analogous to the previos one.
      , \ie $\tm = \tmfour\tmfive \tobv \tmfour\tmsix = \tmtwo$ and $\tm = \tmfour\tmfive \tobv \tmfour\tmseven = \tmthree$ with $\tmfive \tobv \tmsix$ and $\tmfive \tobv \tmseven$. 
      From the hypothesis $\tmtwo \neq \tmthree$ it follows that $\tmsix \neq \tmseven$.
      By \ih, there exists $\tmfive'$ such that $\tmsix \tobv \tmfive'$ and $\tmseven \tobv \tmfive'$.
      So, setting $\tm' = \tmfour\tmfive'$, one has $\tmtwo = \tmfour\tmsix \tobv \tm'$ and $\tmthree = \tmfour\tmseven \tobv \tm'$.

      \item \emph{Application Left for $\tm \tobv \tmtwo$} and \emph{Application Right for $\tm \tobv \tmthree$}, \ie $\tm = \tmfour\tmfive \tobv \tmsix\tmfive = \tmtwo$ and $\tm = \tmfour\tmfive \tobv \tmfour\tmseven = \tmthree$ with $\tmfour \tobv \tmsix$ and $\tmfive \tobv \tmseven$. 
      So, setting $\tm' = \tmsix\tmseven$, one has $\tmtwo = \tmsix\tmfive \tobv \tm'$ and $\tmthree = \tmfour\tmseven \tobv \tm'$.

      \item \emph{Application Right for $\tm \tobv \tmtwo$} and \emph{Application Left for $\tm \tobv \tmthree$}%
%       . This case is analogous to the previous one.
      , \ie $\tm = \tmfour\tmfive \tobv \tmfour\tmsix = \tmtwo$ and $\tm = \tmfour\tmfive \tobv \tmseven\tmfive = \tmthree$ with $\tmfive \tobv \tmsix$ and $\tmfour \tobv \tmseven$. 
      So, setting $\tm' = \tmseven\tmsix$, one has $\tmtwo = \tmfour\tmsix \tobv \tm'$ and $\tmthree = \tmseven\tmfive \tobv \tm'$.
      \qedhere
    \end{itemize}
%     \item It follows immediately from \refpropp{basic-confluence}{strong-normalizable} and \refpropp{basic-plotkin}{strong-confluence}.
%    
%     \item It follows immediately from \refpropp{basic-confluence}{strong-length} and \refpropp{basic-plotkin}{strong-confluence}.
%     \qedhere
%   \end{enumerate}
\end{proof}

\subsubsection{Proofs of Subsection \ref{subsect:fireball} (\texorpdfstring{Open CBV 1: the Fireball Calculus $\firecalc$}{Open CBV 1: the Fireball Calculus})}

\begin{remark}
  Inert terms can be equivalently defined as $\gconst \grameq  \var \mid \gconst\fire$\withproofs{---such a definition is used in the proofs here
%\withoutproofs{ of \cite{ourTechReport}}\withproofs{ 
}.
\end{remark}
Inert terms that are not variables are referred to as \emph{compound inert terms}.

\begin{lemma}[Values and inert terms are $\betaf$-normal]
\label{l:fnormal}
\hfill
  \begin{enumerate}
    \item\label{p:fnormal-value} Every value is $\betaf$-normal.
%     \item\label{p:fnormal-inert-nobetaredex} Every inert term contains no $\beta$-redexes (\ie subterms of the form $(\la\var\tm)\tmtwo$).%FALSE! Under a $\lambda$ a term might have a $\beta$-redex
    \item\label{p:fnormal-inert} Every inert term is $\betaf$-normal.
  \end{enumerate}
\end{lemma}

\begin{proof}\hfill
  \begin{enumerate}
    \item Immediate, since $\tof$ does not reduce under $\l$'s.
    
    \item By induction on the definition of inert term $\gconst$.
  
    \begin{itemize}
%       \item If $\gconst = \var\val$ then $\val$ is $\betaf$-normal by \reflemmap{fnormal}{value}, hence $\gconst$ is $\betaf$-normal.
%       
%       \item If $\gconst = \var\gconsttwo$ then $\gconsttwo$ is $\betaf$-normal by \ih, hence $\gconst$ is $\betaf$-normal.
      \item If $\gconst = \var$ then $\gconst$ is obviously $\betaf$-normal.
%       \item If $\gconst = \gconsttwo\!\val$ then $\gconsttwo\!$ and $\val$ are $\betaf$-normal by \ih and \reflemmap{fnormal}{value} re\-spectively, besides $\gconsttwo$\! is not an abstraction, \mbox{so $\gconst$ is $\betaf$-normal.}
      \item If $\gconst = \gconsttwo\la\var\tm$ then $\gconsttwo$ and $\la\var\tm$ are $\betaf$-normal by \ih and \reflemmap{fnormal}{value} re\-spectively, besides $\gconsttwo$\! is not an abstraction, \mbox{so $\gconst$ is $\betaf$-normal.}

      \item Finally, if $\gconst = \gconsttwo\gconstthree$ then $\gconsttwo$ and $\gconstthree$ are $\betaf$-normal by \ih, moreover $\gconsttwo$ is not an abstraction, hence $\gconst$ is $\betaf$-normal.
      \qedhere
    \end{itemize}
  \end{enumerate}
\end{proof}

\setcounter{propositionAppendix}{\value{prop:open-harmony}}
\begin{propositionAppendix}[Open Harmony]\label{propappendix:open-harmony}
%   A (possibly open) term $\tm$ is a fireball iff it is a $\betaf$-normal form.
  Let $\tm \in \Lambda$:
  \NoteState{prop:open-harmony}
  $\tm$ is $\betaf$-normal iff $\tm$ is a fireball.
\end{propositionAppendix}

\begin{proof}\hfill
  \begin{description}
    \item [$\Rightarrow$:] Proof by induction on $\tm \in \Lambda$.
    If $\tm$ is a value then $\tm$ is a fireball.
    
    Otherwise $\tm = \tmtwo\tmthree$ for some terms $\tmtwo$ and $\tmthree$.
    Since $\tm$ is $\betaf$-normal, then $\tmtwo$ and $ \tmthree$ are $\betaf$-normal, and either $\tmtwo$ is not an abstraction or $\tmthree$ is not a fireball.
    By induction hypothesis, $\tmtwo$ and $\tmthree$ are fireballs. 
    Summing up, $\tmtwo$ is either a variable or an inert term, and $\tmthree$ is a fireball, therefore $\tm = \tmtwo\tmthree$ is an inert term and hence a fireball.
    
    \item [$\Leftarrow$:] By hypothesis, $\tm$ is either a value or an inert term. 
    If $\tm$ is a value, then it is $\betaf$-normal by \reflemmap{fnormal}{value}. 
    Otherwise $\tm$ is an inert term and then it is $\betaf$-normal by \reflemmap{fnormal}{inert}.
    \qedhere
  \end{description}
\end{proof}

\begin{lemma}
\label{l:toin-different}
  For every $\tm, \tm' \in \Lambda$, if $\tm \toin \tm'$ then $\tm \neq \tm'$. 
\end{lemma}

\begin{proof}
  By induction on $\tm \in \Lambda$.
  According to the definition of $\tm \toin \tm'$, there are three cases.
  \begin{itemize}
    \item \emph{Step at the root}, \ie $\tm = (\la\var\tmtwo)\gconst \toin \tmtwo\isub\var\gconst = \tm'$: then, since $\gconst$ is not an abstraction, necessarily $\tm = (\la\var\tmtwo)\gconst \neq \tmtwo\isub\var\gconst = \tm'$.
    
    \item \emph{Application Left}, \ie $\tm = \tmtwo\tmthree \toin \tmtwo'\tmthree = \tm'$ with $\tmtwo \toin \tmtwo'$: by \ih, $\tmtwo \neq \tmtwo'$ and hence $\tm = \tmtwo\tmthree \neq \tmtwo'\tmthree = \tm'$.

    \item \emph{Application Right}, \ie $\tm = \tmtwo\tmthree \toin \tmtwo\tmthree' = \tm'$ with $\tmthree \toin \tmthree'$: by \ih, $\tmthree \neq \tmthree'$ and hence $\tm = \tmtwo\tmthree \neq \tmtwo\tmthree' = \tm'$.
    \qedhere
  \end{itemize}
\end{proof}

\setcounter{propositionAppendix}{\value{prop:basic-fireball}} 
\begin{propositionAppendix}[Basic Properties of $\firecalc$]
\label{propappendix:basic-fireball}\hfill
  \NoteState{prop:basic-fireball}
  \begin{enumerate}
    \item\label{pappendix:basic-fireball-toin-strong-normalization}\label{pappendix:basic-fireball-toin-strong-confluence} $\toin$ is strongly normalizing and strongly confluent.
    \item\label{pappendix:basic-fireball-tobv-toin-strong-commutation} $\tobabs$ and $\toin$ strongly commute.	  
%       \item\label{pappendix:basic-fireball-weak-strong-normalize} $\tm$ is strongly $\betaf$-normalizable if and only if $\tm$ is $\betaf$-normalizable.
    \item\label{pappendix:basic-fireball-strong-confluent}\label{pappendix:basic-fireball-number-steps} $\tof$ is strongly confluent, and all $\betaf$-normalizing derivations $\deriv$ from $\tm \in \Lambda$ (if any) have the same length $\sizef{\deriv}$, the same number $\sizebabs{\deriv}$ of $\betaabs$-steps, and the same number $\sizein{\deriv}$ of $\betain$-steps.
  \end{enumerate}
\end{propositionAppendix}

\begin{proof}\hfill
  \begin{enumerate}
    \item Strong normalization of $\toin$ follows from general termination properties in the ordinary (\ie pure, strong, and call-by-name) $\l$-calculus, as we now explain. 
    Since $\betain$-steps do not substitute abstractions, they can only cause creations of type 1, according to \levy's classification of creations of $\beta$-redexes \cite{thesislevy}. 
    Then %$\toin$-derivations 
    $\betain$-derivations can be seen as special cases of \emph{$\mult$-developments} (see Accattoli, B., Kesner, D., \emph{The Permutative $\l$-Calculus}. In: LPAR. pp. 23-36, 2012), in turn a special case of more famous  \emph{superdevelopments}, \ie reduction sequences reducing only (residuals of) redexes in the original term plus creations of type 1 ($\mult$-developments) or type 1 and 2 (superdevelopments). 
    Both $\mult$-developments and superdevelopments always terminate. 
    Therefore, $\toin$ is strongly normalizing.
    
    \smallskip
    Now, we prove that $\toin$ is strongly confluent, that is if $\tm \toin \tmtwo$ and $\tm \toin \tmthree$ with $\tmtwo \neq \tmthree$, then there exists $\tm' \in \Lambda$ such that $\tmtwo \toin \tm'$ and $\tmthree \toin \tm'$. The proof is by induction on $\tm \in \Lambda$.
    
%     By \refrmk{values-betav-normal}, $\tm$ is not a value \ie $\tm$ is an application.
    Observe that neither $\tm \toin\tmtwo$ nor $\tm \toin \tmthree$ can be a step at the root: indeed, if $\tm \defeq (\la\var\tmfour)\gconst \rtoin \tmfour\isub\var\gconst \defeq \tmtwo$ and $\tm \toin \tmthree$ (or if $\tm \defeq (\la\var\tmfour)\gconst \rtoin \tmfour\isub\var\gconst \eqdef \tmthree$ and $\tm \toin \tmtwo$), then $\tmtwo = \tmthree$ since $\la\var\tmfour$ and $\gconst$ are $\betain$-normal by \reflemmasps{fnormal}{value}{inert} (as $\toin \, \subseteq \, \tof$);
    but this contradicts the hypothesis $\tmtwo \neq \tmthree$.
    So, according to the definition of $\tm \toin \tmtwo$ and $\tm \toin \tmthree$, there are only four cases.
    \begin{itemize}      
      \item \emph{Application Left for $\tm \toin \tmtwo$ and $\tm \toin \tmthree$}, \ie $\tm = \tmfour\tmfive \toin \tmsix\tmfive = \tmtwo$ and $\tm = \tmfour\tmfive \toin \tmseven\tmfive = \tmthree$ with $\tmfour \toin \tmsix$ and $\tmfour \toin \tmseven$. 
      By the hypothesis $\tmtwo \neq \tmthree$ it follows that $\tmsix \neq \tmseven$.
      By \ih, there exists $\tmfour'$ such that $\tmsix \toin \tmfour'$ and $\tmseven \toin \tmfour'$.
      So, setting $\tm' = \tmfour'\tmfive$, one has $\tmtwo = \tmsix\tmfive \toin \tm'$ and $\tmthree = \tmseven\tmfive \toin \tm'$.

      \item \emph{Application Right for $\tm \toin \tmtwo$ and $\tm \toin \tmthree$}, \ie $\tm = \tmfour\tmfive \toin \tmfour\tmsix = \tmtwo$ and $\tm = \tmfour\tmfive \toin \tmfour\tmseven = \tmthree$ with $\tmfive \toin \tmsix$ and $\tmfive \toin \tmseven$. 
      By the hypothesis $\tmtwo \neq \tmthree$ it follows that $\tmsix \neq \tmseven$.
      By \ih, there exists $\tmfive'$ such that $\tmsix \toin \tmfive'$ and $\tmseven \toin \tmfive'$.
      So, setting $\tm' = \tmfour\tmfive'$, one has $\tmtwo = \tmfour\tmsix \toin \tm'$ and $\tmthree = \tmfour\tmseven \toin \tm'$.

      \item \emph{Application Left for $\tm \toin \tmtwo$} and \emph{Application Right for $\tm \toin \tmthree$}, \ie $\tm = \tmfour\tmfive \toin \tmsix\tmfive = \tmtwo$ and $\tm = \tmfour\tmfive \toin \tmfour\tmseven = \tmthree$ with $\tmfour \toin \tmsix$ and $\tmfive \toin \tmseven$. 
      So, setting $\tm' = \tmsix\tmseven$, one has $\tmtwo = \tmsix\tmfive \toin \tm'$ and $\tmthree = \tmfour\tmseven \toin \tm'$.

      \item \emph{Application Right for $\tm \toin \tmtwo$} and \emph{Application Left for $\tm \toin \tmthree$}%
%       . This case is analogous to the previous one.
      , \ie $\tm = \tmfour\tmfive \toin \tmfour\tmsix = \tmtwo$ and $\tm = \tmfour\tmfive \toin \tmseven\tmfive = \tmthree$ with $\tmfive \toin \tmsix$ and $\tmfour \toin \tmseven$. 
      So, setting $\tm' = \tmseven\tmsix$, one has $\tmtwo = \tmfour\tmsix \toin \tm'$ and $\tmthree = \tmseven\tmfive \toin \tm'$.
    \end{itemize}
    
    \item We prove, by induction on $\tm \in \Lambda$, that if $\tm \tobabs \tmtwo$ and $\tm \toin \tmthree$, then $\tmtwo \neq \tmthree$ and there is $\tmp \!\in \Lambda$ such that $\tmtwo \toin \tmp$ and $\tmthree \tobabs \tmp$.
    
%     By \refrmk{values-betav-normal}, $\tm$ is not a value \ie $\tm$ is an application.
    Observe that neither $\tm \tobabs \tmtwo$ nor $\tm \toin \tmthree$ can be a step at the root: 
    indeed, if $\tm \defeq (\la\var\tmfour)\la\vartwo\tmfive \rtobabs\! \tmfour\isub\var{\la\vartwo\tmfive} \eqdef \tmtwo$ (resp.~$\tm \defeq (\la\var\tmfour)\gconst \rtoin\! \tmfour\isub\var\gconst \eqdef \tmthree$) then $\la\vartwo\tmfive$ is not a inert term (resp.~$\gconst$ is not an abstraction), moreover $\la\var\tmfour$ and $\la\vartwo\tmfive$ (resp.~$\gconst$) are $\betain$-normal (resp.~$\betaabs$-normal) by \refprop{characterize-fnormal}, as $\toin \, \subseteq \, \tof$ (resp.~$\tobabs \, \subseteq \, \tof$);
    therefore, $\tm$ is $\betain$-normal (resp.~$\betaabs$-normal) but this contradicts the hypothesis $\tm \toin \tmthree$ (resp.~$\tm \tobabs \tmtwo$).
    So, according to the definitions of $\tm \tobabs \tmtwo$ and $\tm \toin \tmthree$, there are only four cases.
    \begin{itemize}      
      \item \emph{Application Left for both $\tm \tobabs \tmtwo$ and $\tm \toin \tmthree$}, \ie $\tm \defeq \tmfour\tmfive \tobabs \tmsix\tmfive \eqdef \tmtwo$ and $\tm \defeq \tmfour\tmfive \toin \tmseven\tmfive \eqdef \tmthree$ with $\tmfour \tobabs \tmsix$ and $\tmfour \toin \tmseven$. 
      By \ih, $\tmsix \neq \tmseven$ and there exists $\tmfourp$ such that $\tmsix \toin \tmfour'$ and $\tmseven \tobabs \tmfour'$.
      So, $\tmtwo \neq \tmthree$ and, setting $\tmp \defeq \tmfourp\tmfive$, one has $\tmtwo = \tmsix\tmfive \toin \tmp \lRew{\betaabs\!} \tmseven\tmfive = \tmthree$.

      \item \emph{Application Right for both $\tm \tobabs \tmtwo$ and $\tm \toin \tmthree$}%
%       . This case is analogous to the previous one.
      , \ie $\tm \defeq \tmfour\tmfive \tobabs \tmfour\tmsix \eqdef \tmtwo$ and $\tm \defeq \tmfour\tmfive \toin \tmfour\tmseven \eqdef \tmthree$ with $\tmfive \tobabs \tmsix$ and $\tmfive \toin \tmseven$. 
      By \ih, $\tmsix \neq \tmseven$ and there exists $\tmfivep$ such that $\tmsix \toin \tmfivep$ and $\tmseven \tobabs \tmfivep$.
      So, $\tmtwo \neq \tmthree$ and, setting $\tmp \defeq \tmfour\tmfivep$, one has $\tmtwo = \tmfour\tmsix \toin \tmp \lRew{\betaabs\!} \tmfour\tmseven = \tmthree$.

      \item \emph{Application Left for $\tm \tobabs \tmtwo$} and \emph{Application Right for $\tm \toin \tmthree$}, \ie $\tm \defeq \tmfour\tmfive \tobabs \tmsix\tmfive = \tmtwo$ and $\tm = \tmfour\tmfive \toin \tmfour\tmseven \eqdef \tmthree$ with $\tmfour \tobabs \tmsix$ and $\tmfive \toin \tmseven$.
      By \reflemma{toin-different}, $\tmfive \neq \tmseven$ and hence $\tmtwo = \tmsix\tmfive \neq \tmfour\tmseven = \tmthree$.
      Setting $\tmp \defeq \tmsix\tmseven$, one has $\tmtwo = \tmsix\tmfive \toin \tmp \lRew{\betaabs\!} \tmfour\tmseven = \tmthree$.

      \item \emph{Application Right for $\tm \tobabs \tmtwo$} and \emph{Application Left for $\tm \toin \tmthree$}%. 
%       This case is analogous to the previous one.
      , \ie $\tm \defeq \tmfour\tmfive \tobabs \tmfour\tmsix \eqdef \tmtwo$ and $\tm = \tmfour\tmfive \toin \tmseven\tmfive = \tmthree$ with $\tmfive \tobabs \tmsix$ and $\tmfour \toin \tmseven$.
      By \reflemma{toin-different}, $\tmfour \neq \tmseven$ and hence $\tmtwo = \tmfour\tmsix \neq \tmseven\tmfive = \tmthree$.
      Setting $\tmp \defeq \tmseven\tmsix$, one has $\tmtwo = \tmfour\tmsix \toin \tmp \lRew{\betaabs} \tmseven\tmfive = \tmthree$.
    \end{itemize}

    \item It follows immediately from strong confluence of $\tobabs$\! (\refpropp{basic-plotkin}{strong-confluence}) and $\toin$ (\refpropp{basic-fireball}{toin-strong-confluence}), the strong commutation of $\tobabs$ and $\toin$ (\refpropp{basic-fireball}{tobv-toin-strong-commutation}), and Hindley-Rosen (\reflemma{hindley-rosen}).
    \qedhere
  \end{enumerate}
\end{proof}

\subsubsection{Proofs of Subsection \ref{subsect:vsub} (\texorpdfstring{Open CBV 2: the Value Substitution Calculus $\vsubcalc$}{Open CBV 2: the Value Substitution Calculus})}

Note that $\vsub$-values are $\vsub$-normal (since $\tovsub$ does not reduce under $\lambda$'s) and closed under substitution (\ie $\val\isub\var\valtwo$ is a $\vsub$-value, for any $\vsub$-values $\val$ and $\valtwo$).

\setcounter{propositionAppendix}{\value{prop:basic-value-substitution}}
\begin{propositionAppendix}[Basic Properties of $\vsubcalc$, \cite{DBLP:conf/flops/AccattoliP12}]
\label{propappendix:basic-value-substitution}
  {\ }\NoteState{prop:basic-value-substitution}
  \hfill
  \begin{enumerate}
    \item\label{pappendix:basic-value-substitution-tom-toe-terminates} $\tom$ and $\toe$ are strongly normalizing (separately).
    \item\label{pappendix:basic-tom-toe-strong-confluence} $\tom$ and $\toe$ are strongly confluent (separately).
    
    \item\label{pappendix:basic-value-substitution-tom-toe-commute}  $\tom$ and $\toe$ strongly commute.
    \item\label{pappendix:basic-value-substitution-strong-confluence} $\tovsub$ is strongly confluent, and all $\vsub$-normalizing derivations $\deriv$ from $\tm \!\in\! \vsubterms$ (if any) have the same length $\sizevsub{\deriv}$, the same number $\sizee{\deriv}$ of $\expo$-steps, and the same number $\sizee{\deriv}$ of $\mult$-steps.
    \item\label{pappendix:basic-value-substtution-expo-less-than-mult} Let $\tm \!\in\! \Lambda$. For any $\vsub$-derivation $\deriv$ from $\tm$, $\sizee{d} \leq \sizem{\deriv}$.
  \end{enumerate}
\end{propositionAppendix}

\begin{proof} The statements of \refprop{basic-value-substitution} are a refinement of some results proved in \cite{DBLP:conf/flops/AccattoliP12}, where $\tovsub$ is denoted by $\tow$.
  \begin{enumerate}
    \item In \cite[Lemma~3]{DBLP:conf/flops/AccattoliP12} it has been proved that $\todb$ and $\tovs$ are strongly normalizing, separately.
    Since $\tom \, \subseteq \, \todb$ and $\toe \, \subseteq \, \tovs$ ($\todb$ and $\tovs$ are just the extensions of $\tom$ and $\toe$, respectively, obtained by allowing reductions under $\l$'s), one has that $\tom$ and $\toe$ are strongly normalizing, separately.
    
    \item We prove that $\tom$ is strongly confluent, \ie if $\tmtwo \lRew{\mult} \tm \tom \tmthree$ with $\tmtwo \neq \tmthree$ then there exists $\tmp \in \Lambda_\vsub$ such that $\tmtwo \tom \tmp \lRew{\mult} \tmthree$.
    The proof is by induction on the definition of $\tom$. Since there $\tm \tom \tmthree \neq \tmtwo$ and the reduction $\tom$ is weak, there are only eight cases:
    \begin{itemize}
      \item \emph{Step at the Root for $\tm \!\tom\! \tmtwo$ and Application Right for $\tm \!\tom\! \tmthree$}, \ie $\tm \defeq \sctxp{\la\var\tmfive}\tmfour \rtom \sctxp{\tmfive\esub{\var}{\tmfour}} \eqdef \tmtwo$ and $\tm \!\rtom\! \sctxp{\la\var\tmfive}\tmfourp\! \eqdef \tmthree$ with $\tmfour \!\tom\! \tmfourp$: then, $\tmtwo \!\tom\! \sctxp{\tmfive\esub{\var}{\tmfourp}} \!\lRew{\mult}\! \tmthree$;
      
      \item \emph{Step at the Root for $\tm \tom \tmtwo$ and Application Left for $\tm \tom \tmthree$}, \ie, for some $n > 0$, $\tm \defeq (\la\var\tmfive)\esub{\var_1}{\tm_1}\dots\esub{\var_n}{\tm_n}\tmfour \allowbreak\rtom \tmfive\esub{\var}{\tmfour}\esub{\var_1}{\tm_1}\dots\esub{\var_n}{\tm_n} \eqdef \tmtwo$ whereas $\tm \tom \allowbreak (\la\var\tmfive)\esub{\var_1}{\tm_1}\dots\esub{\var_j}{\tmp_j}\dots\esub{\var_n}{\tm_n}\tmfour \eqdef \tmthree$ with $\tm_j \tom \tmp_j$ for some $1 \leq j \leq n$: then, 
      \begin{align*}
        \tmtwo \tom \allowbreak \tmfive\esub{\var}{\tmfour}\esub{\var_1}{\tm_1}\dots\esub{\var_j}{\tmp_j}\dots\esub{\var_n}{\tm_n} \lRew{\mult} \tmthree;
      \end{align*}
      \item \emph{Application Left for $\tm \tom \tmtwo$ and Application Right for $\tm \tom \tmthree$}, \ie $\tm \defeq \tmfour\tmfive \tom \tmfourp\tmfive \eqdef \tmtwo$ and $\tm \tom \tmfour\tmfivep \eqdef \tmthree$ with $\tmfour \tom \tmfourp$ and $\tmfive \tom \tmfivep$: then, $\tmtwo \tom \tmfourp\tmfivep\! \lRew{\mult} \tmthree$;
      \item \emph{Application Left for both $\tm \tom \tmtwo$ and $\tm \tom \tmthree$}, \ie $\tm \defeq \tmfour\tmfive \tom \tmfourp\tmfive \eqdef \tmtwo$ and $\tm \tom \tmfour''\tmfive \eqdef \tmthree$ with $\tmfourp \lRew{\mult} \tmfour \tom \tmfour''$: by \ih, there exists $\tmfour_0 \in \Lambda_\vsub$ such that $\tmfourp \tom \tmfour_0 \lRew{\mult} \tmfour''$, hence $\tmtwo \tom \tmfour_0\tmfive \lRew{\mult} \tmthree$;
      \item \emph{Application Right for both $\tm \tom \tmtwo$ and $\tm \tom \tmthree$}, \ie $\tm \defeq \tmfive\tmfour \tom \tmfive\tmfourp \eqdef \tmtwo$ and $\tm \tom \tmfive\tmfour'' \eqdef \tmthree$ with $\tmfourp \lRew{\mult} \tmfour \tom \tmfour''$: by \ih, there exists $\tmfour_0 \in \Lambda_\vsub$ such that $\tmfourp \tom \tmfour_0 \lRew{\mult} \tmfour''$, hence $\tmtwo \tom \tmfive\tmfour_0 \lRew{\mult} \tmthree$;
      \item \emph{ES Left for $\tm \tom \tmtwo$ and ES Right for $\tm \tom \tmthree$}, \ie $\tm \defeq \tmfour\esub\var\tmfive \tom \tmfourp\esub\var\tmfive \eqdef \tmtwo$ and $\tm \tom \tmfour\esub\var\tmfivep \eqdef \tmthree$ with $\tmfour \tom \tmfourp$ and $\tmfive \tom \tmfivep$: then, $\tmtwo \tom \tmfourp\esub\var\tmfivep\! \lRew{\mult} \tmthree$;
      \item \emph{ES Left for both $\tm \tom \tmtwo$ and $\tm \tom \tmthree$}, \ie $\tm \defeq \tmfour\esub\var\tmfive \tom \tmfourp\esub\var\tmfive \eqdef \tmtwo$ and $\tm \tom \tmfour''\esub\var\tmfive \eqdef \tmthree$ with $\tmfourp \lRew{\mult} \tmfour \tom \tmfour''$: by \ih, there exists $\tmfour_0 \in \Lambda_\vsub$ such that $\tmfourp \tom \tmfour_0 \lRew{\mult} \tmfour''$, hence $\tmtwo \tom \tmfour_0\esub\var\tmfive \lRew{\mult} \tmthree$;
      \item \emph{ES Right for both $\tm \tom \tmtwo$ and $\tm \tom \tmthree$}, \ie $\tm \defeq \tmfive\esub\var\tmfour \tom \tmfive\esub\var\tmfourp \eqdef \tmtwo$ and $\tm \tom \tmfive\esub\var{\tmfour''} \eqdef \tmthree$ with $\tmfourp \lRew{\mult} \tmfour \tom \tmfour''$: by \ih, there exists $\tmfour_0 \in \Lambda_\vsub$ such that $\tmfourp \tom \tmfour_0 \lRew{\mult} \tmfour''$, hence $\tmtwo \tom \tmfive\esub\var{\tmfour_0} \lRew{\mult} \tmthree$.
    \end{itemize}
    
    \smallskip
    We prove that $\toe$ is strongly confluent, \ie if $\tmtwo \lRew{\expo} \tm \toe \tmthree$ with $\tmtwo \neq \tmthree$ then there exists $\tmfour \in \Lambda_\vsub$ such that $\tmtwo \toe \tmp \lRew{\expo} \tmthree$.
    The proof is by induction on the definition of $\toe$. Since there $\tm \toe \tmthree \neq \tmtwo$ and the reduction $\toe$ is weak, there are only eight cases:
    \begin{itemize}
      \item \emph{Step at the Root for $\tm \!\toe\! \tmtwo$} and \emph{ES Left for $\tm \!\toe\! \tmthree$}, \ie $\tm \defeq \tmfour\esub\var{\sctxp{\val}} \rtoe \sctxp{\tmfour\isub{\var}{\val}} \eqdef \tmtwo$ and $\tm \!\rtoe\! \tmfourp\esub\var{\sctxp{\val}}\! \eqdef \tmthree$ with $\tmfour \!\toe\! \tmfourp$: then, $\tmtwo \!\toe\! \sctxp{\tmfourp\esub{\var}{\val}} \!\lRew{\expo}\! \tmthree$;
      
      \item \emph{Step at the Root for $\tm \toe \tmtwo$} and \emph{ES Right for $\tm \toe \tmthree$}, \ie, for some $n > 0$, $\tm \defeq \tmfour\esub\var{\val\esub{\var_1}{\tm_1}\dots\esub{\var_n}{\tm_n}} \allowbreak\rtoe \tmfour\isub{\var}{\val}\esub{\var_1}{\tm_1}\dots\esub{\var_n}{\tm_n} \eqdef \tmtwo$ whereas $\tm \toe \allowbreak \tmfour\esub{\var}{\val\esub{\var_1}{\tm_1}\dots\esub{\var_j}{\tmp_j}\dots\esub{\var_n}{\tm_n}} \eqdef \tmthree$ with $\tm_j \toe \tmp_j$ for some $1 \leq j \leq n$: then, 
      \begin{align*}
        \tmtwo \toe \allowbreak \tmfour\isub{\var}{\val}\esub{\var_1}{\tm_1}\dots\esub{\var_j}{\tmp_j}\dots\esub{\var_n}{\tm_n} \lRew{\expo} \tmthree;
      \end{align*}
      \item \emph{Application Left for $\tm \toe \tmtwo$} and \emph{Application Right for $\tm \toe \tmthree$}, \ie $\tm \defeq \tmfour\tmfive \toe \tmfourp\tmfive \eqdef \tmtwo$ and $\tm \toe \tmfour\tmfivep \eqdef \tmthree$ with $\tmfour \toe \tmfourp$ and $\tmfive \toe \tmfivep$: then, $\tmtwo \toe \tmfourp\tmfivep\! \lRew{\expo} \tmthree$;
      \item \emph{Application Left for both $\tm \toe \tmtwo$ and $\tm \toe \tmthree$}, \ie $\tm \defeq \tmfour\tmfive \toe \tmfourp\tmfive \eqdef \tmtwo$ and $\tm \toe \tmfour''\tmfive \eqdef \tmthree$ with $\tmfourp \lRew{\expo} \tmfour \toe \tmfour''$: by \ih, there exists $\tmfour_0 \in \Lambda_\vsub$ such that $\tmfourp \toe \tmfour_0 \lRew{\expo} \tmfour''$, hence $\tmtwo \toe \tmfour_0\tmfive \lRew{\expo} \tmthree$;
      \item \emph{Application Right for both $\tm \toe \tmtwo$ and $\tm \toe \tmthree$}, \ie $\tm \defeq \tmfive\tmfour \toe \tmfive\tmfourp \eqdef \tmtwo$ and $\tm \toe \tmfive\tmfour'' \eqdef \tmthree$ with $\tmfourp \lRew{\expo} \tmfour \toe \tmfour''$: by \ih, there exists $\tmfour_0 \in \Lambda_\vsub$ such that $\tmfourp \toe \tmfour_0 \lRew{\expo} \tmfour''$, hence $\tmtwo \toe \tmfive\tmfour_0 \lRew{\expo} \tmthree$;
      \item \emph{ES Left for $\tm \toe \tmtwo$} and \emph{ES Right for $\tm \toe \tmthree$}, \ie $\tm \defeq \tmfour\esub\var\tmfive \toe \tmfourp\esub\var\tmfive \eqdef \tmtwo$ and $\tm \toe \tmfour\esub\var\tmfivep \eqdef \tmthree$ with $\tmfour \toe \tmfourp$ and $\tmfive \toe \tmfivep$: then, $\tmtwo \toe \tmfourp\esub\var\tmfivep\! \lRew{\expo} \tmthree$;
      \item \emph{ES Left for both $\tm \toe \tmtwo$ and $\tm \toe \tmthree$}, \ie $\tm \defeq \tmfour\esub\var\tmfive \toe \tmfourp\esub\var\tmfive \eqdef \tmtwo$ and $\tm \toe \tmfour''\esub\var\tmfive \eqdef \tmthree$ with $\tmfourp \lRew{\expo} \tmfour \toe \tmfour''$: by \ih, there exists $\tmfour_0 \in \Lambda_\vsub$ such that $\tmfourp \toe \tmfour_0 \lRew{\expo} \tmfour''$, hence $\tmtwo \toe \tmfour_0\esub\var\tmfive \lRew{\expo} \tmthree$;
      \item \emph{ES Right for both $\tm \toe \tmtwo$ and $\tm \toe \tmthree$}, \ie $\tm \defeq \tmfive\esub\var\tmfour \toe \tmfive\esub\var\tmfourp \eqdef \tmtwo$ and $\tm \toe \tmfive\esub\var{\tmfour''} \eqdef \tmthree$ with $\tmfourp \lRew{\expo} \tmfour \toe \tmfour''$: by \ih, there exists $\tmfour_0 \in \Lambda_\vsub$ such that $\tmfourp \toe \tmfour_0 \lRew{\expo} \tmfour''$, hence $\tmtwo \toe \tmfive\esub\var{\tmfour_0} \lRew{\expo} \tmthree$.
    \end{itemize}
    
    \smallskip
    
    Note that in \cite[Lemma~11]{DBLP:conf/flops/AccattoliP12} it has just been proved the strong confluence of $\tovsub$, not of $\tom$ or $\toe$.
 
    \item We show that $\toe$ and $\tom$ strongly commute, \ie if $\tmtwo \lRew{\expo} \tm \tom \tmthree$, then $\tmtwo \neq \tmthree$ and there is $\tmp \in \Lambda_\vsub$ such that $\tmtwo \tom \tmp \lRew{\expo} \tmthree$. The proof is by induction on the definition of $\tm \toe \tmtwo$. The proof that $\tmtwo \neq \tmthree$ is left to the reader.
    Since the $\toe$ and $\tom$ cannot reduce under $\l$'s, all $\vsub$-values are $\mult$-normal and $\expo$-normal. So, there are the following cases.
    \begin{itemize}
      \item \emph{Step at the Root for $\tm \toe \tmtwo$} and \emph{ES Left for $\tm \tom \tmthree$}, \ie $\tm \defeq \tmfour\esub\varthree{\sctxp\val} \toe \sctxp{\tmfour\isub\varthree{\val}} \eqdef \tmtwo$ and $\tm \tom \tmfourp\esub\varthree{\sctxp{\val}} \eqdef \tmthree$ with $\tmfour \tom \tmfourp$: then $\tmtwo \tom \sctxp{\tmfourp\isub\varthree{\val}} \lRew{\expo} \tmtwo$;
      \item \emph{Step at the Root for $\tm \toe \tmtwo$} and \emph{ES Right for $\tm \tom \tmthree$}, \ie 
      \begin{align*}
        \tm &\defeq \tmfour\esub\varthree{\val\esub{\var_1}{\tm_1}\dots\esub{\var_n}{\tm_n}} \\
	    &\toe \tmfour\isub\varthree{\val}\esub{\var_1}{\tm_1}\dots\esub{\var_n}{\tm_n} \eqdef \tmtwo
      \end{align*}
      and $\tm \tom \tmfour\esub\varthree{\val\esub{\var_1}{\tm_1}\dots\esub{\var_j}{\tmp_j}\dots\esub{\var_n}{\tm_n}} \eqdef \tmthree$
      for some $n > 0$, and $\tm_j \tom \tmp_j$ for some $1 \leq j \leq n$: then, $\tmtwo \tom \tmfour\isub\varthree{\val}\esub{\var_1}{\tm_1}\dots\esub{\var_j}{\tmp_j}\dots\esub{\var_n}{\tm_n} \lRew{\expo} \tmthree$;

      \item \emph{Application Left for $\tm \toe \tmtwo$} and \emph{Application Right for $\tm \tom \tmthree$}, \ie $\tm \defeq \tmfour\tmfive \toe \tmfourp\tmfive \eqdef \tmtwo$ and $\tm \tom \tmfour\tmfivep \eqdef \tmthree$ with $\tmfour \toe \tmfourp$ and $\tmfive \tom \tmfivep$: then, $\tm \tom \tmfourp\tmfivep \lRew{\expo} \tmtwo$;
      \item \emph{Application Left for both $\tm \toe \tmtwo$ and $\tm \tom \tmthree$}, \ie $\tm \defeq \tmfour\tmfive \toe \tmfourp\tmfive \eqdef \tmtwo$ and $\tm \tom \tmfour''\tmfive \eqdef \tmthree$ with $\tmfourp \lRew{\expo} \tmfour \tom \tmfour''$: by \ih, there exists $\tmsix \in \Lambda_\vsub$ such that $\tmfourp \tom \tmsix \lRew{\expo} \tmfour''$, hence $\tmtwo \tom \tmsix\tmfive \lRew{\expo} \tmthree$;
      \item \emph{Application Left for $\tm \toe \tmtwo$} and \emph{Step at the Root for $\tm \tom \tmthree$}, \ie $\tm \defeq (\la\var\tmfive){\esub{\var_1}{\tm_1}\dots\esub{\var_n}{\tm_n}}\tmfour \toe (\la\var\tmfive){\esub{\var_1}{\tm_1}\dots\esub{\var_j}{\tmp_j}\dots\esub{\var_n}{\tm_n}}\tmfour \eqdef \tmtwo$ with $n > 0$ and $\tm_j \toe \tmp_j$ for some $1 \leq j \leq n$, and $\tm \tom \allowbreak \tmfive\esub\var\tmfour{\esub{\var_1}{\tm_1}\dots\esub{\var_n}{\tm_n}} \eqdef \tmthree$: then,
      \begin{equation*}
        \tmtwo \tom \tmfive\esub\var\tmfour{\esub{\var_1}{\tm_1}\dots\esub{\var_j}{\tmp_j}\dots\esub{\var_n}{\tm_n}} \lRew{\expo} \tmthree;
      \end{equation*}
      \item \emph{Application Right for $\tm \toe \tmtwo$} and \emph{Application Left for $\tm \tom \tmthree$}, \ie $\tm \defeq \tmfive\tmfour \toe \tmfive\tmfourp \eqdef \tmtwo$ and $\tm \tom \tmfivep\tmfour \eqdef \tmthree$ with $\tmfour \toe \tmfourp$ and $\tmfive \tom \tmfivep$: then, $\tmtwo \tom \tmfivep\tmfourp \lRew{\expo} \tmthree$;
      \item \emph{Application Right for both $\tm \toe \tmtwo$ and $\tm \tom \tmthree$}, \ie $\tm \defeq \tmfive\tmfour \toe \tmfive\tmfourp \eqdef \tmtwo$ and $\tm \tom \tmfive\tmfour'' \eqdef \tmthree$ with $\tmfourp \lRew{\expo} \tmfour \tom \tmfour''$: by \ih, there exists $\tmsix \in \Lambda_\vsub$ such that $\tmfourp \tom \tmsix \lRew{\expo} \tmfour''$, hence $\tmtwo \tom \tmfive\tmsix \lRew{\expo} \tmthree$;
      \item \emph{Application Right for $\tm \toe \tmtwo$} and \emph{Step at the Root for $\tm \tom \tmthree$}, \ie $\tm \defeq \sctxp{\la\var\tmfive}\tmfour \toe \sctxp{\la\var\tmfive}\tmfourp \eqdef \tmtwo$ with $\tmfour \toe \tmfourp$\!, and $\tm \tom \allowbreak \sctxp{\tmfive\esub\var\tmfour} \eqdef \tmthree$: then, $\tmtwo \tom \sctxp{\tmfive\esub\var\tmfourp} \lRew{\expo} \tmthree$;
      \item \emph{ES Left for $\tm \toe \tmtwo$} and \emph{ES Right for $\tm \tom \tmthree$}, \ie $\tm \defeq \tmfour\esub\var\tmfive \toe \tmfourp\esub\var\tmfive \eqdef \tmtwo$ and $\tm \tom \tmfour\esub\var\tmfivep \eqdef \tmthree$ with $\tmfour \toe \tmfourp$ and $\tmfive \tom \tmfivep$: then, $\tmtwo \tom \tmfourp\esub\var\tmfivep \lRew{\expo} \tmthree$;
      \item \emph{ES Left for both $\tm \toe \tmtwo$ and $\tm \tom \tmthree$}, \ie $\tm \defeq \tmfour\esub\var\tmfive \toe \tmfourp\esub\var\tmfive \eqdef \tmtwo$ and $\tm \tom \tmfour''\esub\var\tmfive \eqdef \tmthree$ with $\tmfourp \lRew{\expo} \tmfour \tom \tmfour''$: by \ih,there exists $\tmsix \in \Lambda_\vsub$ such that $\tmfourp \tom \tmsix \lRew{\expo} \tmfour''$, hence $\tmtwo \tom \tmsix\esub\var\tmfive \lRew{\expo} \tmthree$;
      \item \emph{ES Right for $\tm \toe \tmtwo$} and \emph{ES Left for $\tm \tom \tmthree$}, \ie $\tm \defeq \tmfive\esub\var\tmfour \toe \tmfive\esub\var\tmfourp \eqdef \tmtwo$ and $\tm \tom \tmfivep\esub\var\tmfour \eqdef \tmthree$ with $\tmfour \toe \tmfourp$ and $\tmfive \tom \tmfivep$: then, $\tmtwo \tom \tmfivep\esub\var\tmfourp \lRew{\expo} \tmthree$;
      \item \emph{ES Right for both $\tm \toe \tmtwo$ and $\tm \tom \tmthree$}, \ie $\tm \defeq \tmfive\esub\var\tmfour \toe \tmfive\esub\var{\tmfourp} \eqdef \tmtwo$ and $\tm \tom \tmfive\esub\var{\tmfour''} \eqdef \tmthree$ with $\tmfour \lRew{\expo} \tmfourp \tom \tmfour''$: by \ih, there exists $\tmsix \in \Lambda_\vsub$ such that $\tmfour \tom \tmsix \lRew{\expo} \tmfour''$, hence $\tmtwo \tom \tmfive\esub\var\tmsix \lRew{\expo} \tmthree$.
    \end{itemize}

    \item It follows immediately from strong confluence of $\tom$ and $\toe$ (\refpropp{basic-value-substitution}{tom-toe-strong-confluence}), strong commutation of $\tom$ and $\toe$ (\refpropp{basic-value-substitution}{tom-toe-commute}) and Hindley-Rosen (\reflemma{hindley-rosen}).
    
    A different proof of the strong confluence of $\tovsub$ (without information about the number of steps) is in \cite[Lemma~11]{DBLP:conf/flops/AccattoliP12}.
    \item The intuition behind the proof is that any $\mult$-step creates a new ES, any $\expo$-step erases an ES.
    Formally, let $\tmtwo \!\in\! \Lambda_\vsub$ such that $\deriv \colon \tm \tovsub^* \tmtwo$. 
    We prove by induction on $\sizevsub{\deriv} \in \nat$ that $\sizee{\deriv} \!=\! \sizem{\deriv} \!-\! \nes{\tmtwo}$ (where $\nes{\tmtwo}$ is the number of ES in $\tmtwo$) and any $\vsub$-value that is a subterm of $\tmtwo$ is a value (without ES).
    
    If $\sizevsub{\deriv} = 0$, then $\tmtwo = \tm \in \Lambda$, then we can conclude. 
    
    Suppose $\sizevsub{\deriv} > 0$: then, $\deriv$ is the concatenation of $\deriv' \colon \tm \tovsub^* \tmthree$ and $\tmthree \tovsub \tmtwo$, for some $\tmthree \in \Lambda_\vsub$.
    By \ih, $\sizee{\deriv'} = \sizem{\deriv'} - \nes{\tmthree}$ and that every $\vsub$-value that is a subterm of $\tmthree$ is a value (without ES). There are two cases:
    \begin{itemize}
      \item $\tmthree \defeq \evctxp{\tmfour\esub\var{\sctxp\val}} \toe \evctxp{\sctxp{\tmfour\isub\var{\val}}} \eqdef \tmtwo$, then $\sizem{\deriv} = \sizem{\deriv'}$ and $\nes{\tmthree} = \nes{\tmtwo} + 1$, since $\nes{\val} = 0$ by \ih; therefore $\sizee{\deriv} = \sizee{\deriv'} +1 = \sizem{\deriv'} - \nes{\tmthree} +1 = \sizem{\deriv} - \nes{\tmtwo}$ and any $\vsub$-value that is a subterm of $\tmtwo$ is a value (without ES).
      \item $\tmthree \defeq \evctxp{\sctxp{\la\var\tmfour}\tmfive} \tom \evctxp{\sctxp{\tmfour\esub\var\tmfive}} \eqdef \tmtwo$, then $\nes{\tmtwo} = \nes{\tmthree} + 1$ and $\sizem{\deriv} = \sizem{\deriv'}+1$, therefore $\sizee{\deriv} = \sizee{\deriv'} = \sizem{\deriv'} - \nes{\tmthree} = \sizem{\deriv} - \nes{\tmtwo}$. 
      Moreover, the new occurrence of ES $\esub\var\tmfive$ in $\tmtwo$ cannot be under the scope of a $\lambda$, otherwise the redex in $\tmthree$ which is fired in the $\mult$-step would be under the scope of a $\lambda$, but this is impossible since $\tom$ is a weak reduction.
      So, any $\vsub$-value that is a subterm of $\tmtwo$ is a value (without ES).
    \end{itemize}
    \qedhere
  \end{enumerate}
\end{proof}

\subsubsection{Proof of Subsection \ref{subsect:shuffling} (\texorpdfstring{Open CBV 3: the Shuffling Calculus $\shufcalc$}{Open CBV 3: the Shuffling Calculus})}

\begin{definition}[Occurrences]
\label{def:occurrences}
  For all $\tm \in \Lambda$, let $[\tm]_\lambda$ be the number of occurrences of $\lambda$ in $\tm$, and $[\tm]_\var$ be the number of free occurrences of the variable $\var$ in $\tm$, and $\mathsf{sub}_\tmtwo(\tm)$ be the number of occurrences in $\tm$ of the term $\tmtwo$%, for any subterm $\tmthree$ of $\tm$
  .
\end{definition}

\begin{remark}
\label{rmk:value-sigm-normal}
  Since $\tobvm$ and $\tosigm$ do not reduce under $\l$'s without argument, every value is {$\betavm$-normal} and $\sigm$-normal, and hence $\vmsym$-normal.
\end{remark}

\begin{remark}
\label{rmk:alternativedef-sigm}
  The reduction $\tosigm$ is just the closure under balanced contexts of the binary relation $\rtosigma \, = \, \rtosl \!\cup \rtosr$ on $\Lambda$ (see definitions in \reffig{shuffling-calculus}).
\end{remark}

\begin{lemma}
\label{l:different}
  Let $\tm, \tm' \in \Lambda$.
  \begin{enumerate}
%     \item\label{p:different-occurrences} If $\tm \tosigm \tm'$ then $[\tm]_\lambda = [\tm']_\lambda$ and $[\tm]_\var = [\tm']_\var$ for any variable $\var$ (in particular, $\fv{\tm} = \fv{\tm'}$).
% %     and $\mathsf{sub}_\val(\tm) = \mathsf{sub}_\val(\tm')$ for any value $\val$. FALSE! $(\la\var{(\la\vartwo\vartwo)(\var\varthree) I})I$
    \item\label{p:different-tosigm-sub} For every value $\val$, if $\tm \tosigm \tm'$ then $(\la\var\tm')\val \neq \tm\isub\var\val$.
    \item\label{p:different-tosigm} If $\tm \tosigm \tm'$ then $\tm \neq \tm'$.
    \item\label{p:different-sub} For every value $\val$, one has $\tm\isub\var\val \neq \la\var\tm\val$.
  \end{enumerate}
\end{lemma}

\begin{proof}\hfill
  \begin{enumerate}

    \item 
%     We can suppose without loss of generality that $\var \notin \fv{\val}$.
%     According to \refdef{occurrences} and \reflemmap{different}{occurrences}, $[\tm\isub\var\val]_\lambda = [\tm]_\lambda + [\val]_\lambda \!\cdot\! [\tm]_\var$ and $[(\la\var\tm')\val]_\lambda = 1 + [\tm']_\lambda + [\val]_\lambda = 1 + [\tm]_\lambda + [\val]_\lambda$.
% 
%     If $(\la\var\tm')\val = \tm\isub\var\val$ then $[(\la\var\tm')\val]_\lambda = [\tm\isub\var\val]_\lambda$, thus
%     \begin{align}
%     \label{eq:occurrences-tosigm}
%       [\val]_\lambda \!\cdot\! [\tm]_\var &= 1 + [\val]_\lambda 
%     \end{align}
%     The only solution to Eq.\,\refeq{occurrences-tosigm} is $[\val]_\lambda = 1$ and $[\tm]_\var = 2$.
%     Since $\var \notin \fv{\val}$, $\mathsf{sub}_\val((\la\var\tm')\val) = 1 + \mathsf{sub}_\val(\tm')$ and $\mathsf{sub}_\val(\tm\isub\var\val) = \mathsf{sub}_\val(\tm) + [\tm]_\var = \mathsf{sub}_\val(\tm') + 2$ ($\mathsf{sub}_\val(\tm') = \mathsf{sub}_\val(\tm)$ by \reflemmap{different}{occurrences} 
%     FALSE! take $\tm = (\la\vartwo\vartwo)(\var\varthree) I$), so $\mathsf{sub}_\val((\la\var\tm)\val) \neq \mathsf{sub}_\val(\tm'\isub\var\val)$ and $(\la\var\tm)\val \neq \tm'\isub\var\val$.
      By induction on the definition of $\tm \tosigm \tm'$, using \refrmk{alternativedef-sigm}.
    
    \item 
    In \cite[Proposition~2]{DBLP:conf/fossacs/CarraroG14} it has been proved that there exists a size $\# \colon \Lambda \to \nat$ such that if $\tm \tosig \tm'$ then $\#(\tm) > \#(\tm')$, where $\tosig$ is just the extension of $\tosigm$ obtained by allowing reductions under $\l$'s.
    Therefore, $\tosigm \, \subseteq \, \tosig$ and hence if $\tm \tosigm \tm'$ then $\#(\tm) > \#(\tm')$, in particular $\tm \neq \tm'$.
    
    \item 
%     By induction on $\tm \in \Lambda$.
%     
%     If $\tm$ is a variable then either $\tm = \var$ and hence $\tm\isub\var\val = \val \neq \la\var\tm\val$, or $\tm = \vartwo \neq \var$ and hence $\tm\isub\var\val = \vartwo \neq \la\var\tm\val$.
%     
%     If $\tm = \la\vartwo \tmtwo$ then we can suppose without loss of generality that $\vartwo \notin \fv{\val} \cup \{\var\}$ and thus $\tm\isub\var\val = \la\vartwo\tmtwo\isub\var\val$.
%     By \ih, $\tmtwo\isub\var\val \neq \la\var\tmtwo\val$.
%     
%     If $\tmtwo$ is a value then $\tmtwo\isub\var\val$ is a value (see \cite[Property~1.2.4.ii]{parametricBook}) and not an application, therefore $\la\vartwo \tmtwo\isub\var\val \neq \la\var{(\la\vartwo\tmtwo)\val}$.
%     
%     Otherwise $\tmtwo = \tmtwo_1\tmtwo_2$ and then $\tmtwo\isub\var\val = \tmtwo_1\isub\var\val \tmtwo_2\isub\var\val$
%     
%     $\tm\isub\var\val = \la\vartwo \tmtwo\isub\var\val \neq \la\var{(\la\vartwo\tmtwo)\val} = \la\var\tm\val$
%     
%     Finally, if $\tm = \tmtwo\tmthree$ then $\tm\isub\var\val = \tmtwo\isub\var\val \tmthree\isub\var\val \neq \la\var\tm\val$ since $\la\var\tm\val$ is not an application.    
    According to \refdef{occurrences}, $[\tm\isub\var\val]_\lambda = [\tm]_\lambda + [\val]_\lambda \!\cdot\! [\tm]_\var$ and $[\la\var\tm\val]_\lambda = 1 + [\tm]_\lambda + [\val]_\lambda$, and $[\tm\isub\var\val]_\var = [\tm]_\var \!\cdot\! [\val]_\var$ and $[\la\var\tm\val]_\var = 0$.
    
    Suppose $\tm\isub\var\val = \la\var\tm\val$: then, $[\tm\isub\var\val]_\lambda = [\la\var\tm\val]_\lambda$ and $[\tm\isub\var\val]_\var \allowbreak= [\la\var\tm\val]_\var$, thus
    \begin{align}
    \label{eq:occurrences}
      [\val]_\lambda \!\cdot\! [\tm]_\var &= 1 + [\val]_\lambda 
      & [\tm]_\var \!\cdot\! [\val]_\var &= 0.
    \end{align}
    The only solution to the first equation of \refeq{occurrences} is $[\val]_\lambda = 1$ and $[\tm]_\var = 2$, whence $[\val]_\var = 0$ according to the second equation of \refeq{occurrences}.
    As $\var \notin \fv{\val}$, one has $\mathsf{sub}_\val(\la\var\tm\val) = 1 + \mathsf{sub}_\val(\tm)$ and $\mathsf{sub}_\val(\tm\isub\var\val) = \mathsf{sub}_\val(\tm) + [\tm]_\var = \mathsf{sub}_\val(\tm) + 2$, therefore $\mathsf{sub}_\val(\la\var\tm\val) \neq \mathsf{sub}_\val(\tm\isub\var\val)$ and hence $\la\var\tm\val \neq \tm\isub\var\val$.
    Contradiction.
    \qedhere
  \end{enumerate}
\end{proof}

\setcounter{propositionAppendix}{\value{prop:basic-shuffling}}
\begin{propositionAppendix}[Basic Properties of $\shufcalc$, \cite{DBLP:conf/fossacs/CarraroG14}]
\label{propappendix:basic-shuffling}\hfill
\NoteState{prop:basic-shuffling}
  \begin{enumerate}
    \item\label{pappendix:basic-shuffling-different} Let $\tm, \tmtwo, \tmthree \in \Lambda$: if $\tm \tobvm \tmtwo$ and $\tm \tosigm \tmthree$ then $\tmtwo \neq\tmthree$.
    \item\label{pappendix:basic-shuffling-terminates} $\tosigm$ is strongly normalizing and (not strongly) confluent.
    \item  $\toshuf$ is (not strongly) confluent.
    \item  Let $\tm \in \Lambda$: $\tm$ is strongly $\vmsym$-normalizable iff $\tm$ is $\vmsym$-normalizable.
  \end{enumerate}
\end{propositionAppendix}

\begin{proof}\hfill
  \begin{enumerate}
    \item By induction on $\tm \in \Lambda$.
    According to the definition of $\tm \tosigm \tmthree$ and \refrmk{alternativedef-sigm}, the following cases are impossible.
    \begin{itemize}
      \item \emph{Step at the root for $\tm \tobvm \tmtwo$} and either the \emph{Step at the root} or the \emph{Application Left} or the \emph{Application Right for $\tm \tosigm \tmthree$}.
      Indeed, if $\tm = (\la\var\tmfour)\val \rtobv \tmfour\isub\var\val = \tmtwo$ then $\la\var\tmfour$ and $\val$ are $\sigm$-normal by \refrmk{value-sigm-normal};
      moreover $\tm$ is neither a $\slsym$-redex nor a $\srsym$-redex, because $\la\var\tmfour$ and $\val$, respectively, are not applications.
     \item \emph{Application Left for $\tm \tobvm \tmtwo$} and \emph{Step inside a $\beta$-context for $\tm \tosigm \tmthree$}, \ie $\tm = \tmfour\tmfive \tobvm \tmsix\tmfive = \tmtwo$ with $\tmfour \tobvm \tmsix$, and $\tm = (\la\var\tmfour')\tmfive \tosigm (\la\var\tmseven)\tmfive = \tmthree$ with $\tmfour = \la\var\tmfour'$ and $\tmfour' \tosigm \tmseven$.
     Indeed $\tmfour$ is $\betavm$-normal by \refrmk{value-sigm-normal}. 
     \item \emph{Step inside a $\beta$-context for $\tm \tobvm \tmtwo$} and \emph{Application Left for $\tm \tosigm \tmthree$}, \ie $\tm = \tmfour\tmfive \tosigm \tmsix\tmfive = \tmthree$ with $\tmfour \tosigm \tmsix$, and $\tm = (\la\var\tmfour')\tmfive \tobvm (\la\var\tmseven)\tmfive = \tmtwo$ with $\tmfour = \la\var\tmfour'$ and $\tmfour' \tobvm \tmseven$.
     Indeed $\tmfour$ is $\sigm$-normal by \refrmk{value-sigm-normal}. 
    \end{itemize}

    Therefore, according to the definition of $\tm \tosigm \tmthree$ and \refrmk{alternativedef-sigm}, there are ``only'' eleven cases.
    \begin{itemize}
      \item \emph{Step at the root for $\tm \tobvm \tmtwo$} and \emph{Step inside a $\beta$-context for $\tm \tosigm \tmthree$}, \ie $\tm = (\la\var\tmfour)\val \rtobv \tmfour\isub\var\val = \tmtwo$ and $\tm = (\la\var\tmfour)\val \tosigm (\la\var\tmfour')\val = \tmthree$ with $\tmfour \tosigm \tmfour'$. 
      By \reflemmap{different}{tosigm-sub}, $\tmtwo \neq \tmthree$.

      \item \emph{Application Left for $\tm \tobvm \tmtwo$} and \emph{Step at the root for $\tm \tosigm \tmthree$}, \ie $\tm = \tmfour\tmfive \tobvm \tmsix\tmfive = \tmtwo$ with $\tmfour \tobvm \tmsix$, and $\tm \rtosigma \tmthree$ (see \refrmk{alternativedef-sigm}).
      It is impossible that $\tm \rtosr \tmthree$, otherwise $\tmfour$ would be a value and hence $\betavm$-normal by \refrmk{value-sigm-normal}, but this contradicts that $\tmfour \tobvm \tmsix$.
      Thus, $\tm = (\la\var\tmfour')\tmfour''\tmfive \rtosl (\la\var\tmfour'\tmfive)\tmfour'' = \tmthree$ with $\var \notin \fv{\tmfive}$ and $\tmfour = (\la\var\tmfour')\tmfour''$.
      We claim that $\tmtwo \neq \tmthree$.
      Indeed, if $\tmtwo = \tmthree$ then $\tmfive = \tmfour''$ and $\tmsix = \la\var\tmfour'\tmfive$ with $\tmfour = (\la\var\tmfour')\tmfive \tobvm \la\var\tmfour'\tmfive = \tmsix$, hence necessarily $\tmfour \rtobv \tmsix$ (\ie $\tmfour \tobvm \tmsix$ by a step at the root) and thus $\tmfive$ is a value and $\la\var\tmfour'\tmfive = \tmsix = \tmfour'\isub\var\tmfive$, but this is impossible by \reflemmap{different}{sub}.

      \item \emph{Application Left for $\tm \tobvm \tmtwo$ and $\tm \tosigm \tmthree$}, \ie $\tm = \tmfour\tmfive \tobvm \tmsix\tmfive = \tmtwo$ and $\tm = \tmfour\tmfive \tosigm \tmseven\tmfive = \tmthree$ with $\tmfour \tobvm \tmsix$ and $\tmfour \tosigm \tmseven$.
      By \ih, $\tmsix \neq \tmseven$ and hence $\tmtwo = \tmsix\tmfive \neq \tmseven\tmfive = \tmthree$.

      \item \emph{Application Left for $\tm \tobvm \tmtwo$} and \emph{Application Right for $\tm \tosigm \tmthree$}, \ie $\tm = \tmfour\tmfive \tobvm \tmsix\tmfive = \tmtwo$ and $\tm = \tmfour\tmfive \tosigm \tmfour\tmseven = \tmthree$, with $\tmfour \tobvm \tmsix$ and $\tmfive \tosigm \tmseven$.
      By \reflemmap{different}{tosigm}, $\tmfive \neq \tmseven$ and hence $\tmtwo = \tmsix\tmfive \neq \tmfour\tmseven = \tmthree$.

      \item \emph{Application Right for $\tm \tobvm \tmtwo$} and \emph{Step at the root for $\tm \tosigm \tmthree$}, \ie $\tm = \tmfour\tmfive \tobvm \tmfour\tmsix = \tmtwo$ with $\tmfive \tobvm \tmsix$, and $\tm \rtosigma \tmthree$ %, \ie either $\tm \rtosl \tmthree$ or $\tm \rtosr \tmthree$ 
      (see \refrmk{alternativedef-sigm}).
      \begin{itemize}
        \item If $\tm \rtosl \tmthree$ then $\tm = (\la\var\tmfour')\tmfour''\tmfive \rtosl (\la\var\tmfour'\tmfive)\tmfour'' = \tmthree$ with $\var \notin \fv{\tmfive}$ and $\tmfour = (\la\var\tmfour')\tmfour''$.
	We claim that $\tmtwo \neq \tmthree$.
	Indeed, if $\tmtwo = \tmthree$ then $\tmsix = \tmfour''$ and $\tmfour = \la\var\tmfour'\tmfive$, therefore $(\la\var\tmfour')\tmsix = \tmfour = \la\var\tmfour'\tmfive$ which is impossible.

	\item If $\tm \rtosr \tmthree$ then $\tm = \tmfour ((\la\var\tmfive')\tmfive'') \rtosr (\la\var\tmfour\tmfive')\tmfive'' = \tmthree$ where $\tmfour$ is a value, $\var \notin \fv{\tmfour}$ and $\tmfive = (\la\var\tmfive')\tmfive''$.
	We claim that $\tmtwo \neq \tmthree$.
	Indeed, if $\tmtwo = \tmthree$ then $\tmfour = \la\var\tmfour\tmfive'$ which is impossible.
      \end{itemize}

      \item \emph{Application Right for $\tm \tobvm \tmtwo$ and $\tm \tosigm \tmthree$}, \ie $\tm = \tmfour\tmfive \tobvm \tmsix\tmfive = \tmtwo$ and $\tm = \tmfour\tmfive \tosigm \tmseven\tmfive = \tmthree$ with $\tmfive \tobvm \tmsix$ and $\tmfive \tosigm \tmseven$.
      By \ih, $\tmsix \neq \tmseven$ and hence $\tmtwo = \tmfour\tmsix \neq \tmfour\tmseven = \tmthree$.

      \item \emph{Application Right for $\tm \tobvm \tmtwo$} and \emph{Application Left for $\tm \tosigm \tmthree$}, \ie $\tm = \tmfour\tmfive \tobvm \tmfour\tmsix = \tmtwo$ and $\tm = \tmfour\tmfive \tosigm \tmseven\tmfive = \tmthree$, with $\tmfive \tobvm \tmsix$ and $\tmfour \tosigm \tmseven$.
      By \reflemmap{different}{tosigm}, $\tmfour \neq \tmseven$ and hence $\tmtwo = \tmfour\tmsix \neq \tmseven\tmfive = \tmthree$.
      
     \item \emph{Application Right for $\tm \tobvm \tmtwo$} and \emph{Step inside a $\beta$-context for $\tm \tosigm \tmthree$}, \ie $\tm = \tmfour\tmfive \tobvm \tmfour\tmsix = \tmtwo$ with $\tmfive \tobvm \tmsix$, and $\tm = (\la\var\tmfour')\tmfive \tosigm (\la\var\tmseven)\tmfive = \tmthree$ with $\tmfour = \la\var\tmfour'$ and $\tmfour' \tosigm \tmseven$.
     By \reflemmap{different}{tosigm}, $\tmfour' \neq \tmseven$ whence $\tmfour = \la\var\tmfour' \neq \la\var\tmseven$ and thus $\tmtwo \neq \tmthree$.
     
     \item \emph{Step inside a $\beta$-context for $\tm \tobvm \tmtwo$} and \emph{Step at the root for $\tm \tosigm \tmthree$}, \ie $\tm = (\la\var\tmfour)\tmfive \tobvm (\la\var\tmfour')\tmfive = \tmtwo$ with $\tmfour \tobvm \tmfour'$, and $\tm \rtosigma \tmthree$ (see \refrmk{alternativedef-sigm}).
      It is impossible that $\tm = (\la\var\tmfour)\tmfive \rtosl \tmthree$ because $\la\var\tmfour$ is not an application.
      Thus, $\tm = (\la\var\tmfour)((\la\vartwo\tmfive')\tmfive'') \rtosr (\la\vartwo(\la\var\tmfour)\tmfive')\tmfive'' = \tmthree$ with $\tmfive = (\la\vartwo\tmfive')\tmfive''$ and  $\vartwo \notin \fv{\la\var\tmfour}$, therefore $\tmfive \neq \tmfive''$ and hence $\tmtwo \neq \tmthree$.

     \item \emph{Step inside a $\beta$-context for $\tm \tobvm \tmtwo$} and \emph{Application Right for $\tm \tosigm \tmthree$}, \ie $\tm = \tmfour\tmfive \tosigm \tmfour\tmsix = \tmthree$ with $\tmfive \tosigm \tmsix$, and $\tm = (\la\var\tmfour')\tmfive \tobvm (\la\var\tmseven)\tmfive = \tmtwo$ with $\tmfour = \la\var\tmfour'$ and $\tmfour' \tobvm \tmseven$.
     By \reflemmap{different}{tosigm}, $\tmfive \neq \tmsix$ whence $\tmtwo \neq \tmthree$. 

      \item \emph{Step inside a $\beta$-context for $\tm \tobvm \tmtwo$ and $\tm \tosigm \tmthree$}, \ie $\tm = (\la\var\tmfour)\tmfive \tobv (\la\var\tmsix)\tmfive = \tmtwo$ and $\tm = (\la\var\tmfour)\tmfive \tosigm (\la\var\tmseven)\tmfive = \tmthree$ with $\tmfour \tobvm \tmsix$ and $\tmfour \tosigm \tmseven$. 
      By \ih, $\tmsix \neq \tmseven$ and hence $\tmtwo \neq \tmthree$.
    \end{itemize}
    
    \item In \cite[Proposition~2]{DBLP:conf/fossacs/CarraroG14} it has been proved that $\tosig$ is strongly normalizing, where $\tosig$ is just the extension of $\tosigm$ obtained by allowing reductions under $\l$'s.
    Therefore, $\tosigm \, \subseteq \, \tosig$ and hence $\tosigm$ is strongly normalizing.
    
    The (not strong) confluence of $\tosigm$ has been proved in \cite[Lemma 9.ii]{DBLP:conf/fossacs/CarraroG14}, where $\tosigm$ is denoted by $\to_{\wsym[\sigma]}$.
    \item See \cite[Proposition~10]{DBLP:conf/fossacs/CarraroG14}, where $\toperm$ is denoted by $\tow$.
    \item See \cite[Theorem~24]{DBLP:conf/fossacs/CarraroG14}, where $\toperm$ is denoted by $\tow$.
    \qedhere
  \end{enumerate}  
\end{proof}

\subsubsection{Proofs of Subsection \ref{subsect:vseq} (\texorpdfstring{Open CBV 4: the Value Sequent Calculus $\vseqcalc$}{Open CBV 4: the Value Sequent Calculus})}

We aim to prove the strong confluence of $\tolbarmut$.

Note that values are closed under substitution: for all values $\val, \valtwo$, $\val\isub{\var}{\valtwo}$ is a value.
Moreover, values are $\lambdabar$-, $\mut$- and $\vseq$-normal. 

% !TEX root = main.tex

\begin{definition}\label{def:env-reduction}
  For any $\Rule \in \{\lambdabar, \mut, \vseq\}$, given two environments $\cotm$ and $\cotmtwo$, we define $\cotm \torule \cotmtwo$ by induction on $\cotm$. 
  If $\cotm = \stempty$ then there is no $\cotmtwo$ such that $\cotm \torule \cotmtwo$.
  If $\cotm = \mutilde{\var}{\cm}$ then $\cotmtwo = \mutilde{\var}{\cmtwo}$ and $\cm \torule \cmtwo$.
  If $\cotm = \stacker{\val}{\cotm_0}$ then $\cotmtwo = \stacker{\val}{\cotmtwo_0}$ and $\cotm_0 \torule \cotmtwo_0$.
\end{definition}

\begin{remark}\label{rmk:env-reduction}
  Let $\cm$ and $\cmtwo$ be commands and $\Rule \in \{\lambdabar, \mut, \vseq\}$. 
  One has $\cm \torule \cmtwo$ iff 
  \begin{itemize}
    \item either $\cm = \comm{\la{\var}{\cm_0}}{\stacker{\val}{\cotm}}$ and $\cmtwo = \comm{\val}{\append{(\mutilde{\var}{\cm_0})}{\cotm}}$,
    \item or $\cm = \comm{\val}{\mutilde\var\cm_0}$ and $\cmtwo = \cm_0\isub{\var}\val$,
    \item or $\cm = \comm{\val}{\cotm}$, $\cmtwo = \comm{\val}{\cotmtwo}$ and $\cotm \torule \cotmtwo$.
  \end{itemize}
\end{remark}

\begin{lemma}[Substitution]
\label{l:substitution}
  Let $\cm$ and $\cm'$ be commands, $\env$ and $\env'$ be environments, $\val$ be a value and $\var$ be a variable.
  Let $\Rule \in \{\lambdabar, \mut, \vseq\}$.
  \begin{varenumerate}
    \item  If $\env \torule \env'$ then $\env\isub{\var}{\val} \torule \env'\isub{\var}{\val}$;
    \item If $\cm \torule \cm'$ then $\cm\isub{\var}{\val} \torule \cm'\isub{\var}{\val}$.
  \end{varenumerate}
\end{lemma}

\begin{proof}
  Both points are proved simultaneously by mutual induction on $\cm$ and $\env$.
  Cases:
  \begin{enumerate}
    \item \emph{Step at the root for $\cm \tobvmu \cm'$}, \ie $\cm \defeq \comm{\la{\vartwo}{\cm_0}}{\stacker{\valtwo}{\env_0}} \rtobvmu \comm{\valtwo}{\append{(\mutilde{\vartwo}{\cm_0})}{\env_0}} \eqdef \cm'$. 
    We can suppose without loss of generality that $\vartwo \notin \fv{\val} \cup \{\var\}$. %by the restriction $\alpha \notin \fv{\val}$
    So, $\cm\isub{\var}{\val} = \comm{\la{\vartwo}{\cm_0\isub{\var}{\val}}}{\stacker{\valtwo\isub{\var}{\val}}{\env_0\isub{\var}{\val}}} \tobvmu \comm{\valtwo\isub{\var}{\val}}{\append{(\mutilde{\vartwo}{\cm_0\isub{\var}{\val}})}{\env_0\isub{\var}{\val}}} = \cm'\isub{\var}{\val}$.
    \item \emph{Step at the root for $\cm \tomut \cm'$}, \ie $\cm \defeq \comm{\valtwo}{\mutilde{\vartwo}{\cm_0}} \rtomut \cm_0\isub{\vartwo}{\valtwo} \eqdef \cm'$. 
    We can suppose without loss of generality that $\vartwo \notin \fv{\val} \cup \{\var\}$. %by the restriction $\alpha \notin \fv{\val}$
    So, $\cm\isub{\var}{\val} = \comm{\valtwo\isub{\var}{\val}}{\mutilde{\vartwo}{\cm_0\isub{\var}{\val}}} \tomut \cm_0\isub{\var}{\val} \isub{\vartwo}{\valtwo\isub{\var}{\val}} = \cm'\isub{\var}{\val}$.
    \item \emph{Environment step for $\cm \torule \cm'$}, \ie $\cm \defeq \comm{\valtwo}{\env} \torule \comm{\valtwo}{\env'} \eqdef \cm' $ with $\env \torule \env'$: by \ih, $\env\isub{\var}{\val} \torule \env'\isub{\var}{\val}$ and hence $\cm\isub{\var}{\val} = \comm{\valtwo\isub{\var}{\val}}{\env\isub{\var}{\val}} \torule \comm{\valtwo\isub{\var}{\val}}{\env'\isub{\var}{\val}} = \cm'\isub{\var}{\val}$ according to \refrmk{env-reduction}.
    \item \emph{$\tilde{\mu}$-environment step for $\env \torule \env'$}, \ie $\env \defeq \mutilde{\vartwo}{\cm} \torule \mutilde{\vartwo}{\cmtwo} \eqdef \env'$ with $\cm \torule \cmtwo$. 
    We can suppose without loss of generality that $\vartwo \notin \fv{\val} \cup \{\var\}$.
    By \ih $\cm\isub{\var}{\val} \torule \cmtwo\isub{\var}{\val}$, and thus $\env\isub{\var}{\val} = \mutilde{\vartwo}{\cm\isub{\var}{\val}} \torule \mutilde{\vartwo}{\cmtwo\isub{\var}{\val}} = \env'\isub{\var}{\val}$ according to \refdef{env-reduction}.
    \item \emph{Environment step for $\env \torule \env'$}, \ie $\env \defeq \stacker{\valtwo}{\env_0} \torule \stacker{\valtwo}{\env_0'} \eqdef \env' $ with $\env_0 \torule \env_0'$: by \ih, $\env_0\isub{\var}{\val} \torule \env_0'\isub{\var}{\val}$ and hence $\env\isub{\var}{\val} = \stacker{\valtwo\isub{\var}{\val}}{\env_0\isub{\var}{\val}} \torule \stacker{\valtwo\isub{\var}{\val}}{\env_0'\isub{\var}{\val}} = \env'\isub{\var}{\val}$ according to \refdef{env-reduction}.
    \qedhere
  \end{enumerate}

\end{proof}

\begin{lemma}[Append]
\label{l:cosubstitution}
  Let $\Rule \in \{\lambdabar, \mut, \vseq\}$, $\cm$ be a command and $\env_0$, $\env$ and $\envtwo$ be environments. 
  If $\env \torule \env'$ then $\append{\env_0}{\env} \torule \append{\env_0}{\envtwo}$ and $\append{\cm}{\env} \torule \append{\cm}{\envtwo}$.
\end{lemma}

\begin{proof}
%   By induction on the definition of $\cm$. The key-point is that in $\cm$ there is exactly one free occurrence of $\alpha$.
  We prove simultaneously that $\append{\env_0}{\env} \torule \append{\env_0}{\envtwo}$ and $\append{\cm}{\env} \torule \append{\cm}{\envtwo}$ by mutual induction on $\cm$ and $\env_0$.
  Cases:
  \begin{enumerate}
    \item $\cm = \comm{\val}{\env_0}$: by \ih, $\append{\env_0}{\env} \allowbreak\torule \append{\env_0}{\envtwo}$. 
    Thus, $\append{\cm}{\env} = \comm{\val}{\append{\env_0}{\env}} \torule \comm{\val}{\append{\env_0}{\envtwo}} = \append{\cm}{\envtwo}$ according to \refrmk{env-reduction}.
    
    \item $\env_0 = \stempty$: then, $\append{\env_0}{\env} = \env \torule \envtwo = \append{\env_0}{\envtwo}$.
        
    \item $\env_0 = \stacker{\val_0}{\envtwo_0}$: by \ih, $\append{\envtwo_0}{\env} \torule \append{\envtwo_0}{\envtwo}$. 
    Hence, $\append{\env_0}{\env} = \stacker{\val}{(\append{\envtwo_0}{\env})} \torule \stacker{\val}{(\append{\envtwo_0}{\envtwo})} = \append{\env_0}{\envtwo}$ according to \refdef{env-reduction}.
    
    \item $\env_0 = \mutilde{\vartwo}{\cm}$: we can suppose without loss of generality that $\vartwo \notin \fv{\val} \cup \fv{\env} \cup \{\var\}$, whence $\vartwo \notin \fv{\envtwo}$.
    By \ih, $\append{\cm}{\env} \torule \append{\cm}{\envtwo}$.
    Hence, $\append{\env_0}{\env} = \mutilde{\vartwo}{\append{\cm}{\env}} \torule \mutilde{\vartwo}{\append{\cm}{\envtwo}} = \append{\env_0}{\envtwo}$ according to \refdef{env-reduction}.
    \qedhere
   \end{enumerate}

\end{proof}

% \begin{lemma}[Substitution with append]
% \label{l:substitution-append}
%   Let $\cotm$ and $\cotmtwo$ be  environments, $\val$ be a value, $\cm$ be a command, and $\var \notin \fv{\cotmtwo}$ be a variable.
%   \begin{enumerate}
%     \item \label{p:substitution-append-cotm} \emph{Environments:} $\append{\cotm\isub{\var}{\val}}{\cotmtwo} = (\append{\cotm}{\cotmtwo})\isub{\var}{\val}$.
%     \item \label{p:substitution-append-cm} \emph{Commands:} $\append{\cm\isub{\var}{\val}}{\cotmtwo} = (\append{\cm}{\cotmtwo})\isub{\var}{\val}$.
%   \end{enumerate}
% \end{lemma}
% 
% \begin{proof}
%   Bot points are proved bu mutual induction on $\cm$ and $\cotm$.
%   \begin{itemize}
%     \item $\cm = \comm{\var}{\cotm}$: 
%   \end{itemize}
% 
% \end{proof}

\begin{lemma}[Append Commutes]\hfill
\label{l:append-commutes} % \reflemmap{append-commutes}{cm-steps}
\begin{enumerate}
\item \emph{Evaluation Contexts}: \label{p:append-commutes-ev-ctxs}
$\append{\cmctxp\cm}\cotm = \cmctxp{\append\cm\cotm}$ and $\append{\cotctxp\cotmtwo}\cotm = \cotctxp{\append\cotmtwo\cotm}$.
\item \emph{Rewriting Steps in Commands and Environments}: \label{p:append-commutes-cm-steps}
if $\cm \to_{\lambdamucalc} \cmtwo$ (resp.~$\cotm \to_{\lambdamucalc} \cotmtwo$) then $\append\cm\cotm_0  \tovseq \append\cmtwo\cotm_0$ (resp.~$\append\cotm\cotm_0  \tovseq \append\cotmtwo\cotm_0$).
% \item \emph{Rewriting Steps in Co-Terms}: \label{p:append-commutes-cotm-steps}
% if $\cotm \to_{\lambdamucalc} \cotmtwo$ then $\append\cm\cotm  \to_{\lambdamucalc} \append\cm\cotmtwo$.
\end{enumerate}
\end{lemma}

\begin{proof}\hfill
\begin{enumerate}
\item By mutual induction on $\cmctx$ and $\cotctx$ (see \reffig{lambdamu-calculus}). Cases:
\begin{itemize}
  \item $\cmctx = \ctxhole$: then, $\append{\cmctxp\cm}\cotm = \append{\cm}{\cotm} = \cmctxp{\append\cm\cotm}$.
  \item $\cmctx = \cotctxp{\mutilde\var\cmctxtwo}$: we can suppose without loss of generality that $\var \notin \fv{\cotm}$.
  So, $\append{\cmctxp\cm}\cotm = \append{\cotctxp{\mutilde\var\cmctxtwop{\cm}}}{\cotm} \overset{\ih}{=} \cotctxp{\append{(\mutilde\var\cmctxtwop{\cm})}{\cotm}} = \cotctxp{\mutilde\var{(\append{\cmctxtwop{\cm}}{\cotm})}} \allowbreak \overset{\ih}{=} \cotctxp{\mutilde\var{\cmctxtwop{\append{\cm}{\cotm}}}} \allowbreak=  \cmctxp{\append{\cm}\cotm}$, with we have applied the \ih the first time to $\cotctx$, the second time to $\cmctxtwo$.
  \item $\cotctx = \comm{\val}{\ctxhole}$: then,  $\append{\cotctxp\cotmtwo}\cotm = \append{\comm{\val}{\cotmtwo}}{\cotm} = \comm{\val}{\append{\cotmtwo}{\cotm}} = \cotctxp{\append\cotmtwo\cotm}$.
  \item $\cotctx = \cotctxtwop{\stacker{\val}{\ctxhole}}$: one has $\append{\cotctxp{\cotmtwo}}{\cotm} = \append{\cotctxtwop{\stacker{\val}{\cotmtwo}}}{\cotm} \overset{\ih}{=} \cotctxtwop{\append{(\stacker{\val}{\cotmtwo})}{\cotm}} = \cotctxtwop{\stacker{\val}(\append{\cotmtwo}{\cotm})}$. 
\end{itemize}

\item %By induction on the evaluation context. \ben{To do}.
  By mutual induction on $\cm$ and $\cotm$.
  According to \refrmk{env-reduction}, there are three cases for $\cm \tolbarmut \cmtwo$:
  \begin{itemize}
    \item either $\cm = \comm{\la{\var}{\cm_0}}{\stacker{\val}{\cotm}} \tolbarmut \comm{\val}{\append{(\mutilde{\var}{\cm_0})}{\env}} = \cmtwo$:
    then, $\append\cm\cotm_0 = \comm{\la{\var}{\cm_0}}{\stacker{\val}{(\append{\cotm}{\cotm_0})}} \tolbarmut \comm{\val}{\mutilde\var{\append{\cm_0}{(\append{\cotm}{\cotm_0})}}} = \comm{\val}{\mutilde\var{\append{(\append{\cm_0}{\cotm})}{\cotm_0}}} = \append{\cmtwo}{\cotm_0}$, where the next-to-last identity holds by \reflemmap{append-commutes}{ev-ctxs} taking $\cotctx = \comm{\val}{\mutilde\var{\append{\cm_0}{\ctxhole}}}$;
    \item or $\cm = \comm{\val}{\mutilde\var\cm_0} \tolbarmut \cm_0\isub{\var}\val = \cmtwo$:
    we can suppose without loss of generality that $\var \notin \fv{\cotm_0}$;
    so, $\append{\cm}{\cotm_0} = \comm{\val}{\mutilde\var{(\append{\cm_0}{\cotm_0})}} \tolbarmut (\append{\cm_0}{\cotm_0})\isub\var\val = \append{\cmtwo}{\cotm_0} $;
    \item or $\cm = \comm{\val}{\env} \tolbarmut \comm{\val}{\envtwo} = \cmtwo$ with $\env \tolbarmut \envtwo$: then, $\append{\cm}{\cotm_0} = \comm{\val}{\append{\cotm}{\cotm_0}} \tolbarmut \comm{\val}{\append{\cotmtwo}{\cotm_0}} = \append{\cmtwo}{\cotm_0}$ by \refrmk{env-reduction}.
  \end{itemize}

  According to \refdef{env-reduction}, there are only two cases for $\cotm \tolbarmut \cotmtwo$:  
  \begin{itemize}
    \item either $\env = \mutilde{\var}{\cm}$ and $\envtwo = \mutilde{\var}{\cmtwo}$ and $\cm \tolbarmut \cmtwo$:
    we can suppose without loss of generality that $\var \notin \fv{\cotm_0}$;
    by \ih, $\append\cm\cotm_0 \tolbarmut \append\cmtwo\cotm_0$ and hence $\append \cotm \cotm_0 = \mutilde\var{(\append{\cm}{\cotm_0})} \tolbarmut \mutilde\var{(\append{\cmtwo}{\cotm_0})} = \append\cotmtwo\cotm_0$;
    \item or $\cotm = \stacker{\val}{\cotm_1}$ and $\cotmtwo = \stacker{\val}{\cotmtwo_1}$ with $\cotm_1 \tolbarmut \cotmtwo_1$;
    by \ih, $\append\cotm\cotm_1  \to_{\lambdamucalc} \append\cotmtwo\cotm_1$ and hence $\append\cotm\cotm_0 = \stacker{\val}{(\append{\cotm_1}{\cotm_0})} \to_{\lambdamucalc} \stacker{\val}{(\append{\cotmtwo_1}{\cotm_0})} = \append\cotmtwo\cotm_0$.
    \qedhere
  \end{itemize}

% \item By induction on the evaluation context. \ben{To do}.
\end{enumerate}

\end{proof}

\setcounter{propositionAppendix}{\value{prop:basic-lambdamu}}
\begin{propositionAppendix}[Basic properties of $\lambdamucalc$]
\label{propAppendix:basic-lambdamu}\hfill
\NoteState{prop:basic-lambdamu}
  \begin{enumerate}
    \item\label{pappendix:basic-lambdamu-tobvmu-strong-confluence} $\tobvmu$ is strongly normalizing and strongly confluent.
    \item\label{pappendix:basic-lambdamu-tomut-strong-confluence} $\tomut$ is strongly normalizing and strongly confluent.
    \item\label{pappendix:basic-lambdamu-strong-commutation} $\tobvmu$ and $\tomut$ strongly commute.
    \item $\tovseq$ is strongly confluent, and all $\vseq$-normalizing derivations $\deriv$ from a command $\cm$ or an environment $\env$ (if any) have the same length $\sizevseq{\deriv}$, the same number $\sizemut{\deriv}$ of $\tilde{\mu}$-steps, and the same number $\sizelbar{\deriv}$ of $\lambdabar$-steps.
  \end{enumerate}
\end{propositionAppendix}

\begin{proof}\hfill
  \begin{enumerate}
    \item Note that if $\cm \tobvmu \cmtwo$ then the number of occurrences of $\lambda$ in $\cmtwo$ is strictly less than in $\cm$: this is enough to prove that $\tobvmu$ is strongly normalizing.
    
    Concerning the strong confluence of $\tobvmu$, we prove that 
    \begin{enumerate}
      \item (commands) if $\cm \tobvmu \cm_1$ and $\cm \tobvmu \cm_2$ with $\cm_1 \neq \cm_2$, then there exists $\cm'$ such that $\cm_1 \tobvmu \cm'$ and $\cm_2 \tobvmu \cm'$; 
      \item (environments) if $\env \tobvmu \env_1$ and $\env \tobvmu \env_2$ with $\env_1 \neq \env_2$, then there exists $\cm'$ such that $\env_1 \tobvmu \env'$ and $\env_2 \tobvmu \env'$. 
    \end{enumerate}
    
    The proof is by mutual induction on $\cm$ and $\env$. Cases:
  
    \begin{itemize}
      \item \emph{Step at the root for $\cm \tobvmu \cm_1$ and Step on a $\val$-environment for $\cm \tobvmu \cm_2$}, \ie $\cm \defeq \comm{\la{\var}{\cm_0}}{\stacker{\val}{\env}} \tobvmu \comm{\val}{\append{(\mutilde{\var}{\cm_0})}{\env}} \eqdef \cm_1$ and $\cm \tobvmu \comm{\la{\var}{\cm_0}}{\stacker{\val}{\envtwo}} \eqdef \cm_2$ with $\env \tobvmu \envtwo$.
      Then, $\cm_2 \tobvmu \comm{\val}{\append{(\mutilde{\var}{\cm_0})}{\envtwo}} \eqdef \cm'$. 
      According to the co-substitution lemma (\reflemma{cosubstitution}), $\cm_1 \tobvmu \cm'$.

      \item \emph{Step on an environment for both $\cm \tobvmu \cm_1$ and  $\cm \tobvmu \cm_2$}, \ie $\cm \defeq \comm{\val}{\env} \tobvmu \comm{\val}{\env_1} \eqdef \cm_1$ and $\cm \tobvmu \comm{\val}{\env_2}\eqdef \cm_2$ with $\env_1 \lRew{\lambdabar} \env \tobvmu \env_2$.
      By \ih, there is an environment $\env'$ such that $\env_1 \tobvmu \env' \lRew{\lambdabar} \env_2$. 
      According to ~\refrmk{env-reduction}, $\cm_1 \tobvmu \cm' \lRew{\lambdabar} \cm_2$ by taking $\cm' \defeq \comm{\val}{\env'}$.
      
      \item \emph{Step on a $\tilde{\mu}$-environment for both $\env \tobvmu \env_1$ and $\env \tobvmu \env_2$}, \ie $\env \defeq \mutilde{\alpha}{\cm} \tobvmu \mutilde{\alpha}{\cm_1} \eqdef \env_1$ and $\env \tobvmu \mutilde{\alpha}{\cm_2} \eqdef \env_2$ with $\cm_1 \lRew{\lambdabar} \cm \tobvmu \cm_2$.       
      By \ih, there is a command $\cm'$ such that $\cm_1 \tobvmu \cm' \lRew{\lambdabar} \cm_2$. 
      So, $\env_1 \tobvmu \env' \lRew{\lambdabar} \env_2$ by taking $\env' \defeq \mutilde{\alpha}{\cm'}$, according to \refdef{env-reduction}.
      
      \item \emph{Step on an environment for both $\env \tobvmu \env_1$ and  $\env \tobvmu \env_2$}, \ie $\env \defeq \stacker{\val}{\env'} \tobvmu \stacker{\val}{\env_1'} \eqdef \env_1$ and $\env \tobvmu \stacker{\val}{\env_2'}\eqdef \env_2$ with $\env_1' \lRew{\lambdabar} \env' \tobvmu \env_2'$.
      By \ih, there is an environment $\env_0'$ such that $\env_1' \tobvmu \env_0' \lRew{\lambdabar} \env_2'$. 
      According to ~\refrmk{env-reduction}, $\env_1 \tobvmu \env_0 \lRew{\lambdabar} \env_2$ by taking $\env_0 \defeq \stacker{\val}{\env_0'}$.
    \end{itemize}

    \item The proof of strong normalization of $\tomut$ is in \cite{Herbelin05}.
    
    Concerning the proof of strong confluence of $\tomut$, we prove that:
    \begin{enumerate}
      \item (commands) if $\cm \tomut \cm_1$ and $\cm \tomut \cm_2$ with $\cm_1 \neq \cm_2$, then there exists $\cm'$ such that $\cm_1 \tomut \cm'$ and $\cm_2 \tomut \cm'$; 
      \item (environments) if $\env \tomut \env_1$ and $\env \tomut \env_2$ with $\env_1 \neq \env_2$, then there exists $\cm'$ such that $\env_1 \tomut \env'$ and $\env_2 \tomut \env'$. 
    \end{enumerate}
    
    The proof is by mutual induction on $\cm$ and $\env$. Cases:
  
    \begin{itemize}
      \item \emph{Step at the root for $\cm \tomut \cm_1$ and Step on a $\tilde{\mu}$-environment for $\cm \tomut \cm_2$}, \ie $\cm \defeq \comm{\val}{\mutilde{\var}{\cm_0}} \tomut \cm_0\isub{\var}{\val} \eqdef \cm_1$ and $\cm \tomut \comm{\val}{\mutilde{\var}{\cmtwo'}} \eqdef \cm_2$ with $\cm_0 \tomut \cmtwo'$.
      Then, $\cm_2 \tomut \cmtwo'\isub{\var}{\val} \eqdef \cm'$. 
      According to the substitution lemma (\reflemma{substitution}), $\cm_1 \tomut \cm'$.

      \item \emph{Step on an environment for both $\cm \tomut \cm_1$ and  $\cm \tomut \cm_2$}, \ie $\cm \defeq \comm{\val}{\env} \tomut \comm{\val}{\env_1} \eqdef \cm_1$ and $\cm \tomut \comm{\val}{\env_2}\eqdef \cm_2$ with $\env_1 \lRew{\tilde\mu} \env \tomut \env_2$.
      By \ih, there is an environment $\env'$ such that $\env_1 \tomut \env' \lRew{\tilde\mu} \env_2$. 
      According to ~\refrmk{env-reduction}, $\cm_1 \tomut \cm' \lRew{\tilde\mu} \cm_2$ by taking $\cm' \defeq \comm{\val}{\env'}$.
      
      \item \emph{Step on a $\tilde{\mu}$-environment for both $\env \tomut \env_1$ and $\env \tomut \env_2$}, \ie $\env \defeq \mutilde{\alpha}{\cm} \tomut \mutilde{\alpha}{\cm_1} \eqdef \env_1$ and $\env \tomut \mutilde{\alpha}{\cm_2} \eqdef \env_2$ with $\cm_1 \lRew{\tilde\mu} \cm \tomut \cm_2$.       
      By \ih, there is a command $\cm'$ such that $\cm_1 \tomut \cm' \lRew{\tilde\mu} \cm_2$. 
      So, $\env_1 \tomut \env' \lRew{\tilde\mu} \env_2$ by taking $\env' \defeq \mutilde{\alpha}{\cm'}$.
      
      \item \emph{Step on a environment for both $\env \tomut \env_1$ and  $\env \tomut \env_2$}, \ie $\env \defeq \stacker{\val}{\env'} \tomut \stacker{\val}{\env_1'} \eqdef \env_1$ and $\env \tomut \stacker{\val}{\env_2'}\eqdef \env_2$ with $\env_1' \lRew{\tilde\mu} \env' \tomut \env_2'$.
      By \ih, there is an environment $\env_0'$ such that $\env_1' \tomut \env_0' \lRew{\tilde\mu} \env_2'$. 
      According to ~\refrmk{env-reduction}, $\env_1 \tomut \env_0 \lRew{\tilde\mu} \env_2$ by taking $\env_0 \defeq \stacker{\val}{\env_0'}$.
    \end{itemize}

    \item We prove that 
    \begin{enumerate}
      \item (commands) if $\cm \tomut \cm_1$ and $\cm \tobvmu \cm_2$ then $\cm_1 \neq \cm_2$ and there exists $\cm'$ such that $\cm_1 \tobvmu \cm'$ and $\cm_2 \tomut \cm'$; 
      \item (environments) if $\env \tomut \env_1$ and $\env \tobvmu \env_2$ then $\env_1 \neq \env_2$ and there exists $\cm'$ such that $\env_1 \tobvmu \env'$ and $\env_2 \tomut \env'$. 
    \end{enumerate}
    
    The proof is by mutual induction on $\cm$ and $\env$ (the proof that $\cm_1 \neq \cm_2$ and $\env_1 \neq \env_2$ is left to the reader). Cases:
    \begin{itemize}
%       \item \emph{Step at the root for $\cm \tomut \cm_1$ and Step on a $\tilde{\mu}$-environment for $\cm \tobvmu \cm_2$}, \ie $\cm \defeq \comm{\val}{\mutilde{\var}{\comm{\lbar{\vartwo}{\alpha}{\cmtwo}}{\stacker{\valtwo}{\env}}}} \tomut \mutilde{\var}{\comm{\lbar{\vartwo}{\alpha}{\cmtwo\isub{\var}{\val}}}{\stacker{\valtwo\isub{\var}{\val}}{\env\isub{\var}{\val}}}} \eqdef \cm_1$ (we can suppose without loss of generality that $\vartwo \notin \fv{\val} \cup \{\var\}$) and $\cm \tobvmu \comm{\val}{\mutilde{\var}{\cmtwo\isubtwo{\vartwo}{\valtwo}{\alpha}{\env}}} \eqdef \cm_2$. 
%       Then, $\cm_1 \tobvmu \cmtwo\isub{\var}{\val}\isubtwo{\vartwo}{\valtwo\isub{\var}{\val}}{\alpha}{\env\isub{\var}{\val}} = \cmtwo\isubtwo{\vartwo}{\valtwo}{\alpha}{\env}\isub{\var}{\val} \eqdef \cm'$, and $\cm_2 \tomut \cm'$.
      \item \emph{Step at the root for $\cm \tomut \cm_1$ and Step on a $\tilde{\mu}$-environment for $\cm \tobvmu \cm_2$}, \ie $\cm \defeq \comm{\val}{\mutilde{\var}{\cm_0}} \tomut \cm_0\isub{\var}{\val} \eqdef \cm_1$ and $\cm \tobvmu \comm{\val}{\mutilde{\var}{\cmtwo'}} \eqdef \cm_2$ with $\cm_0 \tobvmu \cmtwo'$. 
      Then, $\cm_2 \tomut \cmtwo'\isub{\var}{\val} \eqdef \cm'$. 
      By  substitution lemma (\reflemma{substitution}), $\cm_1 \tobvmu \cm'$.
      
      \item \emph{Step on a $\val$-environment for $\cm \tomut \cm_1$ and Step at the root for $\cm \tobvmu \cm_2$}, \ie $\cm \defeq \comm{\la{\var}{\cm_0}}{\stacker{\val}{\env}} \tomut \comm{\la{\var}{\cm_0}}{\stacker{\val}{\envtwo}} \eqdef \cm_1$ with $\env \tomut \envtwo$, and $\cm \rtobvmu \comm{\val}{\append{(\mutilde{\var}{\cm_0})}{\env}} \eqdef \cm_2$.
      Then, $\cm_1 \tobvmu \comm{\val}{\append{(\mutilde{\var}{\cm_0})}{\envtwo}} \eqdef \cm'$ and, by append lemma (\reflemma{cosubstitution}), $\cm_2 \tomut \cm'$.
      
      \item \emph{Step on an environment for both $\cm \tomut \cm_1$ and $\cm \tobvmu \cm_2$}, \ie $\cm \defeq \comm{\val}{\env} \tomut \comm{\val}{\env'} \eqdef \cm_1$ and $\cm \tobvmu \comm{\val}{\env''} \eqdef \cm_2$, with $\env \tomut \env'$ and $\env \tobvmu \env''$. 
      By \ih, there exists an environment $\env_0$ such that $\env' \tobvmu \env_0 \lRew{\tilde{\mu}} \env''$, and hence $\cm_1 \tobvmu \cm' \lRew{\tilde{\mu}} \cm_2$ by taking $\cm' \defeq \comm{\val}{\env_0}$, according to \refrmk{env-reduction}.
      
      \item \emph{Step on a $\tilde{\mu}$-environment for both $\env \tomut \env_1$ and $\env \tobvmu \env_2$}, \ie $\env \defeq \mutilde{\var}{\cm} \tomut \mutilde{\var}{\cm_1} \eqdef \env_1$ and $\env \tobvmu \mutilde{\var}{\cm_2} \eqdef \env_2$, with $\cm \tomut \cm_1$ and $\cm \tobvmu \cm_2$. 
      By \ih, there exists a command $\cm_0$ such that $\cm_1 \tobvmu \cm_0 \lRew{\tilde{\mu}} \cm_2$, and hence $\env_1 \tobvmu \env' \lRew{\tilde{\mu}} \env_2$ by taking $\env' \defeq \comm{\val}{\cm_0}$, according to \refdef{env-reduction}.

      \item \emph{Step on a $\val$-environment for both $\env \tomut \env_1$ and $\env \tobvmu \env_2$}, \ie $\env \defeq \stacker{\val}{\env_{0}} \tomut \stacker{\val}{\env_{01}} \eqdef \env_1$ and $\env \tobvmu \stacker{\val}{\env_{02}} \eqdef \env_2$, with $\env_0 \tomut \env_{01}$ and $\env_0 \tobvmu \env_{02}$. 
      By \ih, there exists an environment $\env_0'$ such that $\env_{01} \tobvmu \env_0' \lRew{\tilde{\mu}} \env_{02}$, and hence $\env_1 \tobvmu \env' \lRew{\tilde{\mu}} \env_2$ by taking $\env' \defeq \stacker{\val}{\env_0'}$, according to \refdef{env-reduction}. 
    \end{itemize}

    \item It follows immediately from strong confluence of $\tobvmu$ and $\tomut$ %(\refpropsps{basic-lambdamu}{tobvmu-strong-confluence}{tomut-strong-confluence})
    (\refpropp{basic-lambdamu}{tobvmu-strong-confluence}), strong commutation of $\tobvmu$ and $\tomut$ (\refpropp{basic-lambdamu}{strong-commutation}) and Hindley-Rosen (\reflemma{hindley-rosen}).
    \qedhere

  \end{enumerate}

\end{proof}

\subsection{Proofs of Section~\ref{sect:fireball-vsub} \texorpdfstring{(Quantitative Equivalences of $\firecalc$, $\shufcalc$ and $\vsubcalc$)}{(Quantitative equivalence of Fireball, Shuffling and Value Substitution Calculi)}}

% !TEX root = main.tex
\subsubsection{Proofs of Subsection \ref{subsect:fire-vsub} (\texorpdfstring{Equivalence of $\firecalc$ and $\vsubcalc$}{Equivalence of Fireball and Value Substitution Calculi})}

\begin{remark}
\label{rmk:eqstruct}
  Let $\tm, \tmtwo \in \vsubterms$.
  \begin{enumerate}
    \item\label{p:eqstruct-same-unfolding} If $\tm \eqstruct \tmtwo$ then $\unf{\tm} = \unf{\tmtwo}$.
    \item\label{p:eqstruct-no-vsub} If $\tm \eqstruct \tmtwo$ then $\tm \not\tovsub \tmtwo$ (in particular, $\tm \not\tom \tmtwo$ and $\tm \not\toe \tmtwo$).
  \end{enumerate}
\end{remark}

\setcounter{lemmaAppendix}{\value{l:proj-tof-on-vsub}}
\begin{lemmaAppendix}[Simulation of a $\tof$-Step by $\tovsub$]
\label{lappendix:proj-tof-on-vsub}
  Let $\tm, \tmtwo \in \Lambda$.
\NoteState{l:proj-tof-on-vsub}
  \begin{enumerate}
    \item\label{pappendix:proj-tof-on-vsub-tobv} If $\tm \tobabs \tmtwo$ then $\tm \tom\toe \tmtwo$.

    \item\label{pappendix:proj-tof-on-vsub-toin} If $\tm \toin\tmtwo$ then $\tm \tom\eqstruct \tmthree$, 
  with $\tmthree \!\in\! \vsubterms$ \proper and $\unf\tmthree = \tmtwo$.
    
%     \item \ben{[versione alternativa del punto precedente]} If $\tm \toin\tmtwo$ then $\tm \tom\eqstruct \tmthree\esub\var\gconst$, where $\tmthree$ is a term with no ES and $\tmthree\isub\var\gconst = \tmtwo$.
  \end{enumerate}
\end{lemmaAppendix}

\begin{proof}
  Both proofs are by induction on the rewriting step.
  \begin{enumerate}
    \item According to the definition of  $\tm \tobabs \tmtwo$, there are three cases:
    \begin{itemize}
      \item \emph{Step at the root}, \ie $\tm = (\la\var\tmthree)(\la\vartwo\tmfour) \rtobabs \tmthree \isub\var{\la\vartwo\tmfour} = \tmtwo$: so, $\tm \tom \tmthree \esub\var{\la\vartwo\tmfour} \toe \tmtwo$.
      \item \emph{Application Left}, \ie $\tm = \tmthree\tmfour \tobabs \tmthreep\tmfour = \tmtwo$ with $\tmthree \tobabs \tmthreep$: by \ih, $\tmthree \tom\toe \tmthreep$ and hence $\tm = \tmthree\tmfour \tom\toe \tmthreep\tmfour = \tmtwo$.
      \item \emph{Application Right}, \ie $\tm = \tmthree\tmfour \tobabs \tmthree\tmfourp = \tmtwo$ with $\tmfour \tobabs \tmfourp$: by \ih, $\tmfour \tom\toe \tmfourp$ and hence $\tm = \tmthree\tmfour \tom\toe \tmthree\tmfourp = \tmtwo$.  
    \end{itemize}

    \item According to the definition of  $\tm \toin \tmtwo$, there are three cases:
    \begin{itemize}
      \item \emph{Step at the root}, \ie $\tm = (\la\var\tmthree)\gconst \rtoin \tmthree \isub\var\gconst = \tmtwo$: then, $\tm \tom \tmthree \esub\var\gconst $ where $\tmthree \esub\var\gconst$ is \proper (since $\tmthree \in \Lambda$) and $\unf{\tmthree\esub\var\gconst} = \unf{\tmthree} \isub\var{\unf{\gconst}} = \tmtwo$ ($\unf{\tmthree} = \tmthree$ and $\unf{\gconst} = \gconst$ because $\tmthree, \gconst \in \Lambda$).
      We conclude since $\eqstruct$ is reflexive.
      \item \emph{Application Left}, \ie $\tm = \tmthree \tmfour \toin \tmthreep \tmfour = \tmtwo$ with $\tmthree \toin \tmthreep$: by \ih, $\tmthree \tom\eqstruct \tmfive$ where $\tmfive$ is a \proper $\vsub$-term such that $\unf{\tmfive} = \tmthreep$.
      So, $\tmfive = \tmfive_0\esub{\var_1}{\gconst_1} \dots \esub{\var_n}{\gconst_n}$ where $\tmfive_0 \in \Lambda$ and $\gconst_1, \dots, \gconst_n$ are inert terms (for some $n \in \nat$), moreover we can suppose without loss of generality that $\{\var_1, \dots, \var_n\} \cap \fv{\tmfour} = \emptyset$.
      Let $\tmtwop = (\tmfive_0\tmfour) \esub{\var_1}{\gconst_1} \dots \esub{\var_n}{\gconst_n}$: then, $\tmtwop$ is a \proper $\vsub$-term such that $\tmfive\tmfour \eqstruct \tmtwop$ and, according to \refrmkp{eqstruct}{same-unfolding}, %TODO: \tm \eqstruct \tmtwo  implies \unf{\tm} = \unf{\tmtwo}
      $\unf{\tmtwop} = \unf{(\tmfive\tmfour)} = \unf{\tmfive}\unf{\tmfour} = \tmthreep\tmfour = \tmtwo$. 
      Hence, $\tm = \tmthree\tmfour \tom\eqstruct \tmfive\tmfour \eqstruct \tmtwo'$ and we conclude since $\eqstruct$ is transitive.
      \item \emph{Application Right}, \ie $\tm = \tmthree \tmfour \toin \tmthree \tmfourp = \tmtwo$ with $\tmfour \toin \tmfourp$. Identical to the \emph{application left} case, just switch left and right.
      \qedhere
    \end{itemize}
    
%     \item According to the definition of  $\tm \toin \tmtwo$, there are three cases:
%     \begin{itemize}
% 	  \item \emph{Step at the root}, \ie $\tm = (\la\var\tmthree)\gconst \rtoin \tmthree \isub\var\gconst = \tmtwo$: then, $\tm \tom \tmthree \esub\var\gconst $ where $\tmthree \esub\var\gconst$ is \proper (since $\tmthree \in \Lambda$) and $\unf{\tmthree\esub\var\gconst} = \unf{\tmthree} \isub\var{\unf{\gconst}} = \tmtwo$ ($\unf{\tmthree} = \tmthree$ and $\unf{\gconst} = \gconst$ because $\tmthree, \gconst \in \Lambda$).
% 	  We conclude since $\eqstruct$ is reflexive.
% 	  \item \emph{Application Left}, \ie $\tm = \tmthree \tmfour \toin \tmthreep \tmfour = \tmtwo$ with $\tmthree \toin \tmthreep$: by \ih, $\tmthree \tom\eqstruct \tmfive\esub\var\gconst$ where $\tmfive$ is a term with no ES and $\tmfive\isub\var\gconst = \tmthreep$. Then $\tm = \tmthree \tmfour \tom\eqstruct \tmfive\esub\var\gconst \tmfour \eqstruct (\tmfive \tmfour)\esub\var\gconst$ with $(\tmfive \tmfour)\isub\var\gconst = \tmfive\isub\var\gconst \tmfour = \tmthreep \tmfour = \tmtwo$.
% 	  \item \emph{Application Right}, \ie $\tm = \tmthree \tmfour \toin \tmthree \tmfourp = \tmtwo$ with $\tmfour \toin \tmfourp$. Identical to the \emph{application left} case, just switch left and right.
% 	  \qedhere
%     \end{itemize}
% 
  \end{enumerate}
\end{proof}

\begin{lemma}[Fireballs are Closed Under Anti-Substitution of \Quiet Terms]
\label{l:inert-anti-sub-bis} % \reflemmap{inert-anti-sub}{value}
  Let $\tm$ be a $\vsub$-term and $\gconst$ be an inert term.
  \begin{enumerate}
    \item \label{p:inert-anti-sub-bis-abs} If $\tm\isub\var\gconst$ is an abstraction then $\tm$ is an abstraction. %More precisely, if $\tm\isub\var\gconst = \la\vartwo\tmtwo$ then $\tm = \la\vartwo\tmthree$ where $\tmtwo = \tmthree\isub\var\gconst$ and $\vartwo \notin \fv{\gconst} \cup \{\var\}$;          
    \item \label{p:inert-anti-sub-bis-inert} If $\tm\isub\var\gconst$ is an inert term then $\tm$ is an inert term;
    \item\label{p:inert-anti-sub-bis-fire} If $\tm\isub\var\gconst$ is a fireball then $\tm$ is a fireball.
  \end{enumerate}
\end{lemma}

\begin{proof}\hfill
  \begin{enumerate}
    \item If $\tm\isub\var\gconst = \la\vartwo\tmthree$ then there is $\tmfour$ such that $\tmthree = \tmfour\isub\var\gconst$, that is  $\tm\isub\var\gconst = \la\vartwo(\tmfour\isub\var\gconst) = (\la\vartwo\tmfour)\isub\var\gconst$ and so $\tm = \la\vartwo\tmfour$ is an abstraction;      
    
    \item By induction on the inert structure of $\tm\isub\var\gconst$. Cases:
    \begin{itemize}
      \item \emph{Variable}, \ie $\tm\isub\var\gconst = \vartwo$, possibly with $\var = \vartwo$. Then $\tm = \var$ or $\tm = \vartwo$, and in both cases $\tm$ is inert.
      \item \emph{Compound Inert}, \ie $\tm\isub\var\gconst = \gconsttwo \fire$. If $\tm$ is a variable then it is inert. Otherwise it is an application $\tm = \tmtwo \tmthree$, and so $\tmtwo\isub\var\gconst = \gconsttwo$ and $\tmthree\isub\var\gconst = \fire$. By \ih, $\tmtwo$ is an inert term. Consider $\fire$. Two cases:
      
      \begin{enumerate}
	\item $\fire$ is an abstraction. Then by \refpoint{inert-anti-sub-bis-abs} $\tmthree$ is an abstraction.
	\item $\fire$ is an inert term. Then by \ih $\tmthree$ is an inert term.
      \end{enumerate}
      
      In both cases $\tmthree$ is a fireball, and so $\tm = \tmtwo \tmthree$ is an inert term.      
      \end{itemize}
    
    \item Immediate consequence of \reflemmasps{inert-anti-sub-bis}{abs}{inert}, since every fireball is either an abstraction or an inert term.
    \qedhere
  \end{enumerate}
\end{proof}

\begin{lemma}[Substitution of \Quiet Terms Does Not Create $\betaf$-Redexes]
\label{l:inert-anti-red-bis} % \reflemmap{inert-anti-red}{value}
  Let $\tm, \tmtwo$~be terms and $\gconst$ be an inert term.
  There is $\tmthree \in \Lambda$ such that:
  \begin{enumerate}
    \item \label{p:inert-anti-red-bis-betav} if $\tm\isub\var\gconst \tobabs \tmtwo$ then 
    $\tm \tobabs \tmthree$ and $\tmthree\isub\var\gconst = \tmtwo$;

    \item \label{p:inert-anti-red-bis-inert} if $\tm\isub\var\gconst \toin \tmtwo$ then   $\tm \toin \tmthree$ and $\tmthree\isub\var\gconst = \tmtwo$.
  \end{enumerate}
\end{lemma}

\begin{proof}
  We prove the two points by induction on the evaluation context closing the root redex. Cases:  
    \begin{itemize}
      \item \emph{Step at the root}:
      \begin{enumerate}
      \item \emph{Abstraction Step}, \ie $\tm\isub\var\gconst \defeq (\la\vartwo\tmfour\isub\var\gconst) \tmfive\isub\var\gconst \allowbreak\rtobabs\allowbreak \tmfour\isub\var\gconst\isub\vartwo{\tmfive\isub\var\gconst} \eqdef \tmtwo$.
      By \reflemmap{inert-anti-sub-bis}{abs}, $\tmfive$ is an abstraction, since $\tmfive\isub\var\gconst$ is an abstraction by hypothesis. Then $\tm = (\la\vartwo\tmfour) \tmfive \rtobabs \tmfour\isub\vartwo{\tmfive}$. Then $\tmthree \defeq \tmfour\isub\var{\tmfive}$ verifies the statement, as $\tmthree\isub\var\gconst = (\tmfour\isub\vartwo{\tmfive})\isub\var\gconst = \tmfour\isub\var\gconst\isub\vartwo{\tmfive\isub\var\gconst} = \tmtwo$.     
      
      \item \emph{Inert Step}, identical to the abstraction subcase, just replace \emph{abstraction} with \emph{inert term} and the use of  \reflemmap{inert-anti-sub-bis}{abs} with the use of \reflemmap{inert-anti-sub-bis}{inert}.
      \end{enumerate}      
      
      \item \emph{Application Left}, \ie $\tm = \tmfour\tmfive$ and reduction takes place in $\tmfour$:
      \begin{enumerate}
      \item \emph{Abstraction Step}, \ie $\tm\isub\var\gconst \defeq \tmfour\isub\var\gconst \tmfive\isub\var\gconst \tobabs \tmsix \tmfive\isub\var\gconst \eqdef \tmtwo$. 
      By \ih there exists $\tmthreep \in \Lambda$ such that $\tmsix = \tmthreep\isub\var\gconst$ and $\tmfour \tobabs \tmthreep$. Then $\tmthree \defeq \tmthreep \tmfive$ satisfies the statement, as $\tmthree\isub\var\gconst = (\tmthreep \tmfive)\isub\var\gconst = \tmthreep\isub\var\gconst \tmfive\isub\var\gconst = \tmtwo$.
      
      \item \emph{Inert Step}, identical to the abstraction subcase.
      \end{enumerate}

      \item \emph{Application Right}, \ie $\tm = \tmfour\tmfive$ and reduction takes place in $\tmfive$. Identical to the \emph{application left} case, just switch left and right.
      \qedhere
    \end{itemize}
      
\end{proof}

\setcounter{lemmaAppendix}{\value{l:proj-via-unfold}}
\begin{lemmaAppendix}[Projection of a $\betaf$-Step on $\tovsub$ via Unfolding]
\label{lappendix:proj-via-unfold}
  Let%
\NoteState{l:proj-via-unfold}
  $\tm$ be a \proper $\vsub$-term and $\tmtwo$ be a term. 
  \begin{enumerate}
    \item \label{pappendix:proj-via-unfold-tobv} If $\unf{\tm} \tobabs\! \tmtwo$ then $\tm \tom\toe \tmthree$, 
%     where $\tmthree$ is a \proper $\vsub$-term such that $\unf\tmthree = \tmtwo$.
    with $\tmthree \!\in\! \vsubterms$ \proper s.t.~$\unf\tmthree \!= \tmtwo$.

    \item\label{pappendix:proj-via-unfold-toin} If $\unf{\tm} \toin\! \tmtwo$ then $\tm \tom \eqstruct \tmthree$, 
%     where $\tmthree$ is a \proper $\vsub$-term such that $\unf\tmthree = \tmtwo$.
    with $\tmthree \!\in\! \vsubterms$ \proper s.t.~$\unf\tmthree \!= \tmtwo$.
  \end{enumerate}
\end{lemmaAppendix}

\begin{proof}
  Since $\tm$ is \proper, there are a $\l$-term $\tmfive$ and some inert $\l$-terms $\gconst_1, \dots, \gconst_n$ (with $n \in \nat$) such that $\tm = \tmfive \esub{\var_1}{\gconst_1}\dots \esub{\var_n}{\gconst_n}$. 
  We prove both points by induction on $n \in \nat$. 
  The base case (\ie $n=0$) is given by the simulation of one-step reductions given by \reflemma{proj-tof-on-vsub}, since $\tm = \tmfive \in \Lambda$ and hence $\unf{\tm} = \tm$ (recall that, when applying \reflemmap{proj-tof-on-vsub}{tobv}, $\tmtwo \in \Lambda$ implies that $\tmtwo$ is \proper and $\unf{\tmtwo} = \tmtwo$).

%   So we treat only 
  Consider now %the case 
  $n>0$. Let $\tm_{n-1} \defeq \tmfive \esub{\var_1}{\gconst_1} \dots \esub{\var_{n-1}}{\gconst_{n-1}}$: so, $\tm = \tm_{n-1} \esub{\var_n}{\gconst_n}$ and $\unf{\tm} = \unf{\tm_{n-1}}\isub{\var_n}{\gconst_n}$. 
  Both points rely on the fact that the substitution of \quiet terms cannot create redexes (\reflemma{inert-anti-red-bis}). Namely,
%  Recall that $\Lambda \subseteq \vsubterms$, and $\tmtwo \in \Lambda$ implies that $\tmtwo$ is \proper and $\unf{\tmtwo} = \tmtwo$.
  
  \begin{enumerate}
    \item \emph{$\betaabs$-step}: the application of \reflemmap{inert-anti-red-bis}{betav} to $\unf{\tm} = \unf{\tm_{n-1}}\isub{\var_n}{\gconst_n} \allowbreak\tobabs \tmtwo$ (since $\unf{\tm_{n-1}}\in \Lambda$, \ie it has no ES) provides $\tmfour \in \Lambda$ such that $\unf{\tm_{n-1}} \tobabs \tmfour$ and $\tmfour \isub{\var_n}{\gconst_n} = \tmtwo$. 
    By \ih, $\tm_n \tom\toe \tmthree$ where $\tmthree$ is a \proper $\vsub$-term such that $\unf{\tmthree} = \tmfour$, and thus $\tm = \tm_{n-1} \esub{\var_n}{\gconst_n}  \tom\toe \tmthree\esub{\var_n}{\gconst_n}$.
    Moreover, $\tmthree\esub{\var_n}{\gconst_n}$ is \proper and $\unf{\tmthree\esub{\var_n}{\gconst_n}} = \unf{\tmthree}\isub{\var_n}{\gconst_n} = \tmfour\isub{\var_n}{\gconst_n} = \tmtwo$.

    \item \emph{$\betain$-step}: the application of \reflemmap{inert-anti-red-bis}{inert} to $\unf{\tm} = \unf{\tm_{n-1}}\isub{\var_n}{\gconst_n} \allowbreak\toin \tmtwo$ provides $\tmfour \in \Lambda$ such that $\unf{\tm_{n-1}} \toin\! \tmfour$ and $\tmfour \isub{\var_n}{\gconst_n} \!= \tmtwo$. 
   By \ih, $\tm_{n-1} \tom\eqstruct \tmthree$ where $\tmthree$ is a \proper $\vsub$-term such that $\unf{\tmthree} = \tmfour$;
      thus, $\tm = \tm_{n-1} \esub{\var_n}{\gconst_n}  \tom\eqstruct \tmthree\esub{\var_n}{\gconst_n}$. Moreover, $\tmthree\esub{\var_n}{\gconst_n}$ is \proper and $\unf{\tmthree\esub{\var_n}{\gconst_n}} = \unf{\tmthree}\isub{\var_n}{\gconst_n} = \tmfour\isub{\var_n}{\gconst_n} = \tmtwo$.
      \qedhere
  \end{enumerate}
\end{proof}

\begin{lemmaAppendix}
\label{lappendix:normal-anti-unfold}
  Let $\tm$
  \NoteState{l:normal-anti-unfold}
  be a \proper $\vsub$-term.
  If $\unf{\tm}$ is a fireball, then $\tm$ is $\set{\msym,\expoabs}$-normal and its body is a fireball.
\end{lemmaAppendix}

\begin{proof}
 First, we prove that if $\unf{\tm}$ is a fireball then for some fireball $\fire$ and inert terms $\gconst_1, \dots, \gconst_n$ one has $\tm = \fire \esub{\var_1}{\gconst_1} \dots \esub{\var_n}{\gconst_n}$. Since $\tm$ is \proper, there are a term $\tmtwo$ and some inerts terms $\gconst_1, \dots, \gconst_n$ (with $n \in \nat$) such that $\tm = \tmtwo \esub{\var_1}{\gconst_1}\dots \esub{\var_n}{\gconst_n}$. 
  We prove by induction on $n \in \nat$ that $\tmtwo$ is a fireball.

  If $n = 0$, then $\tm = \tmtwo \in \Lambda$, thus $\tmtwo = \unf{\tm}$ and hence $\tmtwo$ is a fireball by hypothesis.

  Suppose $n > 0$ and let $\tmthree \defeq \tmtwo \esub{\var_1}{\gconst_1} \dots \esub{\var_{n-1}}{\gconst_{n-1}} $, which is a \proper $\vsub$-term: then, $\tm = \tmthree\esub{\var_n}{\gconst_n}$ and hence $\unf{\tm} = \unf{\tmthree} \isub{\var_n}{\gconst_n}$ (as $\unf{\gconst_n} = \gconst_n$ because $\gconst_n \in \Lambda$).
  By \reflemmap{inert-anti-sub-bis}{fire}, $\unf{\tmthree}$ is a fireball.
  By \ih, $\tmtwo$ is a fireball.
  
  So, we have just proved that $\tm = \fire \esub{\var_1}{\gconst_1} \dots, \esub{\var_n}{\gconst_n}$ for some fireball $\fire$ and inert terms $\gconst_1, \dots \gconst_n$.
%   Now, fireballs are clearly $\vsub$-normal, then $\tm$ can only have $\expovar$-redexes.
  Now, fireballs (in particular, inert terms) are $\vsub$-normal. 
  Indeed, fireballs are without ES and hence without $\expo$-redexes, moreover it is easy to prove that fireballs are $\mult$-normal (by simply adapting the proof of \reflemma{fnormal}).

  Therefore, $\tm = \fire \esub{\var_1}{\gconst_1} \dots \esub{\var_n}{\gconst_n}$ can only have $\expovar$-redexes (when some $\gconst_k$ is a variable).
\end{proof}

\setcounter{lemmaAppendix}{\value{l:toevar-post}}
\begin{lemmaAppendix}[Linear Postponement of $\toevar$]
\label{lappendix:toevar-post}
  Let% 
\NoteState{l:toevar-post}
  $\tm, \tmtwo, \tmthree \in \Lambda_\vsub$.
  \begin{enumerate}
      \item\label{pappendix:toevar-post-tom} If $\tm \toevar \tmthree \tom \tmtwo$ then $\tm\tom\toevar\tmtwo$. 
      \item\label{pappendix:toevar-post-toeabs} If $\tm\toevar\toeabs\tmtwo$ then %$\tm\toeabs\toevar\tmtwo$ or $\tm\toeabs\toeabs\tmtwo$. 
      $\tm \toeabs\toe \tmtwo$.
      \item\label{pappendix:toevar-post-global} If $\deriv \colon \tm\tovsub^*\tmtwo$ then $\derivtwo \colon \tm\Rew{\msym,\expoabs}^*\toevar^*\tmtwo$ with $\sizevsub{\derivtwo} = \sizevsub{\deriv}$, $\sizem{\derivtwo} = \sizem{\deriv}$, $\sizee{\derivtwo} = \sizee{\deriv}$, and $\sizeeabs\derivtwo \geq \sizeeabs\deriv$. 
  \end{enumerate}
\end{lemmaAppendix}

\begin{proof}
  \begin{enumerate}
%     \item If $\tm \toevar \tmthree \tom \tmtwo$ then 
%     \begin{align*}
%       \tm &\defeq \evctxp{\tmfour\esub\var{\vartwo\esub{\var_1}{\tm_1}\dots\esub{\var_n}{\tm_n}}} \\
% 	  &\toevar \evctxp{\tmfour\isub\var{\vartwo}\esub{\var_1}{\tm_1}\dots\esub{\var_n}{\tm_n}} \eqdef \tmthree
%     \end{align*}
%     and it is evident that the $\expovar$-step cannot create in $\tmthree$ new $\mult$-redexes not occurring in $\tm$, therefore the $\mult$-redex fired in $\tmthree \tom \tmtwo$ is (a residual of a $\mult$-redex) already occurring in $\tm$. 
%     Thus, there exists $\tmfour \in \Lambda_\vsub$ such that $\tm \tom \tmfour$
    \item By induction on the definition of $\tm \toevar \tmthree$. 
    Since the $\expovar$-step cannot create in $\tmthree$ new $\mult$-redexes not occurring in $\tm$, the $\mult$-redex fired in $\tmthree \tom \tmtwo$ is (a residual of a $\mult$-redex) already occurring in $\tm$. So, there are the following cases.
    \begin{itemize}
      \item \emph{Step at the Root for $\tm \toevar \tmthree$ and ES Left for $\tmthree \tom \tmtwo$}, \ie $\tm \defeq \tmfour\esub\varthree{\sctxp\var} \toevar \sctxp{\tmfour\isub\varthree{\var}} \eqdef \tmthree$ and $\tmthree \tom \sctxp{\tmfourp\isub\varthree{\var}} \eqdef \tmtwo$ with $\tmfour \tom \tmfourp$: then $\tm \tom \tmfourp\esub\varthree{\sctxp{\var}} \toevar \tmtwo$;
      \item \emph{Step at the Root for $\tm \toevar \tmthree$ and ES ``quasi-Right'' for $\tmthree \tom \tmtwo$}, \ie $\tm \defeq \tmfour\esub\varthree{\var\esub{\var_1}{\tm_1}\dots\esub{\var_n}{\tm_n}} \toevar \tmfour\isub\varthree{\var}\esub{\var_1}{\tm_1}\dots\esub{\var_n}{\tm_n} \eqdef \tmthree$ and 
      \begin{equation*}
        \tm \tom \tmfour\esub\varthree{\var\esub{\var_1}{\tm_1}\dots\esub{\var_j}{\tmp_j}\dots\esub{\var_n}{\tm_n}} \eqdef \tmtwo
      \end{equation*}
      for some $n > 0$, and $\tm_j \tom \tmp_j$ for some $1 \leq j \leq n$: then, $\tm \tom \tmfour\esub\varthree{\var\esub{\var_1}{\tm_1}\dots\esub{\var_j}{\tmp_j}\dots\esub{\var_n}{\tm_n}} \toevar \tmtwo$;

      \item \emph{Application Left for $\tm \toevar \tmthree$} and \emph{Application Right for $\tmthree \tom \tmtwo$}, \ie $\tm \defeq \tmfour\tmfive \toevar \tmfourp\tmfive \eqdef \tmthree$ and $\tmthree \tom \tmfourp\tmfivep \eqdef \tmtwo$ with $\tmfour \toevar \tmfourp$ and $\tmfive \tom \tmfivep$: then, $\tm \tom \tmfour\tmfivep \toevar \tmtwo$;
      \item \emph{Application Left for both $\tm \toevar \tmthree$ and $\tmthree \tom \tmtwo$}, \ie $\tm \defeq \tmfour\tmfive \toevar \tmfourp\tmfive \eqdef \tmthree$ and $\tmthree \tom \tmfour''\tmfive \eqdef \tmtwo$ with $\tmfour \toevar \tmfourp$ and $\tmfourp \tom \tmfour''$: by \ih, $\tmfour \tom\toevar \tmfour''$, hence $\tm \tom\toevar \tmtwo$;
      \item \emph{Application Left for $\tm \toevar \tmthree$} and \emph{Step at the Root for $\tmthree \tom \tmtwo$}, \ie $\tm \defeq (\la\var\tmfive){\esub{\var_1}{\tm_1}\dots\esub{\var_n}{\tm_n}}\tmfour \toevar (\la\var\tmfive){\esub{\var_1}{\tm_1}\dots\esub{\var_j}{\tmp_j}\dots\esub{\var_n}{\tm_n}}\tmfour \eqdef \tmthree$ with $n > 0$ and $\tm_j \toevar \tmp_j$ for some $1 \leq j \leq n$, and 
      \begin{equation*}
        \tmthree \tom \allowbreak \tmfive\esub\var\tmfour{\esub{\var_1}{\tm_1}\dots\esub{\var_j}{\tmp_j}\dots\esub{\var_n}{\tm_n}} \eqdef \tmtwo
      \end{equation*}
      then,
      \begin{equation*}
        \tm \tom \tmfive\esub\var\tmfour{\esub{\var_1}{\tm_1}\dots\esub{\var_n}{\tm_n}} \toevar \tmtwo;
      \end{equation*}
      \item \emph{Application Right for $\tm \toevar \tmthree$} and \emph{Application Left for $\tmthree \tom \tmtwo$}, \ie $\tm \defeq \tmfive\tmfour \toevar \tmfive\tmfourp \eqdef \tmthree$ and $\tmthree \tom \tmfivep\tmfourp \eqdef \tmtwo$ with $\tmfour \toevar \tmfourp$ and $\tmfive \tom \tmfivep$: then, $\tm \tom \tmfivep\tmfour \toevar \tmtwo$;
      \item \emph{Application Right for both $\tm \toevar \tmthree$ and $\tmthree \tom \tmtwo$}, \ie $\tm \defeq \tmfive\tmfour \toevar \tmfive\tmfourp \eqdef \tmthree$ and $\tmthree \tom \tmfive\tmfour'' \eqdef \tmtwo$ with $\tmfour \toevar \tmfourp$ and $\tmfourp \tom \tmfour''$: by \ih, $\tmfour \tom\toevar \tmfour''$, hence $\tm \tom\toevar \tmtwo$;
      \item \emph{Application Right for $\tm \toevar \tmthree$} and \emph{Step at the Root for $\tmthree \tom \tmtwo$}, \ie $\tm \defeq \sctx{\la\var\tmfive}\tmfour \toevar \sctxp{\la\var\tmfive}\tmfourp \eqdef \tmthree$ with $\tmfour \toevar \tmfourp$, and $\tmthree \tom \allowbreak \sctxp{\tmfive\esub\var\tmfourp} \eqdef \tmtwo$: then, $\tm \tom \sctxp{\tmfive\esub\var\tmfour} \toevar \tmtwo$;
      \item \emph{ES Left for $\tm \toevar \tmthree$} and \emph{ES Right for $\tmthree \tom \tmtwo$}, \ie $\tm \defeq \tmfour\esub\var\tmfive \toevar \tmfourp\esub\var\tmfive \eqdef \tmthree$ and $\tmthree \tom \tmfourp\esub\var\tmfivep \eqdef \tmtwo$ with $\tmfour \toevar \tmfourp$ and $\tmfive \tom \tmfivep$: then, $\tm \tom \tmfour\esub\var\tmfivep \toevar \tmtwo$;
      \item \emph{ES Left for both $\tm \toevar \tmthree$ and $\tmthree \tom \tmtwo$}, \ie $\tm \defeq \tmfour\esub\var\tmfive \toevar \tmfourp\esub\var\tmfive \eqdef \tmthree$ and $\tmthree \tom \tmfour''\esub\var\tmfive \eqdef \tmtwo$ with $\tmfour \toevar \tmfourp$ and $\tmfourp \tom \tmfour''$: by \ih, $\tmfour \tom\toevar \tmfour''$, hence $\tm \tom\toevar \tmtwo$;
      \item \emph{ES Right for $\tm \toevar \tmthree$} and \emph{ES Left for $\tmthree \tom \tmtwo$}, \ie $\tm \defeq \tmfive\esub\var\tmfour \toevar \tmfive\esub\var\tmfourp \eqdef \tmthree$ and $\tmthree \tom \tmfivep\esub\var\tmfourp \eqdef \tmtwo$ with $\tmfour \toevar \tmfourp$ and $\tmfive \tom \tmfivep$: then, $\tm \tom \tmfivep\esub\var\tmfour \toevar \tmtwo$;
      \item \emph{ES Right for both $\tm \toevar \tmthree$ and $\tmthree \tom \tmtwo$}, \ie $\tm \defeq \tmfive\esub\var\tmfour \toevar \tmfive\esub\var{\tmfourp} \eqdef \tmthree$ and $\tmthree \tom \tmfive\esub\var{\tmfour''} \eqdef \tmtwo$ with $\tmfour \toevar \tmfourp$ and $\tmfourp \tom \tmfour''$: by \ih, $\tmfour \tom\toevar \tmfour''$, hence $\tm \tom\toevar \tmtwo$.
    \end{itemize}

    \item By induction on the definition of $\tm \toevar \tmthree$. 
    Since the $\expovar$-step cannot create in $\tmthree$ new $\expoabs$-redexes not occurring in $\tm$, the $\expoabs$-redex fired in $\tmthree \toeabs \tmtwo$ is (a residual of a $\expoabs$-redex) already occurring in $\tm$. So, there are the following cases.
    \begin{itemize}
      \item \emph{Step at the Root for both $\tm \toevar \tmthree$ and $\tmthree \toeabs \tmtwo$}, \ie 
      \begin{equation*}
        \tm \defeq \tmfour\esub\var{\sctxtwop{\varthree}\esub\vartwo{\sctxp{\la\var\tmfive}}} \toevar \sctxtwop{\tmfour\isub\var\varthree}\esub\vartwo{\sctxp{\la\var\tmfive}} \allowbreak\eqdef \tmthree
      \end{equation*}
      and $\tmthree \toeabs \sctxp{\sctxtwop{\tmfour\isub\var\varthree}\isub\vartwo{\la\var\tmfive}} \eqdef \tmtwo$ (with possibly $\vartwo = \varthree$). 
      We set $\sctxthree \defeq \sctxtwo\isub\vartwo{\la\var\tmfive}$ \ie $\sctxthree$ is the substitution context obtained from $\sctxtwo$ by the capture-avoiding substitution of $\la\var\tmfive$ for each free occurrence of $\vartwo$ in $\sctxtwo$.
      We can suppose without loss of generality that $\vartwo \notin \fv{\sctx} \cup \fv{\tmfour}$.
%       Thus, $\tmtwo = \sctxp{\sctxthreep{\tmfour \isub\var{\varthree}}}$.
      There are two sub-cases:
      \begin{itemize}
        \item either $\vartwo = \varthree$ and then $\tm \toeabs \tmfour \esub\var{\sctxp{\sctxthreep{\la\var\tmfive}}} \toeabs \sctxp{\sctxthreep{\tmfour \isub\var{\la\var\tmfive}}} \allowbreak= \tmtwo$,
        \item or $\vartwo \neq \varthree$ and then $\tm \toeabs\! \tmfour \esub\var{\sctxp{\sctxthreep\varthree}} \toevar\! \sctxp{\sctxthreep{\tmfour \isub\var{\varthree}}} \allowbreak= \tmtwo$.
      \end{itemize}

      \item \emph{Step at the Root for $\tm \toevar \tmthree$ and ES Left for $\tmthree \toeabs \tmtwo$}, \ie $\tm \defeq \tmfour\esub\varthree{\sctxp\var} \toevar \sctxp{\tmfour\isub\varthree{\var}} \eqdef \tmthree$ and $\tmthree \toeabs \sctxp{\tmfourp\isub\varthree{\var}} \eqdef \tmtwo$ with $\tmfour \toeabs \tmfourp$: then $\tm \toeabs \tmfourp\esub\varthree{\sctxp{\var}} \toevar \tmtwo$;
      \item \emph{Step at the Root for $\tm \toevar \tmthree$ and ES ``quasi-Right'' for $\tmthree \toeabs \tmtwo$}, \ie $\tm \defeq \tmfour\esub\varthree{\var\esub{\var_1}{\tm_1}\dots\esub{\var_n}{\tm_n}} \toevar \tmfour\isub\varthree{\var}\esub{\var_1}{\tm_1}\dots\esub{\var_n}{\tm_n} \eqdef \tmthree$ for some $n > 0$, and $\tm_j \toeabs \tmp_j$ for some $1 \leq j \leq n$, and 
      \begin{equation*}
        \tmthree \toeabs \tmfour\isub\varthree\var \esub{\var_1}{\tm_1}\dots\esub{\var_j}{\tmp_j}\dots\esub{\var_n}{\tm_n} \eqdef \tmtwo \: 
      \end{equation*}
      then, $\tm \toeabs \tmfour\esub\varthree{\var\esub{\var_1}{\tm_1}\dots\esub{\var_j}{\tmp_j}\dots\esub{\var_n}{\tm_n}} \toevar \tmtwo$;

      \item \emph{Application Left for $\tm \toevar \tmthree$ and Application Right for $\tmthree \toeabs \tmtwo$}, \ie $\tm \defeq \tmfour\tmfive \toevar \tmfourp\tmfive \eqdef \tmthree$ and $\tmthree \toeabs \tmfourp\tmfivep \eqdef \tmtwo$ with $\tmfour \toevar \tmfourp$ and $\tmfive \toeabs \tmfivep$: then, $\tm \toeabs \tmfour\tmfivep \toevar \tmtwo$;
      \item \emph{Application Left for both $\tm \toevar \tmthree$ and $\tmthree \toeabs \tmtwo$}, \ie $\tm \defeq \tmfour\tmfive \toevar \tmfourp\tmfive \eqdef \tmthree$ and $\tmthree \toeabs \tmfour''\tmfive \eqdef \tmtwo$ with $\tmfour \toevar \tmfourp$ and $\tmfourp \toeabs \tmfour''$: by \ih, $\tmfour \toeabs\toe \tmfour''$, hence $\tm \toeabs\toe \tmtwo$;
      \item \emph{Application Right for $\tm \toevar \tmthree$ and Application Left for $\tmthree \toeabs \tmtwo$}, \ie $\tm \defeq \tmfive\tmfour \toevar \tmfive\tmfourp \eqdef \tmthree$ and $\tmthree \toeabs \tmfivep\tmfourp \eqdef \tmtwo$ with $\tmfour \toevar \tmfourp$ and $\tmfive \toeabs \tmfivep$: then, $\tm \toeabs \tmfivep\tmfour \toevar \tmtwo$;
      \item \emph{Application Right for both $\tm \toevar \tmthree$ and $\tmthree \toeabs \tmtwo$}, \ie $\tm \defeq \tmfive\tmfour \toevar \tmfive\tmfourp \eqdef \tmthree$ and $\tmthree \toeabs \tmfive\tmfour'' \eqdef \tmtwo$ with $\tmfour \toevar \tmfourp$ and $\tmfourp \toeabs \tmfour''$: by \ih, $\tmfour \toeabs\toe \tmfour''$, hence $\tm \toeabs\toe \tmtwo$;

      \item \emph{ES Left for $\tm \toevar \tmthree$ and Step at the Root for $\tmthree \toeabs \tmtwo$}, \ie $\tm \defeq \tmfour\esub\varthree{\sctxp{\la\vartwo\tmfive}} \toevar \tmfourp\esub\varthree{\sctxp{\la\vartwo\tmfive}} \eqdef \tmthree$ and $\tmthree \toeabs \sctxp{\tmfourp\isub\varthree{\la\vartwo\tmfive}} \eqdef \tmtwo$ with $\tmfour \toevar \tmfourp$: 
      this means that in $\tmfour$ there is an ES of the form $\esub\vartwo\var$ (possibly $\var = \varthree$) which is fired in $\tmfour \toevar \tmfourp$; 
      then, $\tm \toeabs \sctxp{\tmfour\isub\varthree{\la\vartwo\tmfive}} \toe \tmtwo$, where the last $\expo$-step is a $\expoabs$-step if $\var = \varthree$, otherwise it is a $\expovar$-step;
      
      \item \emph{ES Left for $\tm \toevar \tmthree$ and ES Right for $\tmthree \toeabs \tmtwo$}, \ie $\tm \defeq \tmfour\esub\var\tmfive \toevar \tmfourp\esub\var\tmfive \eqdef \tmthree$ and $\tmthree \toeabs \tmfourp\esub\var\tmfivep \eqdef \tmtwo$ with $\tmfour \toevar \tmfourp$ and $\tmfive \toeabs \tmfivep$: then, $\tm \toeabs \tmfour\esub\var\tmfivep \toevar \tmtwo$;
      \item \emph{ES Left for both $\tm \toevar \tmthree$ and $\tmthree \toeabs \tmtwo$}, \ie $\tm \defeq \tmfour\esub\var\tmfive \toevar \tmfourp\esub\var\tmfive \eqdef \tmthree$ and $\tmthree \toeabs \tmfour''\esub\var\tmfive \eqdef \tmtwo$ with $\tmfour \toevar \tmfourp$ and $\tmfourp \toeabs \tmfour''$: by \ih, $\tmfour \toeabs\toe \tmfour''$, so $\tm \toeabs\toe \tmtwo$;
      
      \item \emph{ES Right for $\tm \toevar \tmthree$ and Step at the Root for $\tmthree \toeabs \tmtwo$}, \ie 
      \begin{align*}
        \tm &\defeq \tmfour\esub\varthree{(\la\vartwo\tmfive)\esub{\var_1}{\tm_1}\dots\esub{\var_n}{\tm_n}} \\
	    &\toevar \tmfour\esub\varthree{(\la\vartwo\tmfive) \esub{\var_1}{\tm_1}\dots\esub{\var_j}{\tmp_j}\dots\esub{\var_n}{\tm_n}} \eqdef \tmthree
      \end{align*}
      for some $n > 0$, and $\tm_j \toevar \tmp_j$ for some $1 \leq j \leq n$, and $\tmthree \toeabs \tmfour\isub\varthree{\la\vartwo\tmfive} \esub{\var_1}{\tm_1}\dots\esub{\var_j}{\tmp_j}\dots\esub{\var_n}{\tm_n} \eqdef \tmtwo$: 
      then, 
      \begin{equation*}
        \tm \toeabs \tmfour\isub\varthree{\la\vartwo\tmfive}\esub{\var_1}{\tm_1}\dots\esub{\var_n}{\tm_n}  \allowbreak\toevar \tmtwo\,;
      \end{equation*}

      \item \emph{ES Right for $\tm \toevar \tmthree$ and ES Left for $\tmthree \toeabs \tmtwo$}, \ie $\tm \defeq \tmfive\esub\var\tmfour \toevar \tmfive\esub\var\tmfourp \eqdef \tmthree$ and $\tmthree \toeabs \tmfivep\esub\var\tmfourp \eqdef \tmtwo$ with $\tmfour \toevar \tmfourp$ and $\tmfive \toeabs \tmfivep$: then, $\tm \toeabs \tmfivep\esub\var\tmfour \toevar \tmtwo$;
      \item \emph{ES Right for both $\tm \toevar \tmthree$ and $\tmthree \toeabs \tmtwo$}, \ie $\tm \defeq \tmfive\esub\var\tmfour \toevar \tmfive\esub\var{\tmfourp} \eqdef \tmthree$ and $\tmthree \toeabs \tmfive\esub\var{\tmfour''} \eqdef \tmtwo$ with $\tmfour \toevar \tmfourp$ and $\tmfourp \toeabs \tmfour''$: by \ih, $\tmfour \toeabs\toe \tmfour''$, hence $\tm \toeabs\toe \tmtwo$.
    \end{itemize}
    
    \item By induction on $\sizevsub{\deriv} \in \nat$, using Lemmas~\ref{lappendix:toevar-post}.\ref{pappendix:toevar-post-tom}-\ref{pappendix:toevar-post-toeabs} in the inductive case.
    \qedhere
  \end{enumerate}

\end{proof}

\setcounter{theoremAppendix}{\value{thm:sim-f-into-vsubeq}}
\begin{theoremAppendix}[Quantitative Simulation of $\firecalc$ in $\vsubcalc$]
\label{thmappendix:sim-f-into-vsubeq}
  Let%
\NoteState{thm:sim-f-into-vsubeq}
  $\tm, \tmtwo \!\in\! \Lambda$.
  If $\deriv \colon \tm \tof^*\! \tmtwo$ then %there exists $\derivtwo \colon \tm \tovsub^* \eqstruct \tmthree \in \vsubterms$ such that
  there are $\tmthree, \tmfour \!\in\! \vsubterms$ and $\derivtwo \colon \tm \tovsub^*\! \tmfour$ such that
  \begin{enumerate}
  \item\label{pappendix:sim-f-into-vsubeq-qual} \emph{Qualitative Relationship}: %$\tmtwo = \unf{\tmthree}$ and $\tmthree$ is \proper;
  $\tmfour \eqstruct \tmthree$, $\tmtwo = \unf{\tmthree} = \unf{\tmfour}$ and $\tmthree$ is \proper;

  \item \emph{Quantitative Relationship}: \label{pappendix:sim-f-into-vsubeq-quant} 
    \begin{enumerate}
    \item \label{pappendix:sim-f-into-vsubeq-quant-mult} \emph{Multiplicative Steps:} $\sizef{\deriv} = \sizem{\derivtwo}$;
    \item \label{pappendix:sim-f-into-vsubeq-quant-exp} \emph{Exponential (Abstraction) Steps:} $\sizebabs{\deriv} = \sizeeabs{\derivtwo} = \sizee{\derivtwo}$.
    
    \end{enumerate}
    
    \item \emph{Normal Forms}: if $\tmtwo$ is $\betaf$-normal then there exists $\derivthree\colon\tmfour \toevar^* \tmfive$ such that $\tmfive$ is a $\vsub$-normal form and $\sizeevar{\derivthree} \leq \sizem{\derivtwo} - \sizeeabs{\derivtwo}$.
    \end{enumerate}   
\end{theoremAppendix}

\begin{proof}
The first two points are proved together.
\begin{enumerate}
 \item[1-2.] By the remark at the beginning of this section of the Appendix (\refrmksps{eqstruct}{same-unfolding}{no-vsub}), it is sufficient to show that there exists $\derivtwo \colon \tm \tovsub^* \eqstruct \tmthree \in \vsubterms$ such that $\tmtwo = \unf{\tmthree}$ with $\tmthree$ \proper, and $\sizef{\deriv} = \sizem{\derivtwo}$ and $\sizebabs{\deriv} = \sizeeabs{\derivtwo}$ (the fact that $\sizeeabs{\derivtwo} = \sizee{\derivtwo}$ is immediate, since the simulation obtained by iterating the projection in \reflemma{proj-via-unfold} never uses $\toevar$).
%     The proof of the other points is by induction on $\sizef{\deriv} \in \nat$. 
  We proceed by induction on $\sizef{\deriv} \in \nat$.  
  Cases:
  \begin{itemize}
    \item \emph{Empty derivation}, \ie $\sizef{\deriv} = 0$ then $\tm = \tmtwo$ and $\sizebabs{\deriv} = 0$, so we conclude taking $\tmthree \defeq \tmtwo$ and $\derivtwo$ as the empty derivation.
%     \item Suppose $\sizef{\deriv} > 0$: then, $\deriv \colon \tm \tof \tmp \tof^* \tmtwo$ and let $\deriv' \colon \tmp \tof \tmtwo$ be the sub-reduction sequence obtained from $\deriv$ by removing its first component $\tm$. 
%     By \ih, there is $\derivtwo_1 \colon \tm \tovsubeq^* \tmthreep$ with $\tmtwo = \unf{\tmthreep}$ and $\sizef{\deriv} = \sizem{\derivtwo_1}$, $\sizebv{\deriv} \leq \sizee{\derivtwo_1} \leq \sizef{\deriv}$ and $\sizef{\deriv} + \sizebv{\deriv} \leq \sizevsub{\derivtwo_1} \leq 2\sizef{\deriv}$.
%     Either $\tm \tobv \tmp$ or $\tm \toin \tmp$.
   \item \emph{Non-empty derivation}, \ie $\sizef{\deriv} > 0$: then, $\deriv \colon \tm \tof^* \tmfour \tof \tmtwo$ and let $\deriv' \colon \tm \tof^* \tmfour$ be the derivation obtained from $\deriv$ by removing its last step $\tmfour \tof \tmtwo$.
    By \ih, there is $\derivtwo' \colon \tm \tovsub^*\eqstruct \tmfive$ such that $\tmfour = \unf{\tmfive}$, $\tmfive$ is \proper, $\sizef{\deriv'} = \sizem{\derivtwo'}$, and $\sizebabs{\deriv'} = \sizeeabs{\derivtwo'}$.
    By applying \reflemma{proj-via-unfold} to the last step $\tmfour \tof \tmtwo$ of $\deriv$, we obtain $\tmthree$ such that either $\tmfive \tom\toe \tmthree$, if $\unf{\tmfive} \tobv \tmtwo$, or $\tmfive \tom \eqstruct \tmthree$, if $\unf{\tmfive} \toin \tmtwo$, and in both cases $\tmthree$ is a \proper $\vsub$-term such that $\unf{\tmthreep} = \tmtwo$. Note that both cases can be summed up with $\tmfive \tom \overset{\eqstruct\ }{\toe} \tmthree$.
    Composing the two obtained derivations $\derivtwo' \colon \tm \tovsub^*\eqstruct \tmfive$ and $\tmfive \tom \overset{\eqstruct\ }{\toe} \tmthree$, we obtain the derivation $\derivtwo'' \colon \tm \tovsub^*\eqstruct \tmfive \tom \overset{\eqstruct\ }{\toe} \tmthree$ that satisfies the quantitative relationships but not yet the qualitative one, as $\eqstruct$ appears between two steps of $\derivtwo''$. It is then enough to apply the strong bisimulation property of $\eqstruct$ (\reflemmap{eqstruct-post-and-term}{globpost}), that provides a derivation $\derivtwo: \tm \tovsub^* \tom \overset{\eqstruct\ }{\toe} \eqstruct \tmthree$ with the same quantitative properties of $\derivtwo''$.
    \end{itemize}
 \item[3.] If $\tmtwo$ is $\betaf$-normal then it is a fireball (by open harmony, \refprop{open-harmony}) and so $\tmthree$ is $\set{\msym,\expoabs}$-normal by \reflemma{normal-anti-unfold}. By \refpropp{basic-value-substitution}{tom-toe-terminates}, $\toevar$ terminates and so there are $\tmsix$ and a derivation $\derivfour\colon\tmthree \toevar^* \tmsix$ such that $\tmsix$ is a $\expovar$-normal form. 
 If $\tmsix$ is not a $\vsub$-normal form, then it has a $\set{\msym,\expoabs}$-redex, but by postponement of $\toevar$ (\reflemma{toevar-post}) such a redex was already in $\tmthree$, against hypothesis. So $\tmsix$ is a $\vsub$-normal form. Then we have $\tmfour \eqstruct \tmthree \toevar^* \tmsix$. 
 Postponing $\eqstruct$ (\reflemmasps{eqstruct-post-and-term}{globpost}{normal}), we obtain that there exists a $\vsub$-normal form $\tmfive$ and a derivation $\derivthree\colon\tmfour \toevar^* \tmfive \eqstruct \tmsix$.
 
 To estimate the length of $\derivthree$ consider $\derivtwo$ followed by $\derivthree$, \ie $\derivtwo;\derivthree \colon \tm \Rewn{\msym,\expoabs} \tmfour \toevar^* \tmfive$. By \refpropp{basic-value-substitution}{expo-less-than-mult}, $\sizee{\derivtwo;\derivthree} \leq \sizem{\derivtwo;\derivthree} = \sizem{\derivtwo}$, and since  $\sizee{\derivtwo;\derivthree} = \sizeeabs{\derivtwo;\derivthree} + \sizeevar{\derivtwo;\derivthree} = \sizeeabs{\derivtwo} + \sizeevar\derivthree$ we obtain $\sizeeabs\derivtwo + \sizeevar{\derivthree} \leq \sizem{\derivtwo}$, \ie $\sizeevar{\derivthree} \leq \sizem{\derivtwo} - \sizeeabs\derivtwo$.\qedhere
\end{enumerate}

  %Since structural equivalence $\eqstruct$ preserves $\set{\msym,\expoabs}$-normal forms ({normal}), $\tmthree$ is also $\set{\msym,\expoabs}$-normal. 
\end{proof}

\setcounter{corollaryAppendix}{\value{coro:equivalence-vsub-fire-termination}}
\begin{corollaryAppendix}[Linear Termination Equivalence of $\vsubcalc$ and $\firecalc$]
\label{coroAppendix:equivalence-vsub-fire-termination}
  Let
\NoteState{coro:equivalence-vsub-fire-termination}
  $\tm \in \Lambda$. There exists a $\betaf$-normalizing derivation $\deriv$ from $\tm$ iff there exists a $\vsub$-normalizing derivation $\derivtwo$ from $\tm$. Moreover, $\sizef{\deriv} \leq \sizevsub{\derivtwo} \leq 2\sizef{\deriv}$, \ie they are linearly related.
\end{corollaryAppendix}

\begin{proof}\hfill
  \begin{description}
       \item[$\Rightarrow$:] Let $\deriv \colon \tm \tof^* \tmtwo$ be a $\betaf$-normalizing derivation and $\derivtwo \colon \tm \tovsub^*\allowbreak\toevar^* \tmfive$ be the composition of its projection in $\vsubcalc$ with the extension to a $\expovar$-derivation with $\tmfive$ $\vsub$-normal, according to \refthm{sim-f-into-vsubeq}. 
       Then $\derivtwo$ is a $\vsub$-normalizing derivation from $\tm$.
       
       \item[$\Leftarrow$:] By contradiction, suppose that there is a diverging $\betaf$-derivation from $\tm$ in $\firecalc$. By \refthm{sim-f-into-vsubeq} it projects to a $\vsub$-derivation in $\vsubcalc$ that is at least as long as the one in $\firecalc$, absurd.
%       Then, there is a infinite reduction sequence $\deriv \defeq (\tm_i)_{i \in \nat}$ such that $\tm_0 = \tm$ and $\tm_i \tof \tm_{i+1}$ for any $i \in \nat$.
%       For every $n \in \nat$, let $\deriv_n$ be the sub-reduction sequence of $\deriv$ starting from $\tm_0 = \tm$ and such that $\sizef{\deriv_n} = n$:
%       since $\tof$ can be simulated by $\tovsub$ in such a way that the number of $\tof$-steps are equal to the number of $\tom$-steps (\refpropsps{sim-f-into-vsubeq}{qual}{quant-mult}) and $\unf{\tm} = \tm$ (as $\tm \in \Lambda$), there are $\tmtwo_n \in \vsubterms$ and a reduction sequence $\derivtwo_n \colon \tm \tovsub^* \tmtwo_n$ such that $\sizem{\derivtwo_n} = n$ and hence $\sizevsub{\derivtwo_n} \geq n$.
%       Therefore, $\tm$ is not $\vsub$-normalizable, absurd.
  \end{description}
  
  About lengths, %we have 
  $\sizef{\deriv} \!\leq \sizevsub{\derivtwo}$ since $\sizem{\derivtwo} \!=\! \sizef{\deriv}$ (\refthmp{sim-f-into-vsubeq}{quant}).
%   and $\sizevsub{\derivtwo} \leq 2\sizef{\deriv}$ because 
  By \refpropp{basic-value-substitution}{expo-less-than-mult}, $\sizee{\derivtwo} \leq \sizem{\derivtwo}$ and so $\sizevsub{\derivtwo} = \sizem{\derivtwo} + \sizee{\derivtwo} \leq 2\sizef{\deriv}$.
  \qedhere
\end{proof}
% !TEX root = main.tex
\subsubsection{Proofs of Subsection \ref{subsect:shuf-vsub} (\texorpdfstring{Equivalence of $\shufcalc$ and $\vsubcalc$}{Equivalence of Shuffling and Value Substitution Calculi})}

\begin{lemma}[Simulation of a $\vmsym$-Step on $\vsubcalc$]
\label{l:simul-one-shuf-step-on-vsub}
  Let $\tm, \tmtwo \in \Lambda$.
  \begin{enumerate}
    \item\label{p:simul-one-shuf-step-on-vsub-sigm} If $\tm \tosigm \tmtwo$ then there exist $\tmthree, \tmfour \in \Lambda_\vsub$ s.t. $\tm\tom^+\tmthree \eqstruct \tmfour\lRewp{\mult}\tmtwo$.
    
    \item\label{p:simul-one-shuf-step-on-vsub-betavm} If $\tm \tobvm \tmtwo$ then there exists $\tmthree \in \Lambda_\vsub$ s.t. $\tm\tom^+ \toe \tmthree \lRewn{\mult}\tmtwo$.
  \end{enumerate}
\end{lemma}

\begin{proof}
  \begin{enumerate}
    \item %By induction on $\tm \tosigm \tmtwo$. Cases:
    By induction on the definition of $\tm \tosigm \tmthree$, following \refrmk{alternativedef-sigm}. There are four cases:
    
    \begin{enumerate}
      \item \emph{Step at the root}, %. One case for each root case of $\tosigm$:
      \ie $\tm \rtosigma \tmtwo$. %According to \refrmk{alternativedef-sigm}:
      \begin{enumerate}
      	\item %$\tosl$: The hypothesis is 
      	either $\tm \defeq (\la\var\tmfive)\tmthree\tmfour \rtosl (\la\var\tmfive \tmfour)\tmthree \eqdef \tmtwo$ with $\var \notin \fv{\tmfour}$, and then $\tm = (\la\var\tmfive)\tmthree\tmfour \tom \tmfive\esub\var\tmthree \tmfour \eqstruct (\tmfive\tmfour)\esub\var\tmthree \lRew{\mult} (\la\var{\tmfive \tmfour})\tmthree = \tmtwo$;
	\item %$\tosr$: The hypothesis is 
	or $\tm \defeq \val ((\la\var\tmthree)\tmfour) \rtosr (\la\var\val \tmthree)\tmfour \eqdef \tmtwo$ with $\var \notin \fv{\val}$ and then $\tm = \val ((\la\var\tmthree)\tmfour) \tom \val (\tmthree\esub\var\tmfour) \eqstruct (\val \tmthree)\esub\var\tmfour \lRew\mult (\la\var\val \tmthree)\tmfour = \tmtwo$.
      \end{enumerate}
      
      \item \emph{Application Left}, \ie $\tm \defeq \tmthree\tmfour \tosigm \tmfive\tmfour \eqdef \tmtwo$ with $\tmthree \tosigm \tmfive$. 
      The result follows by the \ih, as $\tom$ and $\eqstruct$ are closed by applicative contexts.
      
      \item \emph{Application Right}, \ie $\tm \defeq \tmthree\tmfour \tosigm \tmthree\tmfive \eqdef \tmtwo$ with $\tmfour \tosigm \tmfive$.
      The result follows by the \ih, as $\tom$ and $\eqstruct$ are closed by applicative contexts.
      
      \item \emph{Inside a $\beta$-context}, %. The hypothesis is 
      \ie $\tm \defeq (\la\var\tmthree)\tmfour \tosigm (\la\var\tmfive)\tmfour \eqdef \tmtwo$ with $\tmthree \tosigm \tmfive$. 
      By \ih, $\tmthree\tom^+\tmthree' \eqstruct \tmfive'\lRewp{\mult}\tmfive$. 
      Now, $\tom$ and $\eqstruct$ are not closed by balanced contexts, but it is enough to apply a further $\tom$ step to the balanced context (as $\tom$ and $\eqstruct$ are instead closed by substitution contexts), obtaining $\tm = (\la\var\tmthree)\tmfour \tom \tmthree\esub\var\tmfour \tom^+ \tmthree'\esub\var\tmfour \eqstruct \tmfive'\esub\var\tmfour \lRewp\mult \tmfive\esub\var\tmfour \lRew\mult (\la\var\tmfive)\tmfour = \tmtwo$.
    \end{enumerate}
    \item By induction on the definition of $\tm \tobvm \tmtwo$, there are four cases:

    \begin{enumerate}
      \item \emph{Step at the root}, \ie $\tm = (\la\var\tmfour)\val \rtobv \tmfour\isub\var\val = \tmtwo$. So, $\tm \tom \tmfour\esub\var\val \toe \tmtwo$.
      
      \item \emph{Application Left}. It follows by the \ih, as $\tom$ and $\toe$ are closed by applicative contexts.
      
      \item \emph{Application Right}. It follows by the \ih, as $\tom$ and $\toe$ are closed by applicative contexts.
      
      \item \emph{Step inside a $\beta$-context}, \ie $\tm = (\la\var\tmthree)\tmfour \tobvm (\la\var\tmfive)\tmfour = \tmtwo$ with $\tmthree \tobvm \tmfive$. 
      By \ih, $\tmthree\tom^+\toe \tmsix\lRewn{\mult}\tmfive$. 
      Now, $\tom$ and $\toe$ are not closed by balanced contexts, but it is enough to apply a further $\tom$ step to the balanced context (as $\tom$ and $\toe$ are instead closed by substitution contexts), obtaining $(\la\var\tmthree)\tmfour \tom \tmthree\esub\var\tmfour \tom^+ \toe \tmsix\esub\var\tmfour \lRewn\mult \tmfive\esub\var\tmfour \lRew\mult (\la\var\tmfive)\tmfour$.      
    \qedhere
    \end{enumerate}
  \end{enumerate}

\end{proof}

\setcounter{lemmaAppendix}{\value{l:tovm-mproj}}
\begin{lemmaAppendix}[Projecting a $\vmsym$-Step on $\tovsubeq$ via $\mult$-nf]%[Permutation steps project on structural equivalence]
\label{lappendix:tovm-mproj}
  Let 
\NoteState{l:tovm-mproj}
  $\tm, \tmtwo \!\in\! \Lambda$.
  \begin{enumerate}
    \item \label{pappendix:tosig-mproj-on-eqstruct-mdev} If $\tm \tosigm \tmtwo$ then $\mnf\tm \eqstruct \mnf\tmtwo$.
    
    \item \label{pappendix:tobvm-mproj-on-toe-mdev} If $\tm \tobvm \tmtwo$ then $\mnf\tm \toe\tom^* \mnf\tmtwo$.
  \end{enumerate}
\end{lemmaAppendix}

\begin{proof}\hfill
  \begin{enumerate}
    \item  By \reflemmap{simul-one-shuf-step-on-vsub}{sigm} there exist $\tmthree, \tmfour \in \Lambda_\vsub$ s.t. $\tm\tom^+\tmthree \eqstruct \tmfour\lRewp{\mult}\tmtwo$. 
    By existence and uniqueness of the $\mult$-normal form (\refpropsps{basic-value-substitution}{tom-toe-terminates}{tom-toe-strong-confluence} and \refpropp{basic-confluence}{unique-normal}), $\tmthree \tom^+ \mnf\tmthree = \mnf\tm$. By \reflemmap{eqstruct-post-and-term}{globpost}, there is $\tmfive \in \Lambda_\vsub$ s.t. $\tmfour \tom^+\tmfive \eqstruct \mnf\tm$. 
% 	By \reflemmap{eqstruct-post-and-term}{locpost}, $\mult$-redexes are preserved by $\eqstruct$ and so $\tmfive$ is $\mult$-normal. 
    By \reflemmap{eqstruct-post-and-term}{normal}, $\tmfive$ is $\mult$-normal;
    in particular, $\tmfive = \mnf\tmfour = \mnf\tmtwo$ according to \refpropp{basic-confluence}{unique-normal}.
    Thus, $\mnf{\tm} \eqstruct q = \mnf{\tmtwo}$.

    \item  By \reflemmap{simul-one-shuf-step-on-vsub}{betavm} there are $\tmthree, \tmfour \in \Lambda_\vsub$ such that $\tm\tom^+ \tmthree \toe \tmfour\lRewn{\mult}\tmtwo$. 
    By existence and uniqueness of the $\mult$-normal form (\refpropsps{basic-value-substitution}{tom-toe-terminates}{tom-toe-strong-confluence} and \refpropp{basic-confluence}{unique-normal}), $\mnf\tmthree = \mnf\tm$. 
    As $\mnf\tm \lRewn{\mult\!} \tmthree \toe \tmfour$, there is $\tmfive \in \Lambda_\vsub$ s.t. $\mnf\tm \toe\tmfive \lRewn\mult \tmfour$ according to strong commutation of $\tom$ and $\toe$ (\refpropp{basic-value-substitution}{tom-toe-commute}). 
    Thus, $\mnf\tm \toe \tmfive \lRewn\mult \tmtwo$ and hence $\mnf\tm \toe \tom^*\mnf \tmtwo$ since $\mnf\tmtwo = \mnf \tmfive$ by \refpropp{basic-confluence}{unique-normal}.
    \qedhere
  \end{enumerate}
\end{proof}

\setcounter{lemmaAppendix}{\value{l:normal-vsub-shuffling}}
\begin{lemmaAppendix}[Projection %of $\shuf$-Normal Forms on $\vsub$-Normal Forms
Preserves Normal Forms]
\label{lappendix:normal-vsub-shuffling}
  Let% 
\NoteState{l:normal-vsub-shuffling}
  $\tm \in \Lambda$. 
  If $\tm$ is $\vmsym$-normal then $\mnf{\tm}$ is $\vsub$-normal.
\end{lemmaAppendix}

\begin{proof}
  In \cite[Prop.~12]{DBLP:conf/fossacs/CarraroG14} (where the reduction $\toshuf$ is denoted by $\to_\wsym$) it has been shown that:
  \begin{enumerate}
    \item a term is $\shuf$-normal iff it is of the form $\shufnf$, 
    \item a term is $\shuf$-normal and is neither a value nor a $\beta$-redex (\ie of the form $(\la\var\tm)\tmtwo$) iff it is of the form $\const$, 
  \end{enumerate}
  where the forms $\shufnf$ and $\const$ are defined by mutual induction as follows:
  \begin{align*}
    \const 	&\grameq \var\val \, \mid \, \var\const \, \mid \, \const\shufnf &
    \shufnf	&\grameq \val \, \mid \, \const \, \mid \, (\la\var\shufnf)\const.
  \end{align*}

  The idea is the following: on the one hand, not only terms of the form $\const$ are not values but also they cannot reduce to value through $\mult$-derivations; on the other hand, any $\mult$-derivation from a term of the form $\shufnf$ cannot create an ES of the form $\esub\var{\sctxp{\val}}$, therefore the $\expo$-normality of $\shufnf$ (which is without ES) is preserved in its $\mult$-normal form $\mnf{\shufnf}$ and hence $\mnf{\shufnf}$ is $\vsub$-normal.
  
  More formally, consider the types $\const_\vsub$ and $\shufnf_\vsub$ of $\vsub$-terms defined by mutual induction as follows ($\val$ is a value, without ES):
  \begin{align*}
    \const_\vsub 	&\grameq \var\val \, \mid \, \var\const_\vsub \, \mid \, \const_\vsub\shufnf_\vsub \\
    \shufnf_\vsub	&\grameq \val \mid \const_\vsub \mid \shufnf_\vsub\esub\var{\const_\vsub}.
  \end{align*}
  
  First, we prove by mutual induction on $\const$ and $\shufnf$ that the $\mult$-normal form $\mnf{\const}$ of $\const$ is of the form $\const_\vsub$, and the $\mult$-normal form $\mnf{\shufnf}$ of $\shufnf$ is of the form $\shufnf_\vsub$.
  The base cases are $\mnf{\val} = \val$ (since $\tom$ does not reduce under $\l$'s) and $\mnf{\var\val} = \var\val$. Inductive cases:
  \begin{enumerate}
    \item $\mnf{\var\const} = \var\mnf{\const} = \var\const_\vsub$ where $\mnf{\const} = \const_\vsub$ by \ih,
    \item $\mnf{\const\shufnf} = \mnf{\const}\mnf{\shufnf} = \const_\vsub\shufnf_\vsub$ (since $\const_\vsub$ is not an abstraction) where $\mnf{\const} = \const_\vsub$ and $\mnf{\shufnf} = \shufnf_\vsub$ by \ih,
    \item $\mnf{(\la\var\shufnf)\const} = \mnf{\shufnf}\esub\var{\mnf{\const}} = \shufnf_\vsub\esub\var{\const_\vsub}$ (since $\const_\vsub$ is not of the form $\sctxp{\val}$) where $\mnf{\const} = \const_\vsub$ and $\mnf{\shufnf} = \shufnf_\vsub$ by \ih.
  \end{enumerate}

  To conclude the proof of \reflemma{normal-vsub-shuffling}, it is sufficient to observe that all terms of type $\shufnf_\vsub$ are $\vsub$-normal, see \cite[Lemma 5]{DBLP:conf/flops/AccattoliP12} (where $\tovsub$ is denoted by $\to_\wsym$).
\end{proof}

\setcounter{theoremAppendix}{\value{thm:sim-shuf-into-vsubeq}}
\begin{theoremAppendix}[Quantitative Simulation of $\shufcalc$ in $\vsubcalc$\!]
\label{thmappendix:sim-shuf-into-vsubeq}
  Let%
\NoteState{thm:sim-shuf-into-vsubeq}
  $\tm, \tmtwo \in \Lambda$.
  If $\deriv \colon \tm \toshuf^* \tmtwo$ then there are $\tmthree \in \vsubterms$ and $\derivtwo \colon \tm \tovsub^* \tmthree$ such that
  \begin{enumerate}
    \item\label{pappendix:sim-shuf-into-vsubeq-qual} \emph{Qualitative Relationship}: $\tmthree \eqstruct \mnf{\tmtwo}$;

    \item\label{pappendix:sim-shuf-into-vsubeq-quant} \emph{Quantitative Relationship}: $\sizebshuf{\deriv} = \sizee{\derivtwo}$;
    \item\label{pappendix:sim-shuf-into-vsubeq-normal} \emph{Normal Forms}: if $\tmtwo$ is $\shuf$-normal then $\tmthree$, $\mnf{\tmtwo}$ are $\vsub$-normal.
  \end{enumerate}
\end{theoremAppendix}

\begin{proof}
  First, by straightforward induction on $\sizevm{\deriv} \!\in \nat$ using %\reflemmasp{tosig-mproj-on-eqstruct}{mdev}{tobvm-mproj-on-toe}{mdev}
  the projection via $\mult$-normal forms (\reflemmasps{tovm-mproj}{sigm-on-eqstruct}{betavm-on-tovsub}), one proves that there is $\derivtwo_1 \colon \mnf\tm \tovsubeq^* \mnf\tmtwo$ with $\sizee{\derivtwo_1} = \sizebvm{\deriv}$.
  By postponement of $\eqstruct$ (\reflemmap{eqstruct-post-and-term}{globpost}), there is $\derivtwo_2 \colon \mnf\tm \tovsub^* \eqstruct \mnf\tmtwo$ with $\sizee{\derivtwo_2} = \sizee{\derivtwo_1}% = \sizebvm{\deriv}
  $.
  Clearly, $\tm \tom^* \mnf{\tm}$.
  It easy to check that $\tmthree \eqstruct \tmfour$ implies $\tmthree \not\toe \tmfour$ for all $\tmthree, \tmfour \in \Lambda_\vsub$.
  Therefore, there exist $\tmthree \in \Lambda_\vsub$ and $\derivtwo \colon \tm \tovsub^* \tmthree$ such that $\tmthree \eqstruct \mnf{\tmtwo}$ and $\sizee{\derivtwo} = \sizee{\derivtwo_2} = \sizebvm{\deriv}$.
  
  Finally, if moreover $\tmtwo$ is $\shuf$-normal then, since normal forms are preserved by multiplicative projection (\reflemma{normal-vsub-shuffling}), $\mnf{\tmtwo}$ is $\vsub$-normal, and hence so is $\tmthree$ (\reflemmap{eqstruct-post-and-term}{normal}, because $\tmthree \eqstruct \mnf{\tmtwo}$).
  \qedhere
\end{proof}

\setcounter{corollaryAppendix}{\value{coro:equivalence-vsub-shuf-termination}}
\begin{corollaryAppendix}[Termination Equivalence of $\vsubcalc$ and $\shufcalc$]
\label{coroappendix:equivalence-vsub-shuf-termination}
  Let%
\NoteState{coro:equivalence-vsub-shuf-termination}
  $\tm \in \Lambda$. 
  There is a $\vmsym$-normalizing derivation $\deriv$ from $\tm$ iff there is a $\vsub$-normalizing derivation $\derivtwo$ from $\tm$. 
  Moreover, $\sizebshuf{\deriv} = \sizee{\derivtwo}$.
\end{corollaryAppendix}

\begin{proof}\hfill
  \begin{description}
       \item[$\Rightarrow$:] Let $\deriv \colon \tm \toshuf^* \tmtwo$ be a $\shuf$-normalizing derivation and $\derivtwo \colon \tm \tovsub^* \tmthree$ be its projection in $\vsubcalc$ with $\tmthree$ $\vsub$-normal, according to \refthm{sim-f-into-vsubeq}. 
       Then $\derivtwo$ is a $\vsub$-normalizing derivation from $\tm$.
       
       \item[$\Leftarrow$:] By contradiction, suppose that there is a diverging $\shuf$-derivation $\deriv$ from $\tm$ in $\shufcalc$.
       Since $\tosigm$ is strongly normalizing (\refpropp{basic-shuffling}{terminates}), necessarily in $\deriv$ there are infinitely many $\betashuf$-steps.
       By \refthm{sim-f-into-vsubeq}, $\deriv$ projects to a $\vsub$-derivation in $\vsubcalc$ that has as many $\expo$-steps as the $\betashuf$-steps in $\shufcalc$, absurd.
%       Then, there is a infinite reduction sequence $\deriv \defeq (\tm_i)_{i \in \nat}$ such that $\tm_0 = \tm$ and $\tm_i \tof \tm_{i+1}$ for any $i \in \nat$.
%       For every $n \in \nat$, let $\deriv_n$ be the sub-reduction sequence of $\deriv$ starting from $\tm_0 = \tm$ and such that $\sizef{\deriv_n} = n$:
%       since $\tof$ can be simulated by $\tovsub$ in such a way that the number of $\tof$-steps are equal to the number of $\tom$-steps (\refpropsps{sim-f-into-vsubeq}{qual}{quant-mult}) and $\unf{\tm} = \tm$ (as $\tm \in \Lambda$), there are $\tmtwo_n \in \vsubterms$ and a reduction sequence $\derivtwo_n \colon \tm \tovsub^* \tmtwo_n$ such that $\sizem{\derivtwo_n} = n$ and hence $\sizevsub{\derivtwo_n} \geq n$.
%       Therefore, $\tm$ is not $\vsub$-normalizable, absurd.
  \end{description}

  About the length, we have $\sizebshuf{\deriv} = \sizee{\derivtwo}$ by \refthmp{sim-f-into-vsubeq}{quant}.
  \qedhere
\end{proof}

\setcounter{corollaryAppendix}{\value{coro:shuffling-number-steps}}
\begin{corollaryAppendix}[Number of $\betavm$-Steps is Invariant]
\label{coroappendix:shuffling-number-steps}
  All%
\NoteState{coro:shuffling-number-steps}
  $\shuf$-nor\-ma\-lizing derivations from $\tm \in \Lambda$ (if any) have the same number of $\betashuf$-steps.
\end{corollaryAppendix}

\begin{proof}
  Let $\deriv \colon \tm \toshuf^* \tmtwo$ and $\deriv' \colon \tm \toshuf^* \tmtwop$ be $\shuf$-normalizing. 
  By confluence of $\toshuf$ (\refpropp{basic-shuffling}{confluence}), $\tmtwo = \tmtwop$.
  According to \refthm{sim-shuf-into-vsubeq}, $\deriv$ and $\deriv'$ project, respectively, to two $\vsub$-normalizing derivations $\derivtwo \colon \tm \tovsub^* \tmthree \in \Lambda_\vsub$ and $\derivtwo' \colon \tm \tovsub^* \tmthreep \in \Lambda_\vsub$ such that $\tmthree \eqstruct \mnf{\tmtwo} \eqstruct \tmthreep$, $\sizee{\derivtwo} = \sizebshuf{\deriv}$ and $\sizee{\derivtwo'} = \sizebshuf{\deriv'}$.
  By \refpropp{basic-value-substitution}{strong-confluence}, $\sizee{\derivtwo} = \sizee{\derivtwo'}$ and hence $\sizebshuf{\deriv} = \sizebshuf{\deriv'}$.
\end{proof}

% \subsection{Proofs of Section~\ref{sect:eq-theories} (Equational Theories)}

%\input{B3_-_Proofs_of_Equational_Theories}

% \subsection{Proofs of Section~\ref{sect:cost-models} (How to Stop Worrying and Love the Bomb)}
% 
% \input{B4_-_How_to_Stop_Worrying_and_Love_the_Bomb}
% 
% 
% \subsection{Proofs of Section~\ref{sect:eglamour} (\eglamour)}
% 
% \input{B5_-_Proofs_of_Easy_GLAMOUr}
% 
% 
% \subsection{Proofs of Section~\ref{sect:minimality} (On the Minimality of the Cost Model)}
% 
% \input{B6_-_Proofs_of_On_the_Minimality_of_the_Cost_Model}

\subsection{Proofs of Section~\ref{sect:kernel} (\texorpdfstring{Quantitative Equivalence of $\vsubcalc$ and $\vseqcalc$, via $\vsubkcalc$}{Quantitative Equivalence of lambda-vsub and lambda-vseq, via lambda-vsubk})
}

\label{s:kernel-proofs}
% !TEX root = mainLong.tex
% \setcounter{lemmaAppendix}{\value{l:vsubtoker-sub}}
% \begin{lemmaAppendix}[Substitution]
%  \label{lappendix:vsubtoker-sub}
%   For 
% \NoteState{l:vsubtoker-sub}
%   any $\vsub$-term $\tm$ and any $\vsub$-value $\val$, one has
%  $\vsubtoker{\tm\isub\var\val} = \vsubtoker\tm\isub\var{\vsubtoker\val}$.
% \end{lemmaAppendix}
\subsubsection{Proofs of Subsection \ref{subsect:vsubk-vsub} (\texorpdfstring{Equivalence of $\vsubcalc$ and $\vsubkcalc$}{Equivalence of Value Substitution Calculus and its kernel})}

We first give some more details on $\vsubkcalc$.

The kernel $\vsubkcalc$ of $\vsubcalc$ is the sublanguage of $\vsubcalc$ defined by the following grammar of terms and values (by mutual induction), and evaluation contexts:
    \begin{align*}
	    \vsubk\textup{-Values} && \val & \grameq  \var \mid \la\var\tm  \\
	    \vsubk\textup{-Terms} && \tm, \tmtwo, \tmthree & \grameq  \val \mid \tm\val \mid \tm\esub\var\tmtwo \\
	    \vsubk\textup{-Evaluation Contexts} && \evctx & \grameq  \ctxhole\mid \evctx\val \mid \evctx\esub\var\tmtwo \mid \tm\esub\var\evctx 
    \end{align*}
 
The top-level rewriting rules $\rtom$ and $\rtoe$ are the same as in $\vsubcalc$. 
Note that evaluation contexts of $\vsubkcalc$ no longer include the case $\tm \evctx$, because in $\vsubkcalc$ such contexts cannot surround redexes, as $\evctx$ necessary is the empty context, that can only be filled in with a value, $\val$, and values are normal forms.
The set of terms of $\vsubkcalc$ is denoted by $\vsubkterms$. The restriction of $\tovsub$ to $\vsubkterms$ (\ie the closure of $\rtom \cup \rtoe$ under $\vsubk$-evaluation contexts) is denoted by $\tovsubk$.
Note that $\vsubk$-terms are closed under substitution of $\vsubk$-values (\ie if $\tm$ is a $\vsubk$-term and $\val$ is a $\vsubk$-value, then $\tm\isub\var\val$ is a $\vsubk$-term), therefore for every $\vsubk$-term $\tm$, if $\tm \tovsub \tmtwo$ then $\tmtwo$ is a $\vsubk$-term: so, $\tovsubk$ is a binary relation on $\vsubkterms$.

\begin{lemma}[Substitution]
 \label{l:vsubtoker-sub}
  For 
  any $\vsub$-term $\tm$ and any $\vsub$-value $\val$, one has
 $\vsubtoker{\tm\isub\var\val} = \vsubtoker\tm\isub\var{\vsubtoker\val}$.
\end{lemma}

\begin{proof}
 By induction on $\tm$. Cases:
 \begin{itemize}
   \item $\tm = \var$: then, $\tm\isub\var\val = \val$ and $\vsubtoker{\tm} = \var$, thus $\vsubtoker{\tm\isub\var\val} = \vsubtoker{\val} = \vsubtoker\tm\isub\var{\vsubtoker\val}$.
   
   \item $\tm = \vartwo \neq \var$: then, $\tm\isub\var{\val} = \vartwo$ and $\vsubtoker{\tm} = \vartwo$, hence $\vsubtoker{\tm\isub\var\val} = \vartwo = \vsubtoker\tm\isub\var{\vsubtoker\val}$.
   
   \item $\tm = \la\vartwo\tmthree$: we can suppose without loss of generality that $\vartwo \notin \fv{\val} \cup \{\var\}$, whence $\vartwo \notin \fv{\vsubtoker{\val}}$. 
   By \ih, $\vsubtoker{\tmthree\isub\var\val} = \vsubtoker{\tmthree}\isub\var{\vsubtoker{\val}}$.
   So, $\vsubtoker{\tm\isub\var\val} = \la\vartwo\vsubtoker{\tmthree\isub\var\val} = \la\vartwo \vsubtoker{\tmthree}\isub\var{\vsubtoker{\val}} = \vsubtoker\tm\isub\var{\vsubtoker\val}$.
   
   \item $\tm = \tmthree\tmfour$: then, $\vsubtoker{\tm} = (\vsubtoker{\tmthree}\vartwo)\esub\vartwo{\vsubtoker{\tmthree}}$ with $\vartwo \notin \fv{\tmthree}$, and we can suppose without loss of generality that $\vartwo \notin \fv{\val} \cup \{\var\}$.
   By \ih, $\vsubtoker{\tmthree\isub\var\val} = \vsubtoker{\tmthree}\isub\var{\vsubtoker{\val}}$ and $\vsubtoker{\tmfour\isub\var\val} = \vsubtoker{\tmfour}\isub\var{\vsubtoker{\val}}$, hence $\vsubtoker{\tm\isub\var\val} = (\vsubtoker{\tmthree\isub\var\val}\vartwo)\esub\vartwo{\vsubtoker{\tmfour\isub\var\val}} =  (\vsubtoker{\tmthree}\isub\var{\vsubtoker{\val}}\vartwo)\esub\vartwo{\vsubtoker{\tmfour}\isub\var{\vsubtoker{\val}}} = \vsubtoker{\tm}\isub\var{\vsubtoker{\val}}$, since $\vartwo \neq \var$.
   
   \item $\tm = \tmthree\esub{\vartwo}{\tmfour}$: we can suppose without loss of generality that $\vartwo \notin \fv{\val} \cup \{\var\}$.
   By \ih, $\vsubtoker{\tmthree\isub\var\val} = \vsubtoker{\tmthree}\isub\var{\vsubtoker{\val}}$ and $\vsubtoker{\tmfour\isub\var\val} = \vsubtoker{\tmfour}\isub\var{\vsubtoker{\val}}$, so $\vsubtoker{\tm\isub\var\val} = \vsubtoker{\tmthree\isub\var\val}\esub\vartwo{\vsubtoker{\tmfour\isub\var\val}} =  \vsubtoker{\tmthree}\isub\var{\vsubtoker{\val}}\esub\vartwo{\vsubtoker{\tmfour}\isub\var{\vsubtoker{\val}}} = \vsubtoker{\tm}\isub\var{\vsubtoker{\val}}$.
   \qedhere
 \end{itemize}

\end{proof}

\setcounter{lemmaAppendix}{\value{l:vsub-to-ker-sim}}
\begin{lemmaAppendix}[Simulation] 
\label{lappendix:vsub-to-ker-sim} % \reflemmap{vsub-to-ker-sim}{mult}
  Let 
\NoteState{l:vsub-to-ker-sim}
  $\tm, \tmtwo \in \vsubterms$.
\begin{enumerate}
	\item \emph{Multiplicative}: \label{pappendix:vsub-to-ker-sim-mult}
	if $\tm \tom \tmtwo$ then $\vsubtoker\tm \tom \vsubtoker\tmtwo$ or $\vsubtoker\tm \tom\toevar\eqstruct\vsubtoker\tmtwo$;
	\item \emph{Exponential Abstractions \& Variables}: \label{pappendix:vsub-to-ker-sim-exp}
	if $\tm \toeabs \tmtwo$ then $\vsubtoker\tm \toeabs \vsubtoker\tmtwo$, and if $\tm \toevar\tmtwo$ then $\vsubtoker\tm \toevar \vsubtoker\tmtwo$.
	\item \emph{Structural Equivalence}: $\tm \eqstruct\tmtwo$ implies $\vsubtoker\tm \eqstruct \vsubtoker\tmtwo$.
\end{enumerate}
\end{lemmaAppendix}

\begin{proof}
Define $\vsubtoker\sctx$ by $\vsubtoker{\ctxhole} \defeq \ctxhole$ and $\vsubtoker{\sctx\esub\var\tmfour} \defeq \vsubtoker\sctx\esub\var{\vsubtoker{\tmfour}}$. Then note that if $\tm = \sctxp\tmthree$ we have $\vsubtoker{\tm} = \vsubtoker\sctx\ctxholep{\vsubtoker{\tmthree}}$. By induction on the evaluation context in which the step takes place. The base cases:
\begin{enumerate}
	\item \emph{Multiplicative}, \ie $\tm = \sctxp{\la\var\tmthree} \tmfour \tom \sctxp{\tmthree\esub\var\tmfour} = \tmtwo$. Then
	\[ \begin{array}{rclclc}
      \vsubtoker{\tm} & = & \vsubtoker{(\sctxp{\la\var\tmthree} \tmfour)}
      \\
      & = &
      (\vsubtoker\sctx\ctxholep{\la\var\vsubtoker{\tmthree}} \vartwo) \esub\vartwo{\vsubtoker{\tmfour}}
      \\
      & \tom &
      \vsubtoker\sctx\ctxholep{\vsubtoker{\tmthree}\esub\var\vartwo} \esub\vartwo{\vsubtoker{\tmfour}}
      \\
      & \toevar &
      \vsubtoker\sctx\ctxholep{\vsubtoker{\tmthree}\isub\var\vartwo} \esub\vartwo{\vsubtoker{\tmfour}}
      \\
      & \eqstruct &
      \vsubtoker\sctx\ctxholep{\vsubtoker{\tmthree}\isub\var\vartwo\esub\vartwo{\vsubtoker{\tmfour}}}
      \\
      & \alphaeq &
      \vsubtoker\sctx\ctxholep{\vsubtoker{\tmthree} \esub\var{\vsubtoker{\tmfour}}}
      \\
      & = & 
      \vsubtoker{\sctxp{\tmthree\esub\var\tmfour}}
      & = & 
      \vsubtoker{\tmtwo}
 \end{array}\]
 
	\item \emph{Exponential Abstractions \& Variables}, \ie $\tm = \tmthree\esub\var{\sctxp\val} \toe \sctxp{\tmthree \isub\var\val} = \tmtwo$.
	
	\[ \begin{array}{rclclc}
      \vsubtoker{\tm} & = & \vsubtoker{\tmthree\esub\var{\sctxp\val}}
      \\
      & = &
      \vsubtoker\tmthree\esub\var{\vsubtoker\sctx\ctxholep{\vsubtoker\val}}
      \\
      & \toe &
      \vsubtoker\sctx\ctxholep{\vsubtoker\tmthree\isub\var{\vsubtoker\val}}
      \\
      & =_\reflemmaeq{vsubtoker-sub} &
      \vsubtoker\sctx\ctxholep{\vsubtoker{\tmthree \isub\var\val}}
      \\
      & = &
      \vsubtoker{\sctxp{\tmthree \isub\var\val}}      
      & = & 
      \vsubtoker{\tmtwo}
 \end{array}\]
	Note that the translation $\vsubtoker{\cdot}$ maps $\toeabs$ steps on $\toeabs$ steps and $\toevar$ steps to $\toevar$ steps because it maps variables to variables and abstractions to abstractions.
	
	\item \emph{Structural Equivalence}:  it is enough to prove that the statement holds for the axioms of $\eqstruct$, \ie if $\tm \eqstruct_\mathsf{r} \tmtwo$ for some $\mathsf{r} \in \{\textup{com}, @_\textup{l}, @_\textup{r}, [\cdot]\}$, then $\vsubtoker{\tm} \eqstruct \vsubtoker{\tmtwo}$. 
       Cases (recall that $\vsubk$ is a sublanguage of $\vsub$):
       \begin{itemize}
         \item $\tm\esub{\vartwo}{\tmthree}\esub{\var}{\tmtwo} \tostructcom \tm\esub{\var}{\tmtwo}\esub{\vartwo}{\tmthree}$ with $\vartwo\notin\fv{\tmtwo}$ and $\var\notin\fv{\tmthree}$: then, $\vsubtoker{\tm\esub{\vartwo}{\tmthree}\esub{\var}{\tmtwo}} = \vsubtoker{\tm}\esub{\vartwo}{\vsubtoker{\tmthree}}\esub{\var}{\vsubtoker{\tmtwo}} \tostructcom \vsubtoker{\tm}\esub{\var}{\vsubtoker{\tmtwo}}\esub{\vartwo}{\vsubtoker{\tmthree}} = \vsubtoker{\tm\esub{\var}{\tmtwo}\esub{\vartwo}{\tmthree}}$\!, since $\vartwo\notin\fv{\vsubtoker{\tmtwo}}$ and $\var\notin\fv{\vsubtoker{\tmthree}}$.
         
         \item $\tm\,\tmthree\esub\var\tmtwo \tostructapr  (\tm\tmthree)\esub\var\tmtwo$ with $\var\not\in\fv{\tm}$: 
         then, $\var\not\in\fv{\vsubtoker{\tm}}$ and $\vsubtoker{(\tm\tmthree)} = (\vsubtoker{\tm}\vartwo)\esub\vartwo{\vsubtoker{\tmthree}}$ with $\vartwo \notin \fv{\tm} = \fv{\vsubtoker{\tm}}$, and we can suppose without loss of generality that %$\vartwo \notin \fv{\tmtwo} \cup \{\var\}$
         $\vartwo \neq \var$.
         So, $\vsubtoker{(\tm\,\tmthree\esub\var\tmtwo)} = (\vsubtoker{\tm}\vartwo)\esub\vartwo{\vsubtoker{\tmthree}\esub\var{\vsubtoker{\tmtwo}}} \tostructes (\vsubtoker{\tm}\vartwo)\esub\vartwo{\vsubtoker{\tmthree}}\esub\var{\vsubtoker{\tmtwo}} = \vsubtoker{(\tm\tmthree)\esub\var\tmtwo}$.
	
	\item $\tm\esub\var\tmtwo\tmthree \tostructapl (\tm\tmthree)\esub\var\tmtwo$ with $\var\not\in\fv\tmthree$:
	then, $\var\not\in\fv{\vsubtoker{\tmthree}}$ and $\vsubtoker{(\tm\tmthree)} = (\vsubtoker{\tm}\vartwo)\esub\vartwo{\vsubtoker{\tmthree}}$ with $\vartwo \notin \fv{\tm} = \fv{\vsubtoker{\tm}}$.
	We can suppose without loss of generality that $\vartwo \notin \fv{\tmtwo} \cup \{\var\}$, hence $\vsubtoker{\tm}\esub\var{\vsubtoker{\tmtwo}}\vartwo \tostructapl (\vsubtoker{\tm}\vartwo)\esub\var{\vsubtoker{\tmtwo}}$.
	Therefore, $\vsubtoker{(\tm\esub\var\tmtwo\tmthree)} = (\vsubtoker{\tm}\esub\var{\vsubtoker{\tmtwo}}\vartwo)\esub\vartwo{\vsubtoker{\tmthree}} \tostruct (\vsubtoker{\tm}\vartwo)\esub\var{\vsubtoker{\tmtwo}}\esub\vartwo{\vsubtoker{\tmthree}} \tostructes (\vsubtoker{\tm}\vartwo)\esub\vartwo{\vsubtoker{\tmthree}}\esub\var{\vsubtoker{\tmtwo}} = \vsubtoker{(\tm\tmthree)\esub\var\tmtwo}$.
	
         \item $\tm\esub{\var}{\tmtwo\esub{\vartwo}{\tmthree}} \tostructes \tm\esub{\var}{\tmtwo}\esub{\vartwo}{\tmthree} $ with $\vartwo\not\in\fv{\tm}$: then, $\vsubtoker{\tm\esub{\var}{\tmtwo\esub{\vartwo}{\tmthree}}} = \vsubtoker{\tm}\esub{\var}{\vsubtoker{\tmtwo}\esub{\vartwo}{\vsubtoker{\tmthree}}} \tostructes \vsubtoker{\tm}\esub{\var}{\vsubtoker{\tmtwo}}\esub{\vartwo}{\vsubtoker{\tmthree}} = \vsubtoker{\tm\esub{\var}{\tmtwo}\esub{\vartwo}{\tmthree}}$, since $\vartwo\not\in\fv{\vsubtoker{\tm}}$.
       \end{itemize}

\end{enumerate}
The inductive cases simply follow from the \ih, note indeed that evaluation contexts are translated to evaluation contexts.
\end{proof}

\begin{definition}
%   Let $\tm$ be a $\vsub$-term: $\tm$ is \emph{\harmless} if for every subterm of the form $\tmtwo\esub{\var}{\sctxp{\val}}$ occurring in $\tm$ one has $\tmtwo = \tmthree\var$ with $\var \notin \fv{\tmthree}$.
% FALSE! the fact that $\tm \toe \tmtwo$ and $\tm$ is harmless does not imply that $\tmtwo$ is harmless, according this definition of harmless. consider $\tm = (\tmthree \var) \esub\var{\val\esub\vartwo\valtwo}$
  Let $\tm$ be a $\vsub$-term: $\tm$ is \emph{\harmless} if for every subterm of the form $\tmtwo\esub{\var}{\sctxp{\val}}$ occurring in $\tm$ one has $\tmtwo = \tmthree\valtwo$ with $\var \notin \fv{\tmthree}$ (where $\val$ and $\valtwo$ are $\vsub$-values).
\end{definition}

\begin{lemma}[Harmless preservations]
\label{l:harmless-preserve}
  Let $\tm \in \vsubterms$. 
  \begin{varenumerate}
    \item\label{p:harmless-preserve-mult} Suppose $\tm$ is \harmless and $\tm \toe \tmtwo$. Then, $\tmtwo$ is \harmless. 
    If, moreover, $\tm$ is $\mult$-normal, then $\tmtwo$ is $\mult$-normal.
    \item\label{p:harmless-preserve-transl} If $\tm$ is $\vsub$-normal, then $\vsubtoker\tm$ is \harmless and $\mult$-normal.
  \end{varenumerate}
\end{lemma}

\begin{proof}\hfill
  \begin{enumerate}
    \item 
%     If $\tmtwo$ is not \harmless, then there exists a subterm $\tmtwop\esub\var{\sctxp{\val}}$ of $\tmtwo$ such that either $\tmtwop \neq \tmthree\var$ or $\tmtwop = \tmthree\var$ but $\var \in \fv{\tmthree}$. Since a $\expo$-step substitutes values for variables, necessarily $\tm$ is not \harmless.
    From $\tm \toe \tmtwo$ it follows that $\tmtwo$ is obtained from $\tm$ by replacing a subterm of the form $\tmp\esub\var{\sctxp{\val}}$ with $\sctxp{\tmp\isub\var\val}$, where all the variables binded by the substitution context $\sctx$ are not free in $\tmp$.
    Since $\tm$ is harmless, then $\tmp = \tmthree\valtwo$ with $\var \notin \fv{\tmthree}$ and hence $\sctxp{\tmp\isub\var\val} = \sctxp{\tmthree(\valtwo\isub\var\val)}$ where $\valtwo\isub\var\val$ is a $\vsub$-value, since $\vsub$-values are closed by substitution.
    Moreover, all the variables binded by the substitution context $\sctx$ are not free in $\tmthree$.
    So, $\sctxp{\tmp\isub\var\val}$ is not an objection to the harmless property.
    Therefore, $\tmtwo$ is harmless.
    
    If $\tmtwo$ is not $\mult$-normal then it has a subterm of the form $\sctxp{\la\var\tmthree}\tmfour$, but $\tm$ is \harmless and hence the step $\tm \toe \tmtwo$ can only substitute a value for a variable in argument position of an application: this means that $\tm$ has a subterm of the form $\sctxp{\la\var\tmthree}\tmfourp$ and hence $\tm$ is not $\mult$-normal.

    \item The translation $\vsubtoker{(\cdot)}$ 
%     does not essentially modify the left side of an application (in particular, does not create new abstractions there), 
    essentially creates new ES from applications, but does not create new abstractions or new applications, 
    hence if $\vsubtoker{\tm}$ contained a subterm of the form $\sctxp{\la\var\tmtwo}\tmthree$ then $\tm$ would contain a subterm of the form $\sctxtwop{\la\var\tmtwop}\tmthreep$: but $\tm$ is $\mult$-normal and therefore $\vsubtoker{\tm}$ is $\mult$-normal.
    Moreover, since $\tm$ is $\expo$-normal, the only subterms of $\vsubtoker{\tm}$ of the form $\tmtwo\esub{\var}{\sctxp{\val}}$ are created by the translation $\vsubtoker{(\cdot)}$ (\ie they are not in $\tm$), and the ES introduced by the translation $\vsubtoker{(\cdot)}$ fulfill the \harmless condition.
    Therefore $\vsubtoker{\tm}$ is \harmless.
    \qedhere
%   First let us characterize the $\vsub$-normal forms, up to $\eqstruct$.
%   \begin{align*}
%     a &\defeq \var w \mid aw & A &\defeq a \mid A\esub\var{a} \\
%     w &\defeq v \mid a & W &\defeq w \mid A \mid W\esub\var{a}
%   \end{align*}
  \end{enumerate}
\end{proof}

\setcounter{theoremAppendix}{\value{thm:sim-vsub-into-vsubk}}
\begin{theoremAppendix}[Quantitative Simulation of $\vsubcalc$ in $\vsubkcalc$]
\label{thmappendix:sim-vsub-into-vsubk}
  Let%
\NoteState{thm:sim-vsub-into-vsubk}
  $\tm, \tmtwo \in \vsubterms$.
  If $\deriv \colon \tm \tovsub^* \tmtwo$ then there are $\tmthree \in \vsubkterms$ and $\derivtwo \colon \vsubtoker\tm \tovsubk^* \tmthree$ such that
  
  \begin{varenumerate}
    \item\label{pappendix:sim-vsub-into-vsubk-qual} \emph{Qualitative Relationship}: $\tmthree \eqstruct \vsubtoker\tmtwo$;

    \item\label{pappendix:sim-vsub-into-vsubk-quant} \emph{Quantitative Relationship}: 
    \begin{enumerate}
    	\item \emph{Multiplicative Steps}: $\sizem{\derivtwo} = \sizem{\deriv}$;
	\item \emph{Exponential Steps}: $\sizeeabs\derivtwo = \sizeeabs\deriv$ and $\sizeevar\derivtwo = \sizeevar\deriv + \sizem\deriv$;
    \end{enumerate}
    \item\label{pappendix:sim-vsub-into-vsubk-normal} \emph{Normal Form}: if $\tmtwo$ is $\vsub$-normal then $\tmthree$ is $\mult$-normal and $\enf{\!\tmthree\!}$ is $\vsubk$\!-normal.
  \end{varenumerate}
\end{theoremAppendix}

\begin{proof}\hfill
\begin{enumerate}
\item[1-2.] By induction on $\size\deriv$ using \reflemmap{vsub-to-ker-sim}{mult} and \reflemmap{vsub-to-ker-sim}{exp} plus the postponement of $\eqstruct$ (\reflemmap{eqstruct-post-and-term}{locpost}).

\item[3.]
  By \reflemmap{harmless-preserve}{transl}, $\vsubtoker{\tmtwo}$ is $\mult$-normal and \harmless, thus also $\enf{\vsubtoker{\tmtwo}}$ is $\mult$-normal according to \reflemmap{harmless-preserve}{mult}. Therefore, $\enf{\vsubtoker{\tmtwo}}$ is $\vsub$-normal. 
  By the properties of strong bisimulation (\reflemma{eqstruct-post-and-term}), $\tmthree \eqstruct \vsubtoker\tmtwo$ implies that $\enf{\tmthree} \eqstruct \enf{\vsubtoker\tmtwo}$ and hence $\enf{\tmthree}$ is $\vsub$-normal.
\qedhere
\end{enumerate}
\end{proof}

\setcounter{corollaryAppendix}{\value{coro:equivalence-vsub-vsubk-termination}}
\begin{corollaryAppendix}[Linear Termination Equivalence of $\vsubcalc$ and $\vsubkcalc$\!]
\label{coroappendix:equivalence-vsub-vsubk-termination}
  Let
\NoteState{coro:equivalence-vsub-vsubk-termination}
  $\tm \in \vsubterms$. 
  There exists a $\vsub$-normalizing derivation $\deriv$ from $\tm$ iff there exists a $\vsubk$-normalizing derivation $\derivthree$ from $\vsubtoker\tm$. 
    Moreover, $\sizem{\deriv} = \sizem{\derivthree}$.
\end{corollaryAppendix}

\begin{proof}\hfill
  \begin{description}
       \item[$\Rightarrow$:] Let $\deriv \colon \tm \tovsub^* \tmtwo$ be a $\vsub$-normalizing derivation and $\derivtwo \colon \vsubtoker\tm \tovsubk^* \tmthree$ be its projection in $\vsubkcalc$, according to \refthm{sim-vsub-into-vsubk}. 
       By \refthmp{sim-vsub-into-vsubk}{normal}, the derivation $\derivthree$ obtained by extending $\derivtwo$ with a normalization with respect to $\toe$ (that always terminate) is a $\vsubk$-normalizing derivation from $\vsubtoker\tm$.
       
       \item[$\Leftarrow$:] By contradiction, suppose that there is a diverging $\vsub$-derivation from $\tm$ in $\vsubcalc$. By \refthm{sim-vsub-into-vsubk} it projects to a $\vsubk$-derivation from $\vsubtoker\tm$ in $\vsubkcalc$ that is at least as long as the one in $\vsubcalc$, absurd.
  \end{description}
  
  About lengths, consider the normalizing derivations $\deriv$, $\derivtwo$, and $\derivthree$ discussed in the proof of the $\Rightarrow$ direction. We have to show that $\sizem\deriv = \sizem\derivthree$. By \refthmp{sim-vsub-into-vsubk}{normal}, $\sizem\deriv = \sizem\derivtwo$. Since $\derivthree$ extends $\derivtwo$ only with exponential steps, we obtain $\sizem\derivthree = \sizem\derivtwo$.
\end{proof}

%\subsection{Proofs of Section~\ref{sect:lambdamu} (\texorpdfstring{The Intuitionistic $\lambdamucalc$-calculus}{The lambdabar-mutilde-calculus})}

% \paragraph{Proof of strong confluence of $\tovseq$}
% 
% \input{lbarmu-strong-confluence}

\subsubsection{Proof of Subsection \ref{subsect:vsubk-vseq} (\texorpdfstring{Equivalence of $\vsubkcalc$ and $\vseqcalc$}{Equivalence of Kernel Value Substitution and Value Sequent Substitution})}

% \setcounter{lemmaAppendix}{\value{l:tolbarmu-subs}}
% \begin{lemmaAppendix}[Translation and Substitution Commute]
% \label{lappendix:tolbarmu-subs}
% Let 
% \NoteState{l:tolbarmu-subs}
% $\val$, $\valtwo$ be values and $\tm$ be a term of $\vsubk$.
% \begin{enumerate}
% \item \emph{Values}: $\tolbarmuv{ (\valtwo \isub\var\val) } = \tolbarmuv\valtwo \isub\var{\tolbarmuv\val}$;
% \item \emph{Terms}: $\tolbarmu{ \tm \isub\var\val } = \tolbarmu\tm \isub\var{\tolbarmuv\val}$.
% \end{enumerate}
% \end{lemmaAppendix}
\begin{lemma}[Translation and Substitution Commute]
\label{l:tolbarmu-subs}
Let 
$\val$, $\valtwo$ be values and $\tm$ be a term of $\vsubk$.
\begin{enumerate}
\item \emph{Values}: $\tolbarmuv{ (\valtwo \isub\var\val) } = \tolbarmuv\valtwo \isub\var{\tolbarmuv\val}$;
\item \emph{Terms}: $\tolbarmu{ \tm \isub\var\val } = \tolbarmu\tm \isub\var{\tolbarmuv\val}$.
\end{enumerate}
\end{lemma}

\begin{proof}\hfill
\begin{enumerate}
\item Cases:
\begin{itemize}
\item \emph{Variable}, \ie $\tm = \var$. Then $\tolbarmuv{ (\var \isub\var\val) } = \tolbarmuv{ \val } = \var \isub\var{\tolbarmuv\val} = \tolbarmuv\var \isub\var{\tolbarmuv\val}$.

	  \item \emph{Abstraction}, \ie $\tm = \la\vartwo\tmtwo$. 
  \[\begin{array}{rclcc}
     \tolbarmuv{ ((\la\vartwo\tmtwo) \isub\var\val)}
     & = & 
     \tolbarmuv{ (\la\vartwo\tmtwo \isub\var\val)}
     \\
     & = & 
	\la\vartwo\tolbarmu{\tmtwo\isub\var\val} 
     \\
     & =_{P.2} & 
	\la\vartwo\tolbarmu\tmtwo\isub\var{\tolbarmuv\val}
     \\
     & = & 
	(\la\vartwo\tolbarmu\tmtwo)  \isub\var{\tolbarmuv\val}
     & = & 
	\tolbarmuv{(\la\vartwo\tmtwo)}  \isub\var{\tolbarmuv\val}
    \end{array}\]

\end{itemize}

\item By induction on $\tm$. Cases:
\begin{itemize}
	\item \emph{Value}, \ie $\tm = \valtwo$. Note that $\valtwo \isub\var\val$ is a value. Then $\tolbarmu{ \valtwo \isub\var\val } = \comm{ \tolbarmuv{ (\valtwo \isub\var\val) } }\stempty =_{P.1} \comm{ \tolbarmuv\valtwo \isub\var{\tolbarmuv\val}} \stempty = \comm{ \tolbarmuv\valtwo } \stempty \isub\var{\tolbarmuv\val} = \tolbarmu\valtwo \isub\var{\tolbarmuv\val}$.
	
	\item \emph{Application}, \ie $\tm = \tmtwo\valtwo$.
  \[\begin{array}{rclcc}
     \tolbarmu{\tmtwo\valtwo \isub\var\val}
     & = & 
     \tolbarmu{\tmtwo\isub\var\val \valtwo\isub\var\val}
     \\
     & = & 
     \append{ \tolbarmu{ \tmtwo\isub\var\val } }{ (\stacker{\tolbarmuv{\valtwo\isub\var\val}}\stempty) }
     \\
     & =_\ih & 
     \append{ \tolbarmu\tmtwo \isub\var{ \tolbarmuv\val } }{ (\stacker{\tolbarmuv{\valtwo\isub\var\val}}\stempty) }
     \\
     & =_{P.1} & 
     \append{ \tolbarmu\tmtwo \isub\var{ \tolbarmuv\val } }{ (\stacker{\tolbarmuv\valtwo\isub\var{ \tolbarmuv\val }}\stempty) }
     \\
     & = & 
     \append{ \tolbarmu\tmtwo }{ (\stacker{ \tolbarmuv\valtwo }\stempty) } \isub\var{ \tolbarmuv\val }
     & = & 
     \tolbarmu{ \tmtwo\valtwo } \isub\var{ \tolbarmuv\val }
    \end{array}\]
    
  \item \emph{Substitution}, \ie $\tm = \tmtwo\esub\vartwo\tmthree$. 
   \[\begin{array}{rclcc}
     \tolbarmu{\tmtwo\esub\vartwo\tmthree \isub\var\val}
     & = & 
     \tolbarmu{\tmtwo\isub\var\val \esub\vartwo{ \tmthree\isub\var\val }}
     \\
     & = & 
     \append{ \tolbarmu{ \tmthree\isub\var\val } }{ \mutilde\vartwo\tolbarmu{\tmtwo\isub\var\val} } 
     \\
     & =_\ih & 
     \append{ \tolbarmu\tmthree \isub\var{ \tolbarmuv\val } }{ \mutilde\vartwo\tolbarmu\tmtwo \isub\var{ \tolbarmuv\val } } 
     \\
     & = & 
          (\append{ \tolbarmu\tmthree }{ \mutilde\vartwo\tolbarmu\tmtwo }) 
 \isub\var{ \tolbarmuv\val }
     & = & 
     \tolbarmu{ \tmtwo\esub\vartwo\tmthree } \isub\var{ \tolbarmuv\val }
    \end{array}\]
  \qedhere
\end{itemize}

\end{enumerate}
\end{proof}

% \setcounter{lemmaAppendix}{\value{l:cotctx}}
% \begin{lemmaAppendix}
%   \label{lappendix:cotctx}
%  Let 
%  \NoteState{l:cotctx}
%  $\tm$ be a $\vsubk$-term. Then there exist a command evaluation context  $\cmctx$ and a co-term context $\cotctx$ such that $\tolbarmu\tm = \cmctxp{\cotctxp\stempty}$.
% \end{lemmaAppendix}
\begin{lemma}
  \label{l:cotctx}
 Let 
 $\tm$ be a $\vsubk$-term. Then there exist a command evaluation context  $\cmctx$ and a environment evaluation context $\cotctx$ such that $\tolbarmu\tm = \cmctxp{\cotctxp\stempty}$.
\end{lemma}

\begin{proof}
 By induction on $\tm$. Cases:
 \begin{enumerate}
  \item \emph{Variable}, \ie $\tm = \var$. Trivial just take $\cmctx \defeq \ctxhole$ and $\cotctx \defeq \comm\var\ctxhole$.
  \item \emph{Abstraction}, \ie $\tm = \la\var\tmtwo$. Trivial just take $\cmctx \defeq \ctxhole$ and $\cotctx \defeq \comm{\la\var\tolbarmu\tmtwo}\ctxhole$.
  \item \emph{Application}, \ie $\tm = \tmtwo\val$.
  \[\begin{array}{rclcc}
     \tolbarmu{\tmtwo\val}
     & = & 
     \append{ \tolbarmu\tmtwo }{ (\stacker{\tolbarmuv\val}\stempty) }
     \\
     & =_\ih & 

     \append{ \cmctxtwop{\cotctxtwop\stempty} }{ (\stacker{\tolbarmuv\val}\stempty) }
     \\
     & = & 
     \cmctxtwop{\cotctxtwop{\stacker{\tolbarmuv\val}\stempty}}
    \end{array}\]
    
    The statement holds with respect to $\cmctx \defeq \cmctxtwo$ and  
    $\cotctx \defeq \cotctxtwop{\stacker\var\ctxhole}$.
    
  \item \emph{Substitution}, \ie $\tm = \tmtwo\esub\var\tmthree$. 
  \[\begin{array}{rclcc}
     \tolbarmu{\tmtwo\esub\var\tmthree}
     & = & 
     \append{\tolbarmu\tmthree}{\mutilde\var{\tolbarmu\tmtwo}}
     \\
     & =_\ih & 
     \append{ \cmctxtwop{\cotctxtwop\stempty} }{\mutilde\var{\tolbarmu\tmtwo}}
     \\
     & = & 
     \cmctxtwop{\cotctxtwop{ \mutilde\var{\tolbarmu\tmtwo} }}
     \\
     & =_\ih & 
     \cmctxtwop{\cotctxtwop{ \mutilde\var{\cmctxthreep{\cotctxthreep\stempty}} }}     
    \end{array}\]
    
    The statement holds with respect to $\cmctx \defeq \cmctxtwop{\cotctxtwop{\mutilde\var{\cmctxthree}}}$ and  
    $\cotctx \defeq \cotctxthree$. 
    \qedhere
 \end{enumerate}

\end{proof}

% \setcounter{lemmaAppendix}{\value{l:sctx-mutilde}}
% \begin{lemmaAppendix}
%   \label{lappendix:sctx-mutilde}
%  Let 
%  \NoteState{l:sctx-mutilde}
%  $\sctx$ be a substitution context of $\vsub$. There exists a command evalution context $\cmctx$ such that $\tolbarmu{\sctxp\tm} = \cmctxp{\tolbarmu\tm}$ for any $\vsubk$-term $\tm$. Moreover, $\fv\cmctx = \fv\sctx$ and $\cmctx$ and $\sctx$ capture they same variables of $\tm$.
% \end{lemmaAppendix}

\begin{lemma}
  \label{l:sctx-mutilde}
 Let 
 $\sctx$ be a substitution context of $\vsub$. There exists a command evalution context $\cmctx$ such that $\tolbarmu{\sctxp\tm} = \cmctxp{\tolbarmu\tm}$ for any $\vsubk$-term $\tm$. Moreover, $\fv\cmctx = \fv\sctx$ and $\cmctx$ and $\sctx$ capture they same variables of $\tm$.
\end{lemma}

\begin{proof}
 By induction on $\sctx$. Cases:
    \begin{enumerate}
      \item \emph{Empty Context}, \ie $\sctx = \ctxhole$. Just take $\cmctx \defeq \ctxhole$.
      \item \emph{Non-Empty Context}, \ie $\sctx = \sctxtwo\esub\var\tmtwo$. Then
      \[\begin{array}{rclcc}
     \tolbarmu{\sctxtwop\tm \esub\var\tmtwo}
     & = & 
     \append{\tolbarmu{\tmtwo}}{\mutilde{\var}{\tolbarmu{\sctxtwop\tm}}}
     \\
     & =_\ih & 
     \append{\tolbarmu{\tmtwo}}{\mutilde{\var}{\cmctxtwop{\tolbarmu\tm}}}
    % \tolbarmu{\tmtwo}{\mutilde{\var}{\cmctxtwop{\transl\tm\cotm}}}
     \\
     & =_{\reflemmaeq{cotctx}} & 
     \append{\cmctxthreep{\cotctxp{\stempty}}}{ \mutilde{\var}{\cmctxtwop{\tolbarmu\tm}} }
     \\
     & = & 
     \cmctxthreep{\cotctxp{ \mutilde{\var}{\cmctxtwop{\tolbarmu\tm}} }}
    \end{array}\]
    The statement holds with respect to $\cmctx \defeq \cmctxthreep{\cotctxp{\mutilde{\var}{\cmctxtwo}}}$. The \emph{moreover} part follows from the \emph{moreover} part of \reflemma{cotctx} and the \ih
     \qedhere
    \end{enumerate}
\end{proof}

\setcounter{lemmaAppendix}{\value{l:lbarmut-simulates-vsubk}}
\begin{lemmaAppendix}[Simulation of $\tovsubk$ by $\tovseq$]
\label{lappendix:lbarmut-simulates-vsubk}
  Let 
  \NoteState{l:lbarmut-simulates-vsubk}
  $\tm, \tmtwo \in \vsubkterms$.
  \begin{enumerate}
    \item \emph{Multiplicative}: if $\tm \tom \tmtwo$ then 
    $\tolbarmu\tm \tobvmu \tolbarmu\tmtwo$.
    \item \emph{Exponential}: if $\tm \toe \tmtwo$ then 
    $\tolbarmu\tm \tomut \tolbarmu\tmtwo$.
  \end{enumerate}
\end{lemmaAppendix}

\begin{proof}
Both points are proved by induction on the evaluation context $\evctx$ in which the step takes place. Cases:
 \begin{itemize}
  \item \emph{Root case}, \ie $\evctx = \ctxhole$.
  \begin{enumerate}
    \item \emph{Multiplicative Step}: $\tm = \sctxp{\la\var\tmthree}\val \rtom \sctxp{\tmthree\esub\var\val} = \tmtwo$. 
  
    \[\begin{array}{rclcc}
     \tolbarmu{\sctxp{\la\var\tmthree}\val}
     & = & 
     \append{\tolbarmu{\sctxp{\la\var\tmthree}}}{(\stacker{\tolbarmuv{\val}}\stempty)}
     \\
     & =_\reflemmaeq{sctx-mutilde} & 
     \append{\cmctxp{\tolbarmu{\la\var\tmthree}}}{(\stacker{\tolbarmuv{\val}}\stempty)}     
     \\
     & = & 
     \append{\cmctxp{\comm{\la\var\tolbarmu\tmthree}\stempty  } }{(\stacker{\tolbarmuv{\val}}\stempty)}
     \\
     & = & 
     \cmctxp{\append{\comm{\la\var\tolbarmu\tmthree}\stempty}  {(\stacker{\tolbarmuv{\val}}\stempty)}} 
     \\
     & = & 
     \cmctxp{\comm{\la\var\tolbarmu\tmthree}{ \stacker{\tolbarmuv{\val}}\stempty} } 
     \\
     & \tobvmu  & 
     \cmctxp{\comm{ \tolbarmuv\val} { \append{\mutilde\var \tolbarmu\tmthree}\stempty }}       
     \\
     & = & 
     \cmctxp{\comm{ \tolbarmuv\val} { \mutilde\var \tolbarmu\tmthree }}       
     \\
     & = & 
     \cmctxp{\tolbarmu{\tmthree\esub\var\val}}       
     
     & =_\reflemmaeq{sctx-mutilde} & 
     \tolbarmu{\sctxp{\tmthree\esub\var\val}}
    \end{array}\]

    \item \emph{Exponential Step}: $\tm = \tmthree\esub\var{\sctxp\val} \rtoe \sctxp{\tmthree\isub\var\val} = \tmtwo$. 
    
    \[\begin{array}{rclcc}
     \tolbarmu{\tmthree\esub\var{\sctxp\val}}
          & = & 
     \append{\tolbarmu{\sctxp{\val}}}{\mutilde\var\tolbarmu\tmthree}
     \\
     & =_\reflemmaeq{sctx-mutilde} & 
     \append{\cmctxp{\tolbarmu{\val}}}{\mutilde\var\tolbarmu\tmthree}
     \\
     & = & 
     \append{\cmctxp{ \comm{\tolbarmuv\val}\stempty }}{\mutilde\var\tolbarmu\tmthree}
     \\
     & = & 
     \cmctxp{ \comm{\tolbarmuv\val}{ \mutilde\var\tolbarmu\tmthree }}
     \\
     & \tomut & 
     \cmctxp{ \tolbarmu\tmthree \isub\var{\tolbarmuv\val} }
     \\
     & =_{\reflemmaeq{tolbarmu-subs}} & 
     \cmctxp{ \tolbarmu{\tmthree \isub\var\val} }
	  & =_\reflemmaeq{sctx-mutilde} & 
     \tolbarmu{\sctxp{\tmthree\isub\var\val}}
    \end{array}\]

  \end{enumerate}
  \item \emph{Inductive Cases}: for each case the two points differs only in the kind of the rewriting step, so we treat them compactly, by referring to $\tovsubk$ and $\tovseq$
  \begin{itemize}
  	\item \emph{Left Application}, \ie $\tm = \tmthree \val \tovsubk \tmfour \val = \tmtwo$ with $\tmthree \tovsubk \tmfour$. By \ih, $\tolbarmu\tmthree \tovseq \tolbarmu\tmfour$. And by \reflemmap{append-commutes}{cm-steps}, $\tolbarmu{ \tmthree \val } = \append{ \tolbarmu\tmthree}  { (\stacker{ \tolbarmu\val} \stempty) } \tovseq \append{ \tolbarmu\tmfour}  { (\stacker{ \tolbarmu\val}\stempty) } = \tolbarmu{ \tmfour \val }$.
	
  	\item \emph{Left of a Substitution}: \ie $\tm = \tmthree \esub\var\tmfive \tovsubk \tmfour \esub\var\tmfive = \tmtwo$ with $\tmthree \tovsubk \tmfour$. By \ih, $\tolbarmu\tmthree \tovseq \tolbarmu\tmfour$. By \reflemma{cosubstitution}, $\tolbarmu{ \tmthree \esub\var\tmfive } = \append{ \tolbarmu\tmfive}  { \mutilde\var\tolbarmu\tmthree } \tovseq \append{ \tolbarmu\tmfive}  { \mutilde\var\tolbarmu\tmfour } = \tolbarmu{ \tmthree \esub\var\tmfour }$.
	
  	\item \emph{Inside a Substitution}:  \ie $\tm = \tmthree \esub\var\tmfive \tovsubk \tmthree \esub\var\tmfour = \tmtwo$ with $\tmfive \tovsubk \tmfour$. By \ih, $\tolbarmu\tmfive \tovseq \tolbarmu\tmfour$. And by \reflemmap{append-commutes}{cm-steps}, $\tolbarmu{ \tmthree \esub\var\tmfive } = \append{ \tolbarmu\tmfive}  { \mutilde\var\tolbarmu\tmthree } \tovseq \append{ \tolbarmu\tmfour}  { \mutilde\var\tolbarmu\tmthree } = \tolbarmu{ \tmthree \esub\var\tmfour }$.
  	\qedhere
  \end{itemize}
 \end{itemize}

\end{proof}

\setcounter{theoremAppendix}{\value{thm:sim-vsubk-into-vseq}}
\begin{theoremAppendix}[Quantitative Simulation of $\vsubkcalc$ in $\vseqcalc$]
\label{thmappendix:sim-vsubk-into-vseq}
  Let %
\NoteState{thm:sim-vsubk-into-vseq}
  $\tm, \tmtwo \in \vsubkterms$.
  If $\deriv \colon \tm \tovsubk^* \tmtwo$ then there is $\derivtwo \colon \tolbarmu\tm \tovseq^* \tolbarmu\tmtwo$ such that ($\sizevseq\derivtwo$ denotes the length of $\derivtwo$)
  
  \begin{varenumerate}
    	\item\label{pappendix:sim-vsubk-into-vseq-mult} \emph{Multiplicative Steps}: $\sizem{\deriv} = \sizehole{\lambdabar}{\derivtwo}$ (the number $\lambdabar$-steps in $\derivtwo$);
	\item\label{pappendix:sim-vsubk-into-vseq-exp} \emph{Exponential Steps}: $\sizee\deriv = \sizehole{\mutildesym}\derivtwo$ (the number $\mut$-steps in $\derivtwo$), so $\sizevsubk\deriv = \sizevseq\derivtwo$;
    \item\label{pappendix:sim-vsubk-into-vseq-normal} \emph{Normal Form}: if $\tmtwo$ is $\vsubk$-normal then $\tolbarmu\tmtwo$ is $\vseq$-normal.
  \end{varenumerate}
\end{theoremAppendix}

\begin{proof}
The existence of $\derivtwo$ and the first two points are immediate consequences of \reflemma{lbarmut-simulates-vsubk}. We prove \refpoint{sim-vsubk-into-vseq-normal} by proving that the translation of a clean normal form $\tm$ of $\vsubkcalc$ is normal. Cases of $\tmtwo$:
\begin{itemize}
\item \emph{Value}: then clearly $\tolbarmu\tmtwo = \comm{\tolbarmuv\tmtwo}\stempty$ is normal. 
\item \emph{Compound Inert Term}: then $\tmtwo$ has the form $\tmtwo = \var \val_1 \ldots \val_k$. A straightforward induction on $k$ shows that it translates to $\comm\var{\stacker{\stacker{\stacker{\tolbarmuv{\val_1}}\ldots} {\tolbarmuv{\val_k}}} \stempty}$, that is normal. 
\item \emph{Substitution}: then $\tmtwo$ has the form $\tmtwo = \tmthree\esub\var\gconst$ where $\tmthree$ is a clean normal form and $\gconst$ is a compound inert term. If $\gconst = \vartwo \val_1 \ldots \val_k$ then $\tolbarmu\gconst = \comm\vartwo{\stacker{\stacker{\stacker{\tolbarmuv{\val_1}}\ldots} {\tolbarmuv{\val_k}}} \stempty}$ and $\tolbarmu{\tmthree\esub\var\gconst} = \append{\comm\vartwo{\stacker{\stacker{\stacker{\tolbarmuv{\val_1}}\ldots} {\tolbarmuv{\val_k}}} \stempty}} {\mutilde\var\tolbarmu\tmthree} = \comm\vartwo{\stacker{\stacker{\stacker{\tolbarmuv{\val_1}}\ldots} {\tolbarmuv{\val_k}}} \mutilde\var\tolbarmu\tmthree}$, that is normal because by \ih $\tolbarmu\tmthree$ is normal.
\qedhere
\end{itemize}
\end{proof}

\setcounter{corollaryAppendix}{\value{coro:equivalence-vsubk-vseq-termination}}
\begin{corollaryAppendix}[Linear Termination Equivalence of $\vsubkcalc$ and $\vseqcalc$\!]
\label{coroappendix:equivalence-vsubk-vseq-termination}
  Let
\NoteState{coro:equivalence-vsubk-vseq-termination}
  $\tm \in \vsubkterms$. 
  There is a $\vsubk$-normalizing derivation $\deriv$ from $\tm$ iff there is a $\vseqsym$-norm\-alizing derivation $\derivtwo$ from $\tolbarmu\tm$. 
  Moreover, \mbox{$\sizevsubk\deriv \!= \sizevseq\derivtwo$, $\sizee\deriv = \sizemut\derivtwo$ and $\sizem\deriv = \sizelbar\derivtwo$.}
\end{corollaryAppendix}

\begin{proof}
  \begin{description}
       \item[$\Rightarrow$:] Let $\deriv \colon \tm \tovsubk^* \tmtwo$ be a $\vsubk$-normalizing derivation and $\derivtwo \colon \tolbarmu\tm \tovseq^* \tolbarmu\tmtwo$ be its projection in $\vseqcalc$, according to \refthm{sim-vsubk-into-vseq}. 
       Then $\derivtwo$ is a $\vsubk$-normalizing derivation from $\tolbarmu\tm$, since the $\vsubk$-normality of $\tmtwo$ implies the $\vseq$-normality of $\tolbarmu\tmtwo$ by \refthmp{sim-vsubk-into-vseq}{normal}.
       
       \item[$\Leftarrow$:] By contradiction, suppose that there is a diverging $\vsubk$-derivation from $\tm$ in $\vsubkcalc$. By \refthm{sim-vsubk-into-vseq} it projects to a $\vseq$-derivation from $\tolbarmu\tm$ in $\vseqcalc$ that is at least as long as the one in $\vsubkcalc$, which is absurd since $\tolbarmu\tm$ is $\vseq$-normalizable and all $\vseq$-normalizing derivations from $\tolbarmu\tm$ have the same length by \refpropp{basic-lambdamu}{strong-confluence}.
%       Then, there is a infinite reduction sequence $\deriv \defeq (\tm_i)_{i \in \nat}$ such that $\tm_0 = \tm$ and $\tm_i \tof \tm_{i+1}$ for any $i \in \nat$.
%       For every $n \in \nat$, let $\deriv_n$ be the sub-reduction sequence of $\deriv$ starting from $\tm_0 = \tm$ and such that $\sizef{\deriv_n} = n$:
%       since $\tof$ can be simulated by $\tovsub$ in such a way that the number of $\tof$-steps are equal to the number of $\tom$-steps (\refpropsps{sim-f-into-vsubeq}{qual}{quant-mult}) and $\unf{\tm} = \tm$ (as $\tm \in \Lambda$), there are $\tmtwo_n \in \vsubterms$ and a reduction sequence $\derivtwo_n \colon \tm \tovsub^* \tmtwo_n$ such that $\sizem{\derivtwo_n} = n$ and hence $\sizevsub{\derivtwo_n} \geq n$.
%       Therefore, $\tm$ is not $\vsub$-normalizable, absurd.
  \end{description}
  
  The result about lengths follows immediately from \refthmsps{sim-vsubk-into-vseq}{mult}{exp}.
  \qedhere
\end{proof}

\paragraph{Structural equivalence for $\vseqcalc$.}

\begin{remark}
  \label{rmk:cotctx}
  Every environment evaluation context can be uniquely written as $\cotctx = \comm{\val}{\stacker{\stacker{\stacker{\val_1}{\dots}}{\val_n}}{\cotm}}$ where either $\cotm = \stempty$ or $\cotm = \mutilde\var\cm$.
\end{remark}

% \setcounter{lemmaAppendix}{\value{l:seq-to-seqbar}}
% \begin{lemmaAppendix}
% \label{lappendix:seq-to-seqbar} % \reflemmap{seq-to-seqbar}{seq}
%  Let 
% \NoteState{l:seq-to-seqbar} 
%  $\tm$ and $\tmtwo$ be $\vsubk$-terms% and $\cotm$ be a co-term
%  . 
%  \begin{enumerate}
%   \item \label{pappendix:seq-to-seqbar-eq} If $\tm \eqstruct_{\essym,\aplsym} \tmp$ then $\tolbarmu\tm = \tolbarmu\tmp$.
%   \item \label{pappendix:seq-to-seqbar-seq} If $\tm \eqstruct_{\comsym} \tmp$ then $\tolbarmu\tm \seqbar \tolbarmu\tmp$.
%  \end{enumerate}
% \end{lemmaAppendix}
\begin{lemma}
\label{l:seq-to-seqbar} % \reflemmap{seq-to-seqbar}{seq}
 Let 
 $\tm$ and $\tmtwo$ be $\vsubk$-terms% and $\cotm$ be a co-term
 . 
 \begin{enumerate}
  \item \label{p:seq-to-seqbar-eq} If $\tm \eqstruct_{\essym,\aplsym} \tmp$ then $\tolbarmu\tm = \tolbarmu\tmp$.
  \item \label{p:seq-to-seqbar-seq} If $\tm \eqstruct_{\comsym} \tmp$ then $\tolbarmu\tm \seqbar_{\mut\mut} \tolbarmu\tmp$.
 \end{enumerate}
\end{lemma}

\begin{proof}
 \begin{enumerate}
  \item If $\tm \eqstruct_{\essym} \tmp$, then $\tm = \tmthree\esub\var{\tmfour\esub{\vartwo}{\tmtwo}}$ and $\tmp = \tmthree\esub\var\tmfour\esub{\vartwo}{\tmtwo}$.
  So, just apply the translation and \reflemmap{append-commutes}{ev-ctxs} (setting $\cmctx = \append{\tolbarmu{\tmtwo}}{\mutilde\vartwo\ctxhole}$, $\cm = \tolbarmu{\tmfour}$ and $\cotm = \mutilde\var{\tolbarmu{\tmthree}}$):
  \begin{align*}
    \tolbarmu{\tm} = \tolbarmu{\tmthree\esub\var{\tmfour\esub{\vartwo}{\tmtwo}}} = \append{(\append{\tolbarmu{\tmtwo}}{\mutilde\vartwo\tolbarmu{\tmfour}})}{\mutilde\var\tolbarmu\tmthree} = \append{\tolbarmu{\tmtwo}}{\mutilde\vartwo{(\append{\tolbarmu{\tmfour}}{\mutilde\var\tolbarmu{\tmthree}})}} = \tolbarmu{\tmthree \esub\var\tmfour \esub\vartwo\tmtwo} = \tolbarmu{\tmp}.
  \end{align*}

  If $\tm \eqstruct_{\aplsym} \tmp$, then $\tm = \tmthree\esub\var{\tmtwo}\val$ and $\tmp = (\tmthree\val)\esub\var{\tmtwo}$.
  So, just apply the translation and \reflemmap{append-commutes}{ev-ctxs} (setting $\cmctx = \append{\tolbarmu{\tmtwo}}{\mutilde\var\ctxhole}$, $\cm = \tolbarmu{\tmthree}$ and $\cotm = \stacker{\tolbarmuv{\val}}\stempty$):
  \begin{align*}
    \tolbarmu{\tm} = \tolbarmu{\tmthree\esub\var{\tmtwo}\val} = \append{(\append{\tolbarmu{\tmtwo}}{\mutilde\var\tolbarmu{\tmthree}})}{(\stacker{\tolbarmuv{\val}}\stempty)} = \append{\tolbarmu{\tmtwo}}{\mutilde\var{(\append{\tolbarmu{\tmthree}}{(\stacker{\tolbarmuv{\val}}\stempty))}}} = \tolbarmu{(\tmthree\val) \esub\var\tmtwo} = \tolbarmu{\tmp}.
  \end{align*}

  \item As $\tm \eqstruct_{\comsym} \tmp$, then $\tm = \tmthree\esub{\vartwo}{\tmfour}\esub{\var}{\tmtwo}$ and $\tmp = \tmthree\esub{\var}{\tmtwo}\esub{\vartwo}{\tmfour}$ with $\var \notin \fv{\tmthree}$ and $\vartwo \notin \fv{\tmfour}$.
  So, just apply the translation and the definition of $\seqbar_{\mut\mut}$, the only axiom generating $\seqbar$ (setting $\cotctx = \append{\tolbarmu{\tmtwo}}{\ctxhole}$ and $\cotctxtwo = \append{\tolbarmu{\tmfour}}{\ctxhole}$):
  \begin{align*}
    \tolbarmu{\tm} = \tolbarmu{\tmthree\esub{\vartwo}{\tmfour}\esub{\var}{\tmtwo}} = \append{\tmtwo}{\mutilde\var(\append{\tolbarmu\tmfour}{\mutilde\vartwo{\tolbarmu\tmthree}})}
    \seqbar_{\mut\mut}
    \append{\tmfour}{\mutilde\vartwo(\append{\tolbarmu\tmtwo}{\mutilde\var{\tolbarmu\tmthree}})} = \tolbarmu{\tmthree\esub{\var}{\tmtwo}\esub{\vartwo}{\tmfour}} = \tolbarmu{\tmp}.
  \end{align*}
  \qedhere

 \end{enumerate}
\end{proof}

\begin{proposition}[Simulation of $\eqstruct$ by $\seqbar$]
\label{prop:seq-to-seqbar} % \reflemmap{seq-to-seqbar}{seq}
 Let 
 $\tm$ and $\tmp$ be $\vsubk$-terms. If $\tm \eqstruct \tmp$ then $\tolbarmu\tm \seqbar \tolbarmu\tmp$.
\end{proposition}

\begin{proof}
  First, observe that there are no $\tmtwo, \tmtwop \in \vsubkterms$ such that $\tmtwo \eqstruct_{\aprsym} \tmtwop$: indeed, $\tmtwo \eqstruct_{\aprsym} \tmtwop$ implies that $\tmtwo = \tmthree \, \tmfour\esub\var\tmfive$ and $\tmfour\esub\var\tmfive$ is not a value, therefore $\tmtwo \notin \vsubkterms$.
  
  Let $\eqstruct'$ the closure of $\eqstruct_{\aplsym} \cup \eqstruct_{\comsym} \cup \eqstruct_{\essym}$ under evaluation contexts of $\vsubkcalc$. 
  As $\eqstruct$ on $\vsubkterms$ is just the reflexive-transitive and symmetric closure of $\eqstruct$ (and $\seqbar$ is an equivalence relation), in order to prove \refprop{seq-to-seqbar} it is enough to prove that the following statement $(*)$: for every $\tm, \tmp \in \vsubkterms$, if $\tm \eqstruct' \tmp$ then $\tolbarmu\tm \seqbar \tolbarmu\tmp$.
  The proof of $(*)$ is by induction on the definition $\tm \eqstruct' \tmp$.
  
  The base cases (\ie when $\tm \eqstruct_{\aplsym} \tmp$ or $\tm \eqstruct_{\comsym} \tmp$ or $\tm \eqstruct_{\essym} \tmp$) are already proved in \reflemma{seq-to-seqbar}.
  Concerning the inductive cases, we have:
  \begin{itemize}
    \item \emph{Application Left}, \ie $\tm \defeq \tmtwo \val \eqstruct \tmtwop \val \eqdef \tmp$ with $\tmtwo \eqstruct \tmtwop$: by \ih, $\tolbarmu{\tmtwo} \seqbar \tolbarmu{\tmtwop} $; so, $\tolbarmu{\tm} = \append{\tolbarmu\tmtwo}{(\stacker{\tolbarmuv\val}\stempty)} \seqbar \append{\tolbarmu\tmtwop}{(\stacker{\tolbarmuv\val}\stempty)} = \tolbarmu{\tmp}$;
    \item \emph{Left of a Substitution}, \ie $\tm \defeq \tmtwo \esub\var\tmthree \eqstruct \tmtwop \esub\var \tmthree \eqdef \tmp$ with $\tmtwo \eqstruct \tmtwop$:
    by \ih, $\tolbarmu{\tmtwo} \seqbar \tolbarmu{\tmtwop} $; so, $\tolbarmu{\tm} = \append{\tolbarmu\tmthree}{\mutilde\var{\tolbarmu\tmtwo}} \seqbar \append{\tolbarmu\tmthree}{\mutilde\var{\tolbarmu\tmtwop}} = \tolbarmu{\tmp}$;
    \item \emph{Inside a Substitution}, \ie $\tm \defeq \tmthree \esub\var\tmtwo \eqstruct \tmthree\esub\var\tmtwop \eqdef \tmp$ with $\tmtwo \eqstruct \tmtwop$:
    by \ih, $\tolbarmu{\tmtwo} \seqbar \tolbarmu{\tmtwop} $; thus, $\tolbarmu{\tm} = \append{\tolbarmu\tmtwo}{\mutilde\var{\tolbarmu\tmthree}} \seqbar \append{\tolbarmu\tmtwop}{\mutilde\var{\tolbarmu\tmthree}} = \tolbarmu{\tmp}$.
    \qedhere
  \end{itemize}
\end{proof}

\begin{proposition}[Basic Properties of Structural Equivalence $\seqbar$]
\label{prop:seqbar-post-and-term} % \reflemmap{eqstruct-post-and-term}{strong-confluence}
  Let 
  $\cm_0, \cm_1$ be commands and $\Rule \in\set{\lambdabar, \mut, \vseq}$.
    \begin{varenumerate}
      \item\label{p:seqbar-post-and-term-locpost} \emph{Strong Bisimulation of $\seqbar$ wrt $\tovseq$}: 
      if $\cm_0 \seqbar \cm_1$ and $\cm_1 \torule \cm_2$ then there exists a command $\cm_3$ such that $\cm_0 \torule \cm_3 \seqbar \cm_2$.      
      \item\label{p:seqbar-post-and-term-globpost} \emph{Postponement of $\seqbar$ wrt $\tovseq$}: 
% 	$\tovsubeq^k \, \subseteq \, \tovsub^k\eqstruct$ preserving the kind and number of steps.
      if $\deriv \colon \cm_0 \tovseqeq^* \cm_1$ then there are $\cm_2 \seqbar \cm_1$ and $\derivtwo \colon \cm_0 \tovseq^* \cm_2$ such that $\sizevseq{\deriv} = \sizevseq{\derivtwo}$, $\sizemut{\deriv} = \sizemut{\derivtwo}$ and $\sizelbar{\deriv} = \sizelbar{\derivtwo}$.
      \item\label{p:seqbar-post-and-term-normal}\emph{Normal Forms}: if $\tm \seqbar \tmtwo$ then $\tm$ is $\Rule$-normal iff $\tmtwo$ is $\Rule$-normal.
% % % 	\item \emph{Simulation and Termination Equivalence of $\tovsub$ and $\tovsubeq$}: \label{p:eqstruct-post-and-term-term}
% % % % 	let $\deriv \colon \tm\tovsub^*\tmtwo$ be a normalising derivation.
% % % % 	Then there exists a normalising derivation $\derivtwo \colon \tm\tovsubeq^*\tmthree$ with $\sizee\derivtwo = \sizee\deriv$.
% % % 	If $\deriv \colon \tm\tovsubeq^*\tmtwo$ then there exists $\derivtwo \colon \tm\tovsub^*\tmthree$ with $\sizee\derivtwo = \sizee\deriv$.
% % % 	If moreover $\deriv$ is $\vsub$-normalizing then $\derivtwo$ is $\vsub$-normalizing.
      \item\label{p:seqbar-post-and-term-strong-confluence} \emph{Strong confluence}: $\tovseqeq$ is strongly confluent.
    \end{varenumerate}	
\end{proposition}

\begin{proof}
  \begin{enumerate}
    \item It is enough to prove the following statement ($*$): if $\cm_0 \seqbarp \cm_1$ and $\cm_1 \torule \cm_2$ then there exists a command $\cm_3$ such that $\cm_0 \torule \cm_3 \seqbarp \cm_2$, where $\seqbarp$ is the reflexive closure under command evaluation contexts of $\seqbarmut$, the unique axiom generating the equivalence $\seqbar$.
    Indeed $\seqbarp$ is reflexive and symmetric, therefore $\seqbar$ is just the transitive closure of $\seqbarp$, so the proof of \refpropp{seqbar-post-and-term}{locpost} follows immediately.
    
    The proof of ($*$) is by induction on the definition of $\seqbarp$.

    In the inductive cases the proof follows immediately from the \ih, since $\seqbarp$ and $\torule$ are closed under the same contexts. 
    
    Concerning the base cases, according to \refrmk{cotctx}, we have 
      \begin{align*}
        \cm_0 &\defeq \comm{\val}{\stacker{\stacker{\stacker{\val_1}{\dots}}{\val_n}}{\mutilde\var{\comm{\valtwo}{\stacker{\stacker{\stacker{\valtwo_1}{\dots}}{\valtwo_{n'}}}{\mutilde\vartwo\cm}}} }} \\
        &\seqbarmut \comm{\valtwo}{\stacker{\stacker{\stacker{\valtwo_1}{\dots}}{\valtwo_{n'}}}{\mutilde\vartwo{\comm{\val}{\stacker{\stacker{\stacker{\val_1}{\dots}}{\val_n}}{\mutilde\var\cm}}} }} \eqdef \cm_1
      \end{align*}
      where $\var \notin \fv{\valtwo} \cup \bigcup_{i'=1}^{n'}\fv{\valtwo_{i'}}$ and $\vartwo \notin \fv{\val} \cup \bigcup_{i=1}^n \fv{\val_i}$.
      Thus there are only four cases:
      \begin{enumerate}
        \item \emph{Internal $\lambdabar$-step}, \ie $\val = \la\varthree\cmp$, $n > 0$ and 
        \begin{equation*}
          \cm_1 \tobvmu \comm{\valtwo}{\stacker{\stacker{\stacker{\valtwo_1}{\dots}}{\valtwo_{n'}}}{\mutilde\vartwo{\comm{\val_1}{\append{(\mutilde\varthree\cmp)} (\stacker{\stacker{\stacker{\val_2}{\dots}}{\val_n}}{\mutilde\var\cm})}} }} = \cm_2
        \end{equation*}
        then, $\cm_0 \tobvmu \comm{\val_1}{\append{(\mutilde\varthree\cmp)}{(\stacker{\val_2}{\stacker{\dots}{\stacker{\val_n}{\mutilde\var{\comm{\valtwo}{\stacker{\valtwo_1}\stacker\dots\stacker{\valtwo_{n'}}{\mutilde\vartwo\cm}}}}}})}} \seqbarp \cm_2$, where the last equivalence holds by applying the axiom $\seqbarmut$ with the environment evaluation contexts $\cotctx = \comm{\val_1}{\append{(\mutilde\varthree\cmp)}{\stacker{\val_2}{\stacker\dots{\stacker{\val_n}{\ctxhole}}}}}$ and $\cotctxtwo = \comm{\valtwo}{\stacker{\valtwo_1}{\stacker{\dots}{\stacker{\valtwo_{n'}}{\ctxhole}}}}$.
        \item \emph{External $\lambdabar$-step}, \ie $\valtwo = \la\varthree\cmp$, $n' > 0$ and 
        \begin{equation*}
          \cm_1 \tobvmu \comm{\valtwo_1}{\append{(\mutilde\varthree\cmp)}{(\stacker{\valtwo_2}{\stacker{\dots}{\stacker{\valtwo_{n'}}{\mutilde\vartwo{\comm{\val}{\stacker{\val_1}\stacker\dots\stacker{\val_{n}}{\mutilde\var\cm}}}}}})}} = \cm_2
        \end{equation*}
        then, $\cm_0 \tobvmu \comm{\val}{\stacker{\stacker{\stacker{\val_1}{\dots}}{\val_n}}{\mutilde\var{\comm{\valtwo_1}{\append{(\mutilde\varthree\cmp)} (\stacker{\stacker{\stacker{\valtwo_2}{\dots}}{\valtwo_n}}{\mutilde\vartwo\cm})}} }}  \seqbarp \cm_2$, where the last equivalence holds by applying the axiom $\seqbarmut$ with the environment evaluation contexts $\cotctx = \comm{\val}{\stacker{\val_1}{\stacker{\dots}{\stacker{\val_n}{\ctxhole}}}}$ and $\cotctxtwo = \comm{\valtwo_1}{\append{(\mutilde\varthree\cmp)}{\stacker{\valtwo_2}{\stacker\dots{\stacker{\valtwo_{n'}}{\ctxhole}}}}}$.
        
        \item \emph{Internal $\mut$-step}, \ie $n = 0$ and $\cm_1 %= \comm{\valtwo}{\stacker{\stacker{\stacker{\valtwo_1}{\dots}}{\valtwo_{n'}}}{\mutilde\vartwo{\comm{\val}{\mutilde\var\cm}} }} 
        \tomut \comm{\valtwo}{\stacker{\stacker{\stacker{\valtwo_1}{\dots}}{\valtwo_{n'}}}{\mutilde\vartwo{\cm\isub\var\val} }} = \cm_2$:
        then, 
        $\cm_0 \tomut \cm_2$ since $\var  \notin \fv{\valtwo} \cup \bigcup_{i'=1}^{n'}\fv{\valtwo_{i'}}$.
        \item \emph{External $\mut$-step}, \ie $n' = 0$ and $\cm_1 \tomut \comm{\val}{\stacker{\stacker{\stacker{\val_1}{\dots}}{\val_n}}{\mutilde\var{\cm\isub\vartwo\valtwo} }} = \cm_2$ (recall that $\vartwo \notin \fv{\val} \cup \bigcup_{i=1}^n \fv{\val_i}$):
        then, 
        $\cm_0 \tomut \cm_2$.
      \end{enumerate}
      
    \item Immediate consequence of \refpropp{seqbar-post-and-term}{locpost}.
    \item Immediate consequence of \refpropp{seqbar-post-and-term}{locpost}.
    \item Immediate consequence of \refpropp{seqbar-post-and-term}{locpost}.
    \qedhere
  \end{enumerate}

\end{proof}

}

\end{document}